\documentclass[journal, 10pt,onecolumn]{IEEEtran}

\ifCLASSINFOpdf

\else

\fi

\usepackage{tikz}
\usepackage{tkz-graph}
\usepackage{pgfplots}
\pgfplotsset{compat=1.18}
\usepgfplotslibrary{statistics,fillbetween,groupplots}
\usetikzlibrary[positioning,patterns,shapes,backgrounds,patterns,calc,arrows,arrows.meta,decorations.pathreplacing, positioning] 

\usepackage{ifthen}

\definecolor{accentG}{HTML}{5DCAA5}   
\definecolor{accentB}{HTML}{85B7EB}   % myopic
\definecolor{accentP}{HTML}{AFA9EC}   % auxiliary
\definecolor{accentR}{HTML}{F0997B}   % omniscient
\definecolor{boxfill}{HTML}{F8F7F5}   % box background
\definecolor{arrcol}{HTML}{C8C6BE}    % arrow colour
\definecolor{capred}{HTML}{A32D2D}    % negative-capacity colour
\definecolor{capgrn}{HTML}{3B6D11}    % positive-capacity colour
\definecolor{bndcol}{HTML}{BA7517}    % boundary colour

\usepackage{amsmath, amsfonts, amsthm, amssymb, wasysym, txfonts, mathrsfs}
\usepackage{pdfpages, nicefrac, graphicx, bbm, enumerate}
\allowdisplaybreaks

\theoremstyle{plain}
\newtheorem{thm}{Theorem}[section]
\newtheorem{co}[thm]{Corollary}
\newtheorem{lem}[thm]{Lemma}

\theoremstyle{definition}
\newtheorem{definition}[thm]{Definition}
\newtheorem{example}[thm]{Example}

\makeatletter
\newenvironment{exmp}[1][\@nil]{%
  \def\@tempa{#1}\def\@tempb{\@nil}%
  \ifx\@tempa\@tempb
    \begin{example}%
  \else
    \begin{example}[#1]%
  \fi}{\end{example}}
\makeatother

\theoremstyle{remark}
\newtheorem{remark}[thm]{Remark}

\newenvironment{proofoutline}
 {\proof[Proof outline]}
 {\endproof}

\DeclareMathOperator{\markov}{\setlength{\unitlength}{.5cm} \begin{picture}(1,1)  \put(0,.22){\line(1,0){1}}  \put(.5,.22){\circle{.3}}   \end{picture}}

\newcommand{\C}{\mathcal{C}}
\newcommand{\F}{\mathbb{F}}

\newcommand{\bfs}{\mathbf{s}}

\newcommand{\x}{\mathbf{x}}
\newcommand{\X}{\mathbf{X}}
\newcommand{\y}{\mathbf{y}}

\newcommand{\z}{\mathbf{z}}
\newcommand{\Z}{\mathbf{Z}}
\newcommand{\bfS}{\mathbf{S}}

\newcommand{\ie}{\emph{i.e.}}
\newcommand{\cA}{\mathscr{A}}
\newcommand{\cB}{\mathscr{B}}

\newcommand{\cK}{\mathscr{K}}

\newcommand{\cP}{\mathscr{P}}

\newcommand{\cS}{\mathscr{S}}

\newcommand{\cX}{\mathcal{X}}
\newcommand{\cY}{\mathcal{Y}}
\newcommand{\cZ}{\mathcal{Z}}

\newcommand{\Prob}{P}
\newcommand{\ProbJ}{J}

\newcommand{\szero}{s_0}

\newcommand{\wy}{\ensuremath{{{W}}}}
\newcommand{\uz}{\ensuremath{{U}}}
\newcommand{\authenc}{P_{X^n|M}}
\newcommand{\authdec}{\phi}
\newcommand{\accept}{{\ensuremath{1}}}
\newcommand{\reject}{{\ensuremath{0}}}
\newcommand{\fP}{\Gamma}
\newcommand{\distribution}{\uz-distribution }
\newcommand{\cst}{C_{\mathrm{auth}}}
\newcommand{\cdet}{C_{\mathrm{det,auth}}}
\newcommand{\ctag}{C_{\mathrm{tag}}}

\newcommand{\set}[1]{\left\{#1\right\}}
\newcommand{\expect}[2]{{\normalfont\mathbb{E}_{#1}\left[#2\right]}}
\usepackage{authblk}

\newcommand{\Perr}{P_{\mathrm{e}}}

\newcommand{\tc}{\mathsf{T}}
\newcommand{\tn}{t}
\newcommand{\Probv}{V}

\newboolean{arxiv}
\setboolean{arxiv}{true}

\newcommand{\eqbr}[1][]{#1}
\usepackage{setspace}
\onecolumn
\usepackage[colorlinks=true, linkcolor=blue, citecolor=blue, urlcolor=blue]{hyperref}

\begin{document}

\title{Authentication Against a Myopic Adversary}
\author{Mayank~Bakshi,
        Allison~Beemer,
        Oliver~Kosut,
        Eric~Graves,
        J\"{o}rg~Kliewer,
        and~Paul~Yu%
\thanks{M. Bakshi is with the School of Informatics, Computing, and Cyber Systems, Northern Arizona University, Flagstaff, AZ 86011 USA (e-mail: mayank.bakshi@nau.edu). A. Beemer is with the Department of Mathematics, University of Wisconsin-Eau Claire, Eau Claire, WI 54701 USA (e-mail: beemera@uwec.edu). O. Kosut is with the School of Electrical, Computer and Energy Engineering, Arizona State University, Tempe, AZ 85287 USA (e-mail: okosut@asu.edu). E. Graves and P. Yu are with the DEVCOM Army Research Laboratory, Adelphi, MD 20783 USA (e-mail: eric.g.graves.civ@army.mil; paul.l.yu.civ@army.mil). J. Kliewer is with the Department of Electrical and Computer Engineering, New Jersey Institute of Technology, Newark, NJ 07102 USA (e-mail: jkliewer@njit.edu).}
\thanks{The work of M. Bakshi is based upon work supported by the National Science Foundation under Grant No. CCF-2107526 and was also funded under the State of Arizona Technology and Research Initiative Fund (TRIF), administered by the Arizona Board of Regents. The work of A. Beemer was supported by the National Science Foundation under Grant No. CCF-2107488. The work of O. Kosut was supported by the National Science Foundation under Grants CCF-1908725 and CCF-2107526. The work of E. Graves and P. Yu was sponsored by the Combat Capabilities Development Command Army Research Laboratory and was accomplished under Cooperative Agreement Number W911NF-17-2-0183. The views and conclusions contained in this document are those of the authors and should not be interpreted as representing the official policies, either expressed or implied, of the Combat Capabilities Development Command Army Research Laboratory or the U.S. Government. The U.S. Government is authorized to reproduce and distribute reprints for Government purposes notwithstanding any copyright notation hereon.}\thanks{This work was presented in part at the 2019 IEEE Conference on Communications and Network Security (CNS), the 2020 IEEE International Symposium on Information Theory, and the 2023  International Symposium on Information Theory appearing in \cite{BeemerKKGY:2019a}, \cite{BeemerGKKY:2020a}, and~\cite{BakK:2023}}}

\markboth{}
{}

\maketitle

\begin{abstract}
We consider keyless authentication for point-to-point communication in the presence of a myopic adversary. In particular, the adversary has access to a non-causal noisy version of the transmission and uses this knowledge to choose the state of an arbitrarily-varying channel between legitimate users. The receiver succeeds by either decoding accurately or correctly detecting adversarial interference.

We introduce a single-letter channel condition called $U$-distribution overwritability that allows an adversary that knows the channel input distribution to overwrite the encoder's intended channel output distribution with an output distribution that corresponds to any channel input of the adversary's choice. We show that as long as the channel does not satisfy this condition, the authentication capacity with stochastic coding is positive and equals the no-adversary channel capacity. Our encoding scheme relies on authentication tags that allow verifying the correctness of a message that has already been received. We provide multi-letter extensions of $U$-distribution overwritability and show that the infinite letter limit yields exact characterization for positivity of authentication capacity. Noting that this characterization may, in general, be uncomputable, we next provide a single-letter channel condition called $U$-overwritability and show that it is a sufficient condition for zero authentication capacity. When restricted to deterministic codes, we show that when the channel is not $I$-overwritable, the authentication capacity is strictly positive. Further, when the channel to the legitimate receiver is degraded with respect to the channel to the adversary, $I$-overwritability is also a sufficient condition for the deterministic authentication capacity to equal zero.

Lastly, we give examples demonstrating that the various notions of overwritability presented in this paper differ from each other. Our results also show that stochastic encoders are necessary for positive authentication capacity in some cases. We examine in detail a binary adversarial channel that illustrates this necessity. Namely, for this channel, we show that when the adversarial channel is degraded with respect to the main channel, the no-adversary capacity of the underlying channel is achievable with a deterministic encoder. Otherwise, provided the channel to the adversary is not perfect, a stochastic encoder is necessary for positive authentication capacity; if such an encoder is allowed, the no-adversary capacity is again achievable.
\end{abstract}

\IEEEpeerreviewmaketitle
\section{Introduction}
When communicating over unsecured channels, verifying the trustworthiness of a received signal is critical. Message \emph{authentication} involves either decoding at the receiving end of a communication channel with the assurance that the message was not (effectively) tampered with by a malicious adversary, or raising an alarm indicating presence of an adversary.

In this paper, we consider the authentication problem over an arbitrarily-varying channel (AVC) with an adversary who has some noisy version of the transmitted sequence. 

An AVC takes as inputs both the legitimate and an adversarial transmission \cite{BlackwellBT:1960}, where the adversarial transmission (called a \textit{state}) may be maliciously chosen with the goal of causing a decoding error at the receiver. 
A plethora of variations on the AVC appear in the literature, in which the adversary has varying degrees of power and knowledge of the legitimate transmission, and the transmitter and receiver may or may not have access to shared secret information. The assumption that the adversary receives a noisy version of the transmitted signal is known as the \emph{myopic AVC}. In the case of authentication, we say that the decoder is successful if one of the following scenarios occurs: (1) the adversary is not active, and the decoder recovers the correct message, or (2) the adversary is active and the decoder either recovers the correct message or declares adversarial interference. This is a relaxation of the classical AVC, in which the decoder is successful if and only if it correctly recovers the transmitted message.  Indeed, we find that authentication is ``easier'' in the sense that it is possible for the AVC capacity to be zero while the authentication capacity is positive (or even unchanged from the no-adversary setting).

\subsection*{Related Work}

The current work lies at the intersection of two broader areas of previous study: authentication over AVCs (with an adversary oblivious to the transmitted sequence), and traditional error correction over AVCs with myopic adversaries.  Lying squarely in the former category are \cite{KosutK:2016, KosutK:2018}. In \cite{KosutK:2016}, a channel condition called \textit{overwritability} is introduced, and it is shown that, with an oblivious adversary, authentication can be achieved with high probability for non-overwritable AVCs at positive (but asymptotically vanishing) rates. The results are extended in \cite{KosutK:2018}, where overwritability is shown to completely classify authentication capacity over an (oblivious adversary) AVC, with all non-overwritable channels having capacity equal to the non-adversarial capacity of the AVC. 

Meanwhile, more knowledgeable adversaries over the AVC in the error correction setting were studied in \cite{Sarwate:2010,CaiCG:2010,DeyJL:2019,ZhangVJS:2018}. The capacity of an AVC with a myopic adversary who does not have knowledge of the codebook realization (i.e., the legitimate users share common randomness) is characterized in \cite{Sarwate:2010}. The authors of \cite{CaiCG:2010} give the capacity of the AVC in the case where the legitimate users have a sufficiently large shared key, and the adversary can see the transmission sequence (i.e. the channel input) perfectly. In \cite{DeyJL:2019}, the authors give bounds on the capacity of an AVC in the presence of a power-constrained myopic adversary, showing that once the channel to the adversary is bad enough, it is essentially oblivious. In the context of Gaussian channels, \cite{ZhangVJS:2018} examines the case where both the transmitter and myopic adversary are quadratically power-constrained.

Prior work on authentication with a variety of more capable adversaries includes \cite{Simmons:1984,Maurer:1997,LaiEP:2009,GravesYS:2016,GungorK:2016,PerazzoneGYB:2018,SangwanBDP:2019a, GravesBKKY:2023}. Simmons \cite{Simmons:1984} considered authentication in the case where the adversary knows the encoding procedure and may intercept the transmission perfectly, but the legitimate users are allowed some shared key with which they randomize their choice of codebook. In \cite{Maurer:1997}, Maurer considers authentication for secret key agreement in the setting where each of the legitimate users and the adversary have access to a random variable, the three of which are correlated in some way. In that model, the adversary may substitute its transmission for that of the actual transmitter (over an otherwise noiseless channel), with the goal of hoodwinking the legitimate receiver; the (authentication) capacity is characterized by whether the joint distribution is \textit{simulatable} by the adversary. An inner bound on keyless authentication capacity is given by Graves, Yu, and Spasojevic in \cite{GravesYS:2016}, and later improved upon by Sangwan et al. in \cite{SangwanBDP:2019a}, for the case in which the adversary knows both the encoding procedure and the message, but does not have any additional knowledge of the transmitted sequence, which is allowed to be the output of a stochastic encoder. 
In particular, Sangwan et al. use a two-stage scheme: in one stage, the transmitter sends a short, randomly generated message (crucially, the adversary is not given knowledge of this message) using an authentication code in the style of Kosut and Kliewer~\cite{KosutK:2018}, while the other stage uses the previously-sent message as the key in an authentication scheme of the style of Simmons~\cite{Simmons:1984}.
In essence, this bootstrapping technique transforms the model proposed in Graves, Yu, and Spasojevic into the solved model of Kosut and Kliewer~\cite{KosutK:2018}. Nevertheless, it is not \emph{a priori} clear whether such a two-stage scheme suffices for all the cases considered in this paper, since here a myopic adversary may have the opportunity to decode the short randomly generated message and use this in their attack. 
Graves et al. \cite{GravesBKKY:2023} address authentication over an AWGN channel with an additive myopic adversary. They show that a modular coding strategy called overlay codes, which introduce cleverly-designed artificial noise to almost any code designed for a non-adversarial AWGN channel, can result in achievability of the non-adversarial code rate provided the adversary is not omniscient.
Lai, El Gamal, and Poor \cite{LaiEP:2009} consider authentication in the presence of a myopic adversary who carries out \textit{impersonation} or \textit{substitution} attacks; the transmitter and receiver are allowed a shared secret key. 
The attacks considered in the current work reduce to substitution attacks when the underlying AVC allows for direct modification of a transmission.
This model is extended in \cite{GungorK:2016}, where the authors additionally provide inner bounds for the keyed and keyless authentication capacities with impersonation/substitution attacks, and extend the simulatability condition of \cite{Maurer:1997} to characterize nonzero capacity for this set of attacks.
Finally, \cite{PerazzoneGYB:2018} improves on the inner bounds of \cite{LaiEP:2009,GungorK:2016} for keyed authentication.

\subsection*{Our Contributions}

\definecolor{capstoch}{HTML}{000000}   % stochastic capacity
\definecolor{capdet}{HTML}{6E6E6E}     % deterministic capacity
\definecolor{axislt}{HTML}{000000}     % rules and ticks
\definecolor{labellt}{HTML}{000000}
\definecolor{arrowgray}{HTML}{000000}  % proven implication
\definecolor{conjline}{HTML}{000000}   % open converse
\definecolor{sepline}{HTML}{000000}    % disproved converse
\definecolor{scaleln}{HTML}{D6D6D6}    % faint capacity track
\definecolor{lconn}{HTML}{8C8C8C}      % sub-case connector
\definecolor{rowln}{HTML}{D8D8D8}      % faint separator between rows

\begin{figure}[t]
\hypersetup{linkcolor=black, citecolor=black}
\centering
\resizebox{\linewidth}{!}{%
\begin{tikzpicture}[
    x=1pt, y=1pt,
    every node/.style={font=\footnotesize, inner sep=0pt},
    cond/.style={font=\normalsize, anchor=west},
    subcond/.style={font=\footnotesize, text=labellt, anchor=west},
    excl/.style={font=\footnotesize\itshape, text=labellt, anchor=west},
    subcase/.style={font=\footnotesize, text=labellt, anchor=west},
    reslbl/.style={font=\footnotesize, text=labellt, anchor=west},
    resbar/.style={font=\scriptsize, text=labellt, anchor=west},
    implL/.style={font=\footnotesize, text=labellt, anchor=east},
    implR/.style={font=\footnotesize, text=labellt, anchor=west},
    axislbl/.style={font=\footnotesize, text=labellt, anchor=mid},
    sectlbl/.style={font=\normalsize\itshape, text=labellt, anchor=west},
    helplbl/.style={font=\small, text=labellt, anchor=west},
    sdot/.style={circle, fill=capstoch, inner sep=0pt, minimum size=5pt},
    ddot/.style={rectangle, fill=capdet, inner sep=0pt, minimum size=5pt},
    qhol/.style={rectangle, fill=white, draw=capdet, inner sep=0pt,
                 minimum size=7.5pt, line width=0.7pt},
    ohol/.style={circle, fill=white, draw=capdet, inner sep=0pt,
                 minimum size=6.5pt, line width=0.9pt},
    qmk/.style={font=\fontsize{7}{8}\selectfont\bfseries},
    track/.style={color=scaleln, line width=0.7pt},
    pimpl/.style={-{Implies[length=5pt, width=5pt]},
                  double, double distance=1.2pt, line width=0.3pt},
    dshaft/.style={double, double distance=1.2pt, line width=0.3pt,
                   dash pattern=on 1.8pt off 1.1pt},
    dhead/.style={-{Implies[length=5pt, width=5pt]},
                  double, double distance=1.2pt, line width=0.3pt},
]

\def\xCond{10}        % channel condition (what the arrows relate)
\def\xLconn{22}
\def\xSub{38}
\def\xCiteL{192}      % result behind the proven implication (right-aligned)
\def\xArrL{200}       % proven implication lane (down)
\def\xArrR{225}       % converse lane (up)
\def\xCiteR{233}      % result behind the converse (left-aligned)
\def\xGrp{326}        % capacity group: extra assumption + bars + results
\def\xExcl{332}
\def\xL{476}          % capacity scale, 0
\def\xR{576}          % capacity scale, C_{s_0}
\def\xMid{526}
\def\xRes{588}        % result behind the capacity claim
\def\xRight{645}

\def\yLeg{136}
\def\yLegA{116}

\draw[axislt, line width=0.7pt] (\xCond,530) -- (\xRight,530);

\node[font=\normalsize\bfseries, text=black, anchor=west] at (\xCond,516)
    {Channel Condition};
\node[font=\normalsize\bfseries, text=black, anchor=mid] at (212,516)
    {Implications};
\node[font=\normalsize\bfseries, text=black, anchor=mid] at (485,516)
    {Authentication Capacity};

\draw[axislt, line width=0.4pt] (\xGrp,504) -- (\xRight,504);

\node[excl] at (\xExcl,492) {assuming also};
\draw[axislt, line width=0.4pt] (\xL,488) -- (\xL,496);
\draw[axislt, line width=0.4pt] (\xR,488) -- (\xR,496);
\node[axislbl] at (\xL,478) {$0$};
\node[axislbl] at (\xR,478) {$C_{s_0}$};

\draw[axislt, line width=0.5pt] (\xCond,470) -- (\xRight,470);

% ------------------------------------------------------------------
% Faint separators between rows, broken around the implications
% ------------------------------------------------------------------
\foreach \yy in {432,392,352,312,272,182}{%
  \draw[rowln, line width=0.3pt] (\xCond,\yy) -- (150,\yy);
  \draw[rowln, line width=0.3pt] (\xGrp,\yy) -- (\xRight,\yy);
}

% ------------------------------------------------------------------
% IMPLICATIONS between consecutive CONDITIONS.  Down = proven
% (stronger implies weaker); up = converse, struck when disproved.
% ------------------------------------------------------------------
% boundary 432
\draw[pimpl, color=arrowgray] (\xArrL,444) -- (\xArrL,420);
\node[implL] at (\xCiteL,432) {Cor.~\ref{cor:re}};
\draw[dshaft, color=sepline] (\xArrR,420) -- (\xArrR,436);
\draw[dhead,  color=sepline] (\xArrR,436) -- (\xArrR,444);
\draw[sepline, line width=1.0pt] (\xArrR-3.5,428.5) -- (\xArrR+3.5,435.5);
\node[implR] at (\xCiteR,432) {Ex.~\ref{exmp:bitflip_Uow}};

% boundary 392
\draw[pimpl, color=arrowgray] (\xArrL,404) -- (\xArrL,380);
\node[implL] at (\xCiteL,392) {Thm~\ref{thm:koverwritability}};
\draw[dshaft, color=conjline] (\xArrR,380) -- (\xArrR,396);
\draw[dhead,  color=conjline] (\xArrR,396) -- (\xArrR,404);
\node[implR] at (\xCiteR,392) {open (conjectured)};

% boundary 352
\draw[pimpl, color=arrowgray] (\xArrL,364) -- (\xArrL,340);
\node[implL] at (\xCiteL,352) {Def.~\ref{def:infinitedimow}};
\draw[dshaft, color=sepline] (\xArrR,340) -- (\xArrR,356);
\draw[dhead,  color=sepline] (\xArrR,356) -- (\xArrR,364);
\draw[sepline, line width=1.0pt] (\xArrR-3.5,348.5) -- (\xArrR+3.5,355.5);
\node[implR] at (\xCiteR,352) {Ex.~\ref{exmp:BSCZ}};

% boundary 312
\draw[pimpl, color=arrowgray] (\xArrL,324) -- (\xArrL,300);
\node[implL] at (\xCiteL,312) {Thm~\ref{thm:kdistoverwritability}};
\draw[dshaft, color=sepline] (\xArrR,300) -- (\xArrR,316);
\draw[dhead,  color=sepline] (\xArrR,316) -- (\xArrR,324);
\draw[sepline, line width=1.0pt] (\xArrR-3.5,308.5) -- (\xArrR+3.5,315.5);
\node[implR] at (\xCiteR,312) {Ex.~\ref{exmp:BSCZ}};

% boundary 272
\draw[pimpl, color=arrowgray] (\xArrL,284) -- (\xArrL,260);
\node[implL] at (\xCiteL,272) {Cor.~\ref{cor:implies_I-over}};
\draw[dshaft, color=sepline] (\xArrR,260) -- (\xArrR,276);
\draw[dhead,  color=sepline] (\xArrR,276) -- (\xArrR,284);
\draw[sepline, line width=1.0pt] (\xArrR-3.5,268.5) -- (\xArrR+3.5,275.5);
\node[implR] at (\xCiteR,272) {Ex.~\ref{exmp:bitflip_IvsUdistow}};

% ------------------------------------------------------------------
% ROW 1 : obliviously overwritable   (y = 452)
% ------------------------------------------------------------------
\def\yA{452}
\node[cond]    at (\xCond,\yA+6) {Oblivious overwritable};
\node[subcond] at (\xCond,\yA-7) {Definition~\ref{def:overwritability}};
\node[excl]    at (\xExcl,\yA)   {(strongest condition)};
\draw[track] (\xL,\yA+5) -- (\xR,\yA+5);
\draw[track] (\xL,\yA-5) -- (\xR,\yA-5);
\node[sdot] at (\xL,\yA+5) {};
\node[ddot] at (\xL,\yA-5) {};
\draw[decorate, decoration={brace, amplitude=3pt}, line width=0.4pt,
      color=labellt] (\xRes,\yA+5) -- (\xRes,\yA-5);
\node[reslbl] at (\xRes+9,\yA) {\cite{KosutK:2018}};

% ------------------------------------------------------------------
% ROW 2 : U-overwritable   (y = 412)
% ------------------------------------------------------------------
\def\yB{412}
\node[cond]    at (\xCond,\yB+6) {$U$-overwritable};
\node[subcond] at (\xCond,\yB-7) {Definition~\ref{def:U-overwritable}};
\node[excl]    at (\xExcl,\yB)   {not obliviously overwritable};
\draw[track] (\xL,\yB+5) -- (\xR,\yB+5);
\draw[track] (\xL,\yB-5) -- (\xR,\yB-5);
\node[sdot] at (\xL,\yB+5) {};
\node[ddot] at (\xL,\yB-5) {};
\draw[decorate, decoration={brace, amplitude=3pt}, line width=0.4pt,
      color=labellt] (\xRes,\yB+5) -- (\xRes,\yB-5);
\node[reslbl] at (\xRes+9,\yB) {Thm~\ref{thm:converse}};

% ------------------------------------------------------------------
% ROW 3 : U^inf-distribution overwritable   (y = 372)
% ------------------------------------------------------------------
\def\yC{372}
\node[cond]    at (\xCond,\yC+6) {$U^{\otimes\infty}$-distribution overwritable};
\node[subcond] at (\xCond,\yC-7) {Definition~\ref{def:infinitedimow}};
\node[excl]    at (\xExcl,\yC)   {not $U$-overwritable};
\draw[track] (\xL,\yC+5) -- (\xR,\yC+5);
\draw[track] (\xL,\yC-5) -- (\xR,\yC-5);
\node[sdot] at (\xL,\yC+5) {};
\node[ddot] at (\xL,\yC-5) {};
\draw[decorate, decoration={brace, amplitude=3pt}, line width=0.4pt,
      color=labellt] (\xRes,\yC+5) -- (\xRes,\yC-5);
\node[reslbl] at (\xRes+9,\yC) {Thm~\ref{thm:infinitedim}};

% ------------------------------------------------------------------
% ROW 4 : k-letter U-distribution overwritable   (y = 332)
% ------------------------------------------------------------------
\def\yD{332}
\node[cond]    at (\xCond,\yD+6) {$k$-letter $U$-dist.\ overwritable};
\node[subcond] at (\xCond,\yD-7) {Definition~\ref{def:kletterchannel} (finite $k>1$)};
\node[excl]    at (\xExcl,\yD)   {not $(k\!+\!1)$-letter};
\draw[track] (\xL,\yD+5) -- (\xR,\yD+5);
\draw[track] (\xL,\yD-5) -- (\xR,\yD-5);
\draw[capstoch, line width=0.9pt] (\xL,\yD+5) -- (\xR,\yD+5);
\node[sdot] at (\xR,\yD+5) {};
\draw[capdet, line width=0.7pt, dash pattern=on 3pt off 2pt]
    (\xL,\yD-5) -- (\xR,\yD-5);
\draw[capdet, line width=0.9pt] (\xL,\yD-2) -- (\xL,\yD-8);
\draw[capdet, line width=0.9pt] (\xR,\yD-2) -- (\xR,\yD-8);
\node[qhol] at (\xMid,\yD-5) {};
\node[qmk]  at (\xMid,\yD-5) {?};
\node[resbar] at (\xRes,\yD+5) {Thm~\ref{thm:infinitedim}};
\node[resbar] at (\xRes,\yD-5) {open};

% ------------------------------------------------------------------
% ROW 5 : U-distribution overwritable   (y = 292)
% ------------------------------------------------------------------
\def\yE{292}
\node[cond]    at (\xCond,\yE+6) {$U$-dist.\ overwritable};
\node[subcond] at (\xCond,\yE-7) {Definition~\ref{def:udistow}};
\node[excl]    at (\xExcl,\yE)   {not $2$-letter $U$-dist.\ overwritable};
\draw[track] (\xL,\yE+5) -- (\xR,\yE+5);
\draw[track] (\xL,\yE-5) -- (\xR,\yE-5);
\draw[capstoch, line width=0.9pt] (\xL,\yE+5) -- (\xR,\yE+5);
\node[sdot] at (\xR,\yE+5) {};
\draw[capdet, line width=0.7pt, dash pattern=on 3pt off 2pt]
    (\xL,\yE-5) -- (\xR,\yE-5);
\draw[capdet, line width=0.9pt] (\xL,\yE-2) -- (\xL,\yE-8);
\draw[capdet, line width=0.9pt] (\xR,\yE-2) -- (\xR,\yE-8);
\node[qhol] at (\xMid,\yE-5) {};
\node[qmk]  at (\xMid,\yE-5) {?};
\node[resbar] at (\xRes,\yE+5) {Thm~\ref{thm:infinitedim}};
\node[resbar] at (\xRes,\yE-5) {open};

% ------------------------------------------------------------------
% ROW 6 : I-overwritable   (parent y = 252, sub-cases 230 / 202)
% ------------------------------------------------------------------
\node[cond]    at (\xCond,258) {$I$-overwritable};
\node[subcond] at (\xCond,245) {Definition~\ref{definition:I-over}};
\node[excl]    at (\xExcl,252) {not $U$-dist.\ overwritable};
\draw[lconn, line width=0.8pt] (\xLconn,239) -- (\xLconn,202) -- (\xSub-4,202);
\draw[lconn, line width=0.8pt] (\xLconn,230) -- (\xSub-4,230);

\def\yFa{230}
\node[subcase] at (\xSub,\yFa) {general channels};
\draw[track] (\xL,\yFa+5) -- (\xR,\yFa+5);
\draw[track] (\xL,\yFa-5) -- (\xR,\yFa-5);
\draw[capstoch, line width=0.9pt] (\xL,\yFa+5) -- (\xR,\yFa+5);
\node[sdot] at (\xR,\yFa+5) {};
\draw[capdet, line width=0.7pt, dash pattern=on 3pt off 2pt]
    (\xL,\yFa-5) -- (\xR,\yFa-5);
\draw[capdet, line width=0.9pt] (\xL,\yFa-2) -- (\xL,\yFa-8);
\draw[capdet, line width=0.9pt] (\xR,\yFa-2) -- (\xR,\yFa-8);
\node[qhol] at (\xMid,\yFa-5) {};
\node[qmk]  at (\xMid,\yFa-5) {?};
\node[resbar] at (\xRes,\yFa+5) {Thm~\ref{thm:infinitedim}};
\node[resbar] at (\xRes,\yFa-5) {open};

\def\yFb{202}
\node[subcase] at (\xSub,\yFb) {stochastically degraded w.r.t.\ $U$};
\draw[track] (\xL,\yFb+5) -- (\xR,\yFb+5);
\draw[track] (\xL,\yFb-5) -- (\xR,\yFb-5);
\draw[capstoch, line width=0.9pt] (\xL,\yFb+5) -- (\xR,\yFb+5);
\node[sdot] at (\xR,\yFb+5) {};
\node[ddot] at (\xL,\yFb-5) {};
\node[resbar] at (\xRes,\yFb+5) {Thm~\ref{thm:infinitedim}};
\node[resbar] at (\xRes,\yFb-5) {Thm~\ref{theorem:essentially_omniscient}};

% ------------------------------------------------------------------
% ROW 7 : not I-overwritable   (y = 162)
% ------------------------------------------------------------------
\def\yG{162}
\node[cond]    at (\xCond,\yG+6) {not $I$-overwritable};
\node[subcond] at (\xCond,\yG-7) {negation of Definition~\ref{definition:I-over}};
\node[excl]    at (\xExcl,\yG)   {(weakest condition)};
\draw[track] (\xL,\yG+5) -- (\xR,\yG+5);
\draw[track] (\xL,\yG-5) -- (\xR,\yG-5);
\draw[capstoch, line width=0.9pt] (\xL,\yG+5) -- (\xR,\yG+5);
\node[sdot] at (\xR,\yG+5) {};
\draw[capdet, line width=0.7pt, dash pattern=on 3pt off 2pt]
    (\xL+3,\yG-5) -- (\xR,\yG-5);
\draw[capdet, line width=0.9pt] (\xR,\yG-2) -- (\xR,\yG-8);
\node[ohol] at (\xL,\yG-5) {};
\node[qhol] at (\xMid,\yG-5) {};
\node[qmk]  at (\xMid,\yG-5) {?};
\node[resbar] at (\xRes,\yG+5) {Thm~\ref{thm:infinitedim}};
\node[resbar] at (\xRes,\yG-5) {Thm~\ref{theorem:nonIover_positive}};

\draw[axislt, line width=0.7pt] (\xCond,148) -- (\xRight,148);

% ------------------------------------------------------------------
%  LEGEND -- one row
% ------------------------------------------------------------------
\node[sectlbl] at (\xCond,\yLeg) {Legend};

\draw[track] (10,\yLegA) -- (66,\yLegA);
\draw[capstoch, line width=0.9pt] (10,\yLegA) -- (46,\yLegA);
\node[sdot] at (46,\yLegA) {};
\node[helplbl] at (76,\yLegA) {stochastic capacity};

\draw[track] (200,\yLegA) -- (256,\yLegA);
\draw[capdet, line width=0.9pt] (200,\yLegA) -- (228,\yLegA);
\node[ddot] at (228,\yLegA) {};
\node[helplbl] at (266,\yLegA) {deterministic capacity};

\draw[capdet, line width=0.7pt, dash pattern=on 3pt off 2pt] (390,\yLegA) -- (440,\yLegA);
\draw[capdet, line width=0.9pt] (390,\yLegA+3) -- (390,\yLegA-3);
\draw[capdet, line width=0.9pt] (440,\yLegA+3) -- (440,\yLegA-3);
\node[qhol] at (415,\yLegA) {};
\node[qmk]  at (415,\yLegA) {?};
\node[helplbl] at (450,\yLegA) {value open};

\node[ohol] at (550,\yLegA) {};
\node[helplbl] at (560,\yLegA) {value $>0$};

\draw[axislt, line width=0.4pt] (\xCond,100) -- (\xRight,100);

\end{tikzpicture}%
}%  end resizebox
\caption{Authentication capacity by channel regime, ordered from
strongest (top) to weakest (bottom) overwritability condition.  The
implication arrows relate the \emph{conditions} in the first column: the
downward double arrow is the implication from the stronger condition to
the weaker one and the upward one its converse, struck when the converse
is disproved, with the result or counterexample establishing each printed
beside it; ``conjectured'' marks the implication we believe holds but have
not proved.  The capacity bars, by contrast, describe channels that meet
the row's condition \emph{and} the further assumption listed under
``assuming also''; the result establishing each bar is printed to its
right.}
\label{fig:capacity_summary}
\end{figure}

Our main results are summarized in Figure~\ref{fig:capacity_summary}.We study keyless authentication in the presence of a {myopic adversary}, a model that bridges the gap between the oblivious and omniscient adversary settings. 
The adversary has full knowledge of the encoding procedure and a non-causal noisy observation of the transmitted sequence, but no direct knowledge of the message being transmitted. The adversary is not power-constrained, and the possible attacks by the adversary are determined by the AVC between legitimate users. The lack of power constraint is not a significant assumption for most channels, since many channels, even if they have zero AVC capacity, still have positive authentication capacity --- in fact, as we show, with stochastic encoding, any channel with positive authentication capacity has authentication capacity equal to the no-adversary capacity.

Since the assumption of a myopic adversary lies between the oblivious and omniscient settings, we find results that differ from each. In particular, unlike the oblivious setting, a sufficiently capable adversary --- \emph{e.g.}, one whose side channel is less noisy than the receiver's --- can decode the message and effectively act omnisciently if the encoding procedure is deterministic. On the other hand, when there is sufficient randomness in the encoding procedure, even if the adversary is able to decode the message, they can at most rely on the knowledge of the distribution of the channel symbols. Thus, in the myopic setting, stochastic encoding offers greater resilience to adversarial attacks. Noting this, in this paper, we largely focus on the stochastic encoding setting, although we also present results on deterministic codes. Our results confirm that stochastic codes can outperform deterministic codes in terms of the capacity.

In the oblivious setting, the authentication capacity is characterized in \cite{KosutK:2018} by the channel condition of \emph{overwritability}. Intuitively, overwritability is a channel condition that ensures that the adversary can choose a channel state distribution such that the channel output under any input symbol $x$, and the  channel state distributed as per the adversarial strategy, looks indistinguishable from that under adversary's chosen \emph{decoy} symbol $x'$, and the channel state being equal to a ``\emph{no-adversary state}''. As found in \cite{KosutK:2018}, the oblivious capacity is positive if and only if the channel is non-overwritable; moreover, if the capacity is positive, then it equals the non-adversary capacity. Our results on the myopic setting involve a number of subtly different conditions analogous to overwritability; Figure~\ref{fig:capacity_summary} lists several of these conditions, and the associated capacity statements.

We first introduce a condition for a myopic adversary called \textit{$U$-overwritability}. (Definition~\ref{def:U-overwritable})

We show that this is a sufficient condition for zero authentication capacity for both deterministic and stochastic encoders (Theorem~\ref{thm:converse}). This $U$-overwritability condition is a natural extension of the overwritability condition in~\cite{KosutK:2018}; it guarantees the existence of adversarial strategies that are essentially \emph{single-letter}, \emph{i.e.}, the adversary can   overwrite the true input sequence by any desired decoy sequence of symbols by overwriting one symbol at a time. In doing so, the adversary  relies only on the adversary's channel observation corresponding to that input symbol, and no other knowledge about either the current input symbol or other channel inputs. We also introduce the omniscient version of this condition, \emph{$I$-overwritability} (Definition~\ref{definition:I-over}), which is the same condition under an identity channel to the adversary. Certainly if a channel is $U$-overwritable for any $U$, then it is also $I$-overwritable.

We note that $U$-overwritability is potentially stronger than what an adversary truly needs. In particular, given any fixed stochastic or deterministic coding scheme, the adversary also knows the input distribution from which the transmission is drawn. To account for this, we introduce a relaxation, \textit{$U$-distribution overwritability} (Definition~\ref{def:udistow}), in which the adversary may design its strategy with knowledge of the input distribution rather than the knowledge of the transmitted symbol. Thus, this condition only requires that the adversary's strategy match the legitimate joint output distribution in expectation over the input distribution and the adversary's observation, rather than for every input symbol. We show that if the channel is not $U$-distribution overwritable, then the capacity of the channel with stochastic codes is strictly positive (Theorem~\ref{thm:achievability}). The intuition behind the achievability is that if the channel is not $U$-distribution overwritable, then there must exist an input distribution $P_X$ such that the adversary is \emph{not} able to exactly overwrite an input drawn from $P_X$ by any chosen decoy symbol $x'$. Thus, we build a code with $P_X$ and $x'$ as two distinguishable channel inputs regardless of the adversary's strategy, showing that the capacity is strictly positive.

We further present a multi-letter refinement of this notion, $k$-letter $U$-distribution overwritability. We show that if there exist an integer $k$ such that channel is not $k$-letter $U$-distribution overwritable, then the capacity is positive. We further establish that the resulting chain of conditions is strict at several points, exhibiting in Example~\ref{exmp:BSCZ} an asymmetric binary channel for which $k$-letter $U$-distribution overwritability holds but $(k\!+\!1)$-letter $U$-distribution overwritability fails. Together these conditions give a tight characterization: 
a channel has positive stochastic authentication capacity if and only if there exists a $k$ such that it is not $k$-letter $U$-distribution overwritable (Theorem~\ref{thm:infinitedim}). This characterization is the myopic analog of an equivalence that holds for oblivious adversaries: in that setting, the standard overwritability condition of~\cite{KosutK:2018} coincides with the infinite-letter limit of distribution overwritability, although this is not stated explicitly in~\cite{KosutK:2018}. While this myopic version of the condition is tight, unfortunately it is not in general computable, since it involves checking the $k$-letter condition for every integer $k$. We conjecture that this condition is equivalent to being not $U$-overwritable.

Our framework also yields a tight connection between the authentication problem and the problem of \emph{identification via channels}~\cite{AhlswedeD:1989}. We consider an associated quantity, the \emph{authentication tag} capacity, which captures the rate at which a receiver can verify a candidate message, with possible adversarial interference, without necessarily decoding it. As in the case of identification codes, we show that the number of messages that can be authenticated using an authentication tag grows doubly exponentially with the length of the tag. This leads to a natural two-phase architecture of stochastic authentication codes that achieves rates arbitrarily close to the capacity of the no-adversary channel --- the first phase involves transmitting the message using any code  designed for the no-adversary channel, and the second phase uses the authentication tag to verify if the message decoded in the first phase is consistent with the authentication tag. Related two-phase constructions, in which a short authentication phase is appended to a reliable payload, have appeared in prior work~\cite{KosutK:2018, SangwanBDP:2019a, BeemerKKGY:2019b,GravesBKKY:2023,Jiang:2014, Jiang:2015,GungorK:2016,TuL:2018}. We show that the stochastic authentication capacity and the authentication tag capacity coincide, and that whenever positive, both equal the non-adversarial channel capacity $C_{\szero}$ (Theorem~\ref{thm:dichotomy}). This yields a strong dichotomy: every channel admits authentication either at the full non-adversarial rate or not at all, with no intermediate rates achievable using stochastic encoders.

When we are restricted to deterministic encoders, a sufficiently capable adversary becomes essentially omniscient. In particular, if the channel satisfies \emph{I}-overwritability, and the adversary has a channel that is less noisy than the receiver, \emph{i.e.}, the receiver's channel is stochastically degraded with respect to the adversary's channel, we show that the \emph{I}-overwritability precisely characterizes when the deterministic authentication capacity equals zero (Theorem~\ref{theorem:essentially_omniscient}). For the general setting (\emph{i.e.}, without the stochastic degradedness relation), we show that when the channel is not \emph{I}-overwritable, the deterministic coding authentication capacity is strictly positive (Theorem~\ref{theorem:nonIover_positive}) We explore in detail  the binary symmetric channel, for which an omniscient adversary can easily force the capacity to zero simply by canceling out the transmitter's signal. When the adversary is myopic, we show that if the (binary symmetric) channel to the adversary is better than that to the receiver (but not perfect), then the deterministic authentication capacity is zero (Theorem~\ref{thm:MBAC_deterministic_q<p}), but the stochastic authentication capacity is not (Theorem~\ref{theorem:MBACp>q}); otherwise, both are positive.

\subsection*{Organization}
The paper is organized as follows. In Section~\ref{section:prelims} we introduce the model and necessary notation. 
We revisit known results for oblivious adversaries and introduce the notion of {distribution overwritability} in Section~\ref{sec:oblivious_adv}.
Section~\ref{sec:tags_dichotomy} establishes the connection between our myopic adversary problem and that of authentication tags, and uses this connection to establish the authentication capacity dichotomy.  
We extend the channel conditions of overwritability and distribution overwritability to myopic adversaries in Sections~\ref{sec:u_over} and \ref{sec:u_dist_over}, and show how each relates to authentication capacity. Section~\ref{section:compare_adversaries} establishes a partial order on adversarial channels via stochastic degradation and implications for our overwritability conditions between adversaries related by this partial order. 
In Section~\ref{sec:multi_letter}, we extend the definitions of $U$-(distribution) overwritability from single to multi-letter, and characterize authentication capacity precisely using a further extension termed $U^{\otimes \infty}$-distribution overwritability.
Examples are collected in Section~\ref{sec:examples}. In particular, we present examples showing the strict inclusions between our suite of overwritability conditions, and we begin an examination of a myopic AVC where the channels to adversary and decoder are binary symmetric channels. The latter example is extended in Section~\ref{sec:det}, concluding with a full characterization of (deterministic) authentication capacity of this particular family. Section~\ref{sec:det} also includes more general results on deterministic authentication capacity.
Section \ref{section:conclusions} concludes the paper.

\section{Preliminaries}
\label{section:prelims}

Notation: Throughout the paper, an \textit{$(M,n)$ code} is a code with $M$ codewords and block length $n$. The notation $H(\cdot)$ will indicate the binary entropy function, $\|\cdot\|$ will denote the Hamming weight of a vector, $\mathbb{E}$ will denote expected value, $\mathbb{V}$ will denote the variational distance between two probability distributions, $\oplus$ will be used to indicate (coordinate-wise) modulo $2$ addition, and ``$\markov$'' will indicate a Markov chain. Denote $[i]:=\{1,2,\ldots,i\}$ for any positive integer $i$. A random variable will be denoted using a capital letter, with the corresponding alphabet and alphabet elements written with script and lowercase. 
For example, $X$ is a random variable with alphabet $\mathcal{X}$ such that for all $x\in \mathcal{X}$, $P_{X}(x)$ (or $P(x)$ when the random variable is understood) is the probability that $X=x$.
If $P_X(x)$ is an integer multiple of $\nicefrac{1}{n}$ for all $x\in \mathcal{X}$, we write $\tau_{X}$ to indicate the type class corresponding to $P_{X}$ (i.e. the set of vectors of length $n$ whose empirical distributions match the distribution $P_{X}$).  Bold font will be used to denote vectors. For $\epsilon >0$, 
we define the typical set $\mathcal{T}_{\epsilon}^{(n)}(X)$ to be the set of vectors $\x\in \mathcal{X}^{n}$ such that the empirical probability of each $x\in \mathcal{X}$ in $\x$ is within $\epsilon \cdot P_{X}(x)$ of $P_{X}(x)$. 
That is, if $\mathbf{x} \in \mathcal{T}_{\epsilon}^{(n)}(X)$ then $|Q(x) - P_{X}(x)|\leq \epsilon \cdot P_{X}(x)$, where $Q$ is the empirical distribution of $\mathbf{x}$, for each $x \in \mathcal{X}$. We may write $\mathcal{T}_{\epsilon}^{(n)}$ when the random variable(s) are clear from context.
A \textit{binary symmetric channel} (BSC) with crossover probability $p$, denoted BSC($p$), is a binary-input, binary-output channel such that the probability the output differs from the input is equal to $p$. A \textit{binary erasure channel} (BEC) with erasure probability $p$ outputs an erasure symbol with probability $p$, and otherwise transmits binary inputs reliably.
\subsection{Channel model}
We consider authentication when there is a legitimate transmitter and receiver as well as an active adversary who induces some channel state at each transmission. We assume that the adversary has full knowledge of the codebook of the legitimate parties, though not necessarily of the specific transmission being sent. More formally, let $W_{Y|X,S}$ be a discrete memoryless adversarial channel with the sets $\mathcal{X}, \mathcal{S}$, and $\mathcal{Y}$ as the input, state, and output alphabets.
In the event of multiple channel uses, we write $W(\y\mid \x, \bfs)$ for the product channel, where the sequence lengths are understood. Similarly, there is a discrete memoryless channel between the legitimate transmitter and the adversary given by $U_{Z|X}$, where $\mathcal{Z}$ is the output alphabet at the adversary, so that the adversary has a noisy version $\z$ of any transmitted sequence $\x$ with probability $U(\z \mid \x)$. See Figure \ref{channel_model} for a depiction of this setup. 

\begin{figure}[t]
\centering
\tikzstyle{block} = [draw, fill=lightgray, rectangle, minimum height=10pt, minimum width=20pt, text centered]
{\scriptsize \begin{tikzpicture}
    %all the blocks
    \node[block] (enc) at (0,0) {$\begin{array}{c} \text{Trans.} \\ F \end{array}$};
    \node[block] (chan) at (2.25,0) {$W(y\mid x,s)$};
    \node[block] (dec) at (4.5,0) {$\begin{array}{c} \text{Rec.} \\ \phi \end{array}$};
    \node[block] (adv_chan) at (1,1) {$U(z\mid x)$};
    \node[block] (adv) at (2.25,2.25) {$\begin{array}{c} \text{Adv.} \\ J(\bfs\mid \z) \end{array}$};
    
    %all the arrows
    \draw[->,thick] (-1.1,0) -- (enc); 
    \draw[->,thick] (enc) -- (adv_chan); 
    \draw[->,thick] (enc) -- (chan); 
    \draw[->,thick] (adv_chan) -- (adv);
    \draw[->,thick] (adv) -- (chan);
    \draw[->,thick] (chan) -- (dec);
    \draw[->,thick] (dec) -- (5.5,0);
    
    %all the labels
    \node at (-1.1,0) [left] {$i\in [M]$};
    \node at (1.1,0) [below] {$\X_{i}\in \mathcal{X}^{n}$};
    \node at (3.4,0) [below] {$\y \in \mathcal{Y}^{n}$};
    \node at (5.5,0) [right] {$\hat{i}\in [M]\cup \{0\}$};
    \node at (2.25,1) [right] {$\bfs \in \mathcal{S}^{n}$};
    \node at (0.5,0.5) [above left] {$\X_{i}\in \mathcal{X}^{n}$};
    \node at (1.5,1.5) [above left] {$\z \in \mathcal{Z}^{n}$};\end{tikzpicture}}
\caption{\label{channel_model} The transmitter sends a length-$n$ sequence $\X_{i}:=F(i)$, which is the (possibly stochastic) encoding of a message $i\in [M]$. The adversary views a noisy version $\z$ of this transmission and sends state $\bfs$ to the channel between transmitter and receiver. The receiver then receives $\y$, whose probability is conditioned on both the legitimate transmission and the adversarial state, and decodes to $\phi(\y)\in [M]\cup \{0\}$, where ``0'' is a declaration of adversarial interference.}
\end{figure}

A \textit{myopic adversary} will be formally defined by a distribution $J_{\mathbf{S} \mid \mathbf{Z}}$,
where the adversary's choice of $J$ may depend on $U$ and on the codebook used by legitimate users. Notice that the adversary has a non-causal noisy view of the transmission, so that they may view the entire block length before choosing a state. The case in which the adversary's choice is deterministic, as in \cite{Sarwate:2010}, is included in the above definition. 

We will assume that $\mathcal{S}$ is finite for the remainder of the paper. Let $s_{0}\in \mathcal{S}$ and the corresponding constant state sequence $\bfs_{0}\in \mathcal{S}^{n}$ represent the state in which there is no adversary, \emph{i.e.}, $W_{Y|X, S=s_{0}}$ is a non-adversarial channel. 
\subsection{Authentication Codes, Error Probability, and Capacities}
\begin{definition} \label{def:authcode} An $(M,n)$ \textit{authentication code} for $W$ is an encoder/decoder pair $(F,\phi)$, where the encoder is possibly \textit{stochastic}:
\begin{align*}
	F&: [M] \to \mathcal{X}^{n}, \quad \\
	\phi&: \mathcal{Y}^{n} \to [M]\cup\{0\},
\end{align*}
where an output of ``0'' under $\phi$ indicates a declaration of adversarial interference.
\end{definition} The decoder $\phi$ is successful if (1) the output is equal to the input message, or (2) $\bfs\neq \bfs_{0}$ and the output is equal to $0$. In other words, the decoder either successfully detects adversarial interference or decodes correctly. When $F$ is a deterministic function, we say that $(F,\phi)$ is a \emph{deterministic authentication code}. Otherwise, we call it a \emph{stochastic authentication code}. There are various ways to formalize a stochastic encoder (also called a \textit{randomized encoder}), allowing the transmitter to use some local randomness. We use the following (see, e.g. \cite{ElGamalK:2011}, pg. 550).
\begin{definition}
\label{def:stoch_enc}
Given a message set $[2^{nR}]$, a {\em stochastic encoder} generates a codeword in $\mathcal{X}^{n}$ for $m\in [2^{nR}]$ according to some conditional distribution $P_{\X|M}(\x \mid m)$. 
\end{definition}

For each $A \subseteq [M]\cup \{0\}$, let $\phi^{-1}(A)\subseteq \mathcal{Y}^{n}$ represent the set of channel outputs which decode to some $i\in A$ under $\phi$, and let $\phi^{-1}(A)^{c}$ be the complement in $\mathcal{Y}^{n}$ of this set. Let $\X_{i}:=F(i)$ denote the length-$n$ encoding of message $i$; notice that if the encoder is stochastic, $F(i)$ can take on one of several values, and so $\X_{i}$ is a random variable.
Given transmitted message $i$ and choice of myopic adversary $J$, we define the probability of
error for the authentication code $(F,\phi)$ as:
\begin{align}
\label{equation:error_given_message}
e(i, J) =
 \mathbb{E}_{\X \mid i}\left[\sum_{\z}U(\z \mid \X) J(\bfs_{0} \mid \z)W(\phi^{-1}(i)^{c} \mid \X, \bfs_{0})+\right. 
 \left.\sum_{\bfs \neq \bfs_{0}}\sum_{\z}U(\z \mid \X) J(\bfs \mid \z)W(\phi^{-1}(\{i,0\})^{c} \mid \X, \bfs)\right]
\end{align}
where the expectation is over the realizations of the stochastic encoder conditioned on message $i$ being sent. If the adversary is absent, then $J(\bfs_{0} \mid \z)=1$ for all $\z \in \mathcal{Z}^{n}$. We denote this distribution by $J_{\bfs_{0}}$, and observe that with this choice of adversarial action, \eqref{equation:error_given_message} reduces to the standard error probability given that message $i$ is transmitted.
\begin{equation*}
e(i, J_{\bfs_{0}}) = \mathbb{E}_{\X\mid i}\left[W(\phi^{-1}(i)^{c} \mid \X, \bfs_{0})\right].
\end{equation*}
Similarly, if the adversary decides that $J(\bfs \mid \z)=1$ for all $\z\in \mathcal{Z}^{n}$ and some particular $\bfs \neq \bfs_{0}$, denoted $J_{\bfs}$,
\begin{equation*}
e(i, J_{\bfs}) = \mathbb{E}_{\X \mid i}\left[W(\phi^{-1}(\{i,0\})^{c} \mid \X, \bfs)\right].
\end{equation*}

The above two cases show the reduction to the so-called \textit{oblivious} case (see~\cite{KosutK:2018} for a thorough treatment of this case), in which the adversary has no knowledge about the transmission before choosing a state vector $\bfs$, and so should simply choose the $\bfs$ which maximizes the probability of decoding failure conditioned on $\bfs$.

We assume that each message in $[M]$ is transmitted with equal probability. Then the average probability of error over all possible messages for a given adversarial choice\footnote{Recall that $J$ represents the adversary's choice of distribution on $\mathbf{S}$ after non-causally observing $\mathbf{Z}$.} of $J$ is given by
\begin{equation}
e(J) = \frac{1}{M}\sum_{i=1}^{M}e(i,J).
\end{equation}
The maximal probability of error is given by
\begin{equation}
e_{\max}(J) = \max_{i}e(i,J).
\end{equation}

\begin{remark}
For fixed $i$ and any choice of $J$, $e(i,J)\leq \max_{\bfs'}e(i,J_{\bfs'})$, implying that $\sup_{J}e(i,J)= \max_{\bfs'}e(i,J_{\bfs'})$. However, since the maximizing choice of $\bfs'$ may differ by message $i$, and a myopic adversary has access to some information about $i$, this does not imply that $e(J)\leq \max_{\bfs'}e(J_{\bfs'})$ for all $J$. This marks a fundamental difference from the oblivious adversary model: in that case, an adversarial strategy is defined by a specific choice of $\bfs$ that is made independent of the transmitted message, and the probability of error for that choice is given by $e(\bfs)=\frac{1}{M}\sum_{i}e(i,\bfs)$. Then, it is clearly true that $e(\bfs)\leq \max_{\bfs'}e(\bfs')$ for every choice of adversarial strategy $\bfs$. {In the case of an omniscient adversary, the adversary will choose $\mathbf{s}$ deterministically based on its observation of the transmission (in this case, an adversarial strategy is a deterministic function $\mathcal{X}^{n}\to \mathcal{S}^{n}$). As in the myopic case, the error probability of any adversarial strategy for a specific choice of message is upper bounded by a maximization over state sequences, but the average error probability of an adversarial strategy is not necessarily upper bounded by the average error maximized over a single (universal) choice of sequence.}
\end{remark}

We say that a rate $R$ is \textit{achievable} if there exists a sequence of
$(2^{nR}, n)$ authentication codes such that
$\sup_{J}e(J) \to 0 \text{ as } n\to \infty$, or $\sup_{J}e_{\max}(J) \to 0 \text{ as } n\to \infty$. For stochastic codes, we will primarily consider the maximum probability of error in this paper, noting that the capacity results with the weaker notion of average probability of error will directly follow from them -- since the maximum error approaching zero with increasing block length implies average error also approaches zero. For deterministic codes, our results will be proved under the average probability of error.% and will explicitly state when we are claiming a stronger achievability result (i.e. with reference to the maximal probability of error).
Notice that $\sup_{J}e(J)$ is the highest error probability the adversary can hope for, achieved by designing $J$ optimally. 
The \textit{authentication capacity} $\cst$ is the supremum of all achievable rates. 
Let $C_{\szero}$ denote the \textit{non-adversarial capacity}, that is, the capacity in the no-adversary setting (i.e., $J=J_{\bfs_{0}}$, and the channel is reduced to $W_{Y|X,S=s_{0}}$). Note that, in general, we allow stochastic codes, \emph{i.e.}, the transmitted codeword is a random variable for each input message. When the encoder is restricted to be deterministic, we denote the corresponding supremum by $\cdet$, and call it the \textit{deterministic authentication capacity}. Note that, in general, authentication capacity is an upper bound on the deterministic authentication capacity. In Section~\ref{sec:det}, we give an example of a channel where the deterministic authentication capacity is zero, although the authentication capacity is positive.

\section{Oblivious Adversaries}
\label{sec:oblivious_adv}

Before presenting our results on myopic adversaries, we briefly discuss related notions for oblivious adversaries. The work in \cite{KosutK:2018} shows that a channel property called \textit{overwritability} exactly determines when the authentication capacity is nonzero for oblivious adversaries. 

\begin{definition}[Overwritability for oblivious adversaries]\label{def:overwritability}
An adversarial channel $W_{Y|X,S}$ with no-adversary state $s_0$
is {\em overwritable} if there exists a distribution $P_{S\mid X'}$ such
that
\begin{equation}
\sum_{s}P_{S|X'}(s\mid x')W(y\mid x,s)=W(y\mid x',s_{0}) \text{ for all } x, x', y.\label{eq:def:overwritability}
\end{equation}
\end{definition}

Less formally, over an overwritable channel, an adversary can seamlessly make their own false message appear legitimate to the receiver without being detected. 

\begin{thm}[\cite{KosutK:2018}]
\label{KK18_theorem}
If a channel with an oblivious adversary is not overwritable, then $\cst = \cdet= C_{\szero}$; if it is, $\cst =\cdet = 0$.
\end{thm}

Overwritability should be compared with symmetrizability for the standard AVC problem \cite{CsiszarN:1988}: $W_{Y|X,S}$ 
is \textit{symmetrizable} if there exists $P_{S|X'}$ such
that
$\sum_{s}P_{S|X'}(s\mid x')W(y\mid x,s)=\sum_{s}P_{S|X'}(s\mid x)W(y\mid x',s)$ for all $x, x', y$.
In \cite{KosutK:2018} it is shown that overwritability implies symmetrizability, but that the converse does not hold. Since a classical AVC has zero capacity if and only if it is symmetrizable \cite{CsiszarN:1988}, this implies that the authentication capacity can be positive even though the AVC capacity is zero.
There is an equivalence between overwritability and \textit{simulatability}, as introduced in \cite{Maurer:1997}, when the adversary is oblivious to the transmission but has complete control over the communication channel (i.e. can substitute in its transmission for the legitimate transmission at will).

Next, we introduce a weaker condition that we call \emph{distribution overwritability}. In contrast to overwritability that requires the adversary to simultaneously overwrite all possible channel inputs using a common strategy $P_{S|X'}$, distribution overwritability requires the that the adversary be able to overwrite with the knowledge of the input distribution $P_X$.

\begin{definition}[Distribution overwritability for oblivious adversaries]\label{def:distow} We say that $\wy_{Y|X,S}$ with no-adversary state $\szero$ is distribution overwritable if for every $\Prob_X$, there is an adversarial strategy $\Prob_{S|X'}$ such that 
\begin{equation}
\sum_{x, s}\Prob_X(x) P_{S\mid X'}(s\mid x'){\wy(y|x,s)}=\wy(y\mid x',\szero) \quad \forall x'\in\cX, y\in\cY.	\label{eq:distow_nonmyopic}
\end{equation}
\end{definition}	
\begin{remark} \label{rem:nonow} The distribution overwritability condition implies that, for any distribution $P_X$ and $x'\in\cX$, the  adversary is able to overwrite the channel output to mimic legitimate transmission of symbol $x'$ when the true input to the channel is drawn from the distribution $P_X$. Conversely, when $\wy_{Y|X,S}$ is not distribution overwritable, there exist $P_X$ and $x'\in\cX$ such that for all $\ProbJ_{S}$,
\begin{equation}
\mathbb{V}\left(\expect{X,S}{\wy(\cdot|X,S)},\wy(\cdot|x',\szero)\right)>0.
\end{equation}
Further, since we restrict our attention to finite alphabet channels with finite state spaces, there exists $\eta>0$ such that 
\begin{equation}\label{eq:nonow}
\max_{P_X, x'}\min_{\ProbJ_{S|Z}}\mathbb{V}\left(\expect{X,S}{\wy(\cdot|X,S)},\wy(\cdot|x',\szero)\right)>\eta.
\end{equation}
%Following a similar reasoning as Remark~\ref{rem:nonowchar}, it follows that, when $\wy_{Y|X,S}$ is not distribution overwritable, there exists $\eta>0$ such that 
% \begin{equation}\label{eq:nonownonmyopic}
% \max_{\Prob_X,x'}\min_{\ProbJ_{S|Z}}\mathbb{V}\left(\expect{X,S}{\wy(\cdot|X,S)},\wy(\cdot|x',\szero)\right)>\eta.
% \end{equation}
\end{remark}
We note that overwritability implies distribution overwritability. Indeed, if the channel is overwritable, an adversary that knows the input distribution can satisfy~\eqref{eq:distow_nonmyopic} by employing any strategy $P_{S|X'}$ that is guaranteed by~\eqref{eq:def:overwritability}. On the other hand, distribution overwritability does not imply overwritability. Example~\ref{ex:azbsc} shows a class of channels that are distribution overwritable but not overwritable.

%%%%%%%%%%%%%%%%%%%%%%%%%%%%%%%%%%%%%%%%%%%%%%%%%%%%%
%%%%%%%%%%%%%%%%%%%%%%%%%%%%%%%%%%%%%%%%%%%%%%%%%%%%%
\section{Authentication Tags and the Authentication Capacity Dichotomy}
\label{sec:tags_dichotomy}

In this section we consider a standalone authentication problem, wherein the goal is to authenticate a candidate message already available to the decoder by sending a tag over the channel $(\wy,\uz)$ using \emph{authentication tags} (see Figure~\ref{fig:tag}). Though the authentication tag problem statement differs slightly from ours, we will show that the authentication tag capacity is in fact equal to the authentication capacity of a channel, and that they experience a capacity dichotomy.

An $(N,n)$-authentication tag for a channel $(\wy,\uz)$ consists of a stochastic encoder $\authenc$ that maps messages $m\in[N]$ to codewords $X^n\sim \Prob_{X^n|M=m}$ and a collection of deterministic decoders $\authdec_{\hat{m}}:\cY^n\to\set{\accept,\reject}$ for every $\hat{m}\in[N]$. 
There are two kinds of error probabilities for authentication tags: \emph{(i)} the maximal  probability of rejecting the correct message $m$ when there is no adversary at all and \emph{(ii)} the maximal  probability of accepting a wrong candidate message, whether or not an adversary acts. Note that we don't consider the event that the correct message is rejected when there is adversary present in our definition of error probabilities. %\begin{definition}[Error Probabilities for Authentication Tags] 
We say that an $(N,n)$-authentication tag $(\Prob_{X^n|M},\set{\authdec_m}_{m\in[N]})$ achieves  error probabilities $(\lambda_1,\lambda_2)$ if the following conditions are simultaneously satisfied:
\begin{itemize}
\item[\emph{A.}]\emph{False Alarm:}\vspace{-0.5em}
\[\max_{m\in [N]}\Pr_{X^n,Y^n}(\authdec_m(Y^n)=\reject\vert m\mbox{ sent},S^n=\szero^n)\leq \lambda_1,\ \mathrm{and}\]
\item [\emph{B.}] \emph{Missed Detection} \vspace{-0.5em}
\[\max_{\substack{m\in [N]\\ \hat{m}\in[N]\setminus\set{m}}}\sup_{\ProbJ}\Pr_{X^n,Y^n,Z^n,S^n}(\authdec_{\hat{m}}(Y^n)=\accept\vert m\mbox{ sent})\leq \lambda_2.\]\vspace{-1em}
\end{itemize}

The authentication tag capacity $\ctag$ for a channel $(\wy,\uz)$ is defined to be the supremum over all $R$ such that given any  $\lambda_1,\lambda_2\in (0,1)$, there is a sequence of $(N_n,n)$-authentication tags achieving error probabilities $(\lambda_1,\lambda_2)$, and $\liminf_{n\to\infty} \frac{1}{n}\log_2\log_2{N_n}\geq R$. The problem of authentication using tags is related to the problem of \emph{identification via channels}~\cite{AhlswedeD:1989}. The set of authentication tags includes, as special cases, the set of identication codes. In particular, an  $(N,n)$-authentication tag that achieves error probabilities $(\lambda_1,\lambda_2)$ on the channel $\wy_{Y|X,S}$ is also an $(N,n)$-identification code for the non-adversarial channel $\wy_{Y|X,S=\szero}$ that achieves the error probabilities $(\lambda_1,\lambda_2)$. The reverse is also true when $\mathcal{S}=\set{\szero}$, \emph{i.e.}, for non-adversarial channels, the notions of identification codes and authentication tags coincide. 
%\max_{m\in [N]}\Pr_{X^n,Y^n}(\authdec_m(Y^n)=\reject\vert m\mbox{ sent},s^n=\szero^n)\leq \lambda_1, \mbox{ and}\\
%&\max_{m\in [N]}\max_{\hat{m}\in[N]\setminus\set{m}}\sup_{\ProbJ_{s^n|Z^t}}\Pr_{X^n,Y^n,Z^t,s^n}(\authdec_{\hat{m}}(Y^n)=\accept\vert m\mbox{ sent})\leq \lambda_2.
%\end{align}
%\end{definition}

\begin{figure}[t]
\centering
\tikzstyle{block} = [draw, fill=lightgray, rectangle, minimum height=10pt, minimum width=20pt, text centered]
{\scriptsize \begin{tikzpicture}
    %all the blocks
    \node[block] (enc) at (0,0) {$\begin{array}{c} \text{Trans.} \\ F \end{array}$};
    \node[block] (chan) at (2.25,0) {$W(y\mid x,s)$};
    \node[block] (dec) at (4.5,0) {$\begin{array}{c} \text{Rec.} \\ \phi_{\hat{m}} \end{array}$};
    \node[block] (adv_chan) at (1,1) {$U(z\mid x)$};
    \node[block] (adv) at (2.25,2.25) {$\begin{array}{c} \text{Adv.} \\ J(\bfs\mid \z) \end{array}$};
    
    %all the arrows
    \draw[->,thick] (-1.1,0) -- (enc); 
    \draw[->,thick] (enc) -- (adv_chan); 
    \draw[->,thick] (enc) -- (chan); 
    \draw[->,thick] (adv_chan) -- (adv);
    \draw[->,thick] (adv) -- (chan);
    \draw[->,thick] (chan) -- (dec);
    \draw[->,thick] (dec) -- (5.5,0);
    %\draw[->,thick] (4.5,-1) -- (dec); 
    %all the labels
    \node at (-1.1,0) [left] {$m\in [M]$};
    \node at (1.1,0) [below] {$\X_{m}\in \mathcal{X}^{n}$};
    \node at (3.4,0) [below] {$\y \in \mathcal{Y}^{n}$};
    \node at (5.5,0) [right] {\reject\ or \accept};
    \node at (2.25,1) [right] {$\bfs \in \mathcal{S}^{n}$};
    \node at (0.5,0.5) [above left] {$\X_{m}\in \mathcal{X}^{n}$};
    \node at (1.5,1.5) [above left] {$\z \in \mathcal{Z}^{n}$};
    %\node at (4.5,-1) [below] {$\hat{m}\in[M]$};
\end{tikzpicture}}
\caption{In the authentication tag problem, the transmitter sends a length-$n$ sequence $\X_{m}:=F(m)$, which is a stochastic encoding of a message $m\in [M]$. The adversary views a noisy version $\z$ of this transmission and sends state $\bfs$ to the channel between transmitter and receiver. The receiver has a candidate message $\hat{m}$. Upon receiving $\y$, whose probability is conditioned on both the legitimate transmission and the adversarial state, the receiver decodes to $\phi_{\hat{m}}(\y)\in \{\reject,\accept\}$, where ``\reject'' is a declaration that either $m$ and $\hat{m}$ being unequal or adversarial interference is present, and ``\accept'' is a declaration that $m$ and $\hat{m}$ are equal. Note that declaring ``\accept'' does not preclude the possibility of adversarial interference being present.}\label{fig:tag}
\end{figure}

In the following, we show that the capacities for the authentication problem with stochastic codes and the authentication tag problems are identical. Further, whenever these capacities are positive, they equal the no-adversary capacity of the channel. Before formally giving this result in Theorem~\ref{thm:dichotomy}, we state the auxilliary lemmas needed for the proof of the theorem. 
We first give the following lemma that shows that the aunthentication tag capacity is an upper bound on the authentication capacity.
\begin{lem}$\ctag\geq \cst$.\label{lem:ctaggreater}
\end{lem}

\begin{proof}[Proof sketch] The proof proceeds along the lines of~\cite{SangwanBDP:2019a}. We assume the existence of authentication codes at rates arbitrarily close to $\cst$. As shown in Figure~\ref{fig:twostep_horiz}, we construct a sequence of authentication tags for the channel by composing identification codes~\cite{AhlswedeD:1989} (whose rates approach $1$) for the identity channel with a sequence of authentication codes for our channel (whose rates approach $\cst$). The rate of the resulting sequence of authentication tags, thus, approaches $\cst$. The guarantees on error probabilities follow from that on the identification and authentication codes.  We furnish the details in Appendix~\ref{app:ctaggreater}. 
\end{proof}

\begin{figure}[t]
\centering
% Define a block style with a fixed minimum width for uniformity
\tikzstyle{block} = [draw, fill=lightgray, rectangle, minimum height=20pt, minimum width=1cm, text centered, align=center]
{\scriptsize
\begin{tikzpicture}[node distance=1.5cm and 1.5cm, auto]
    % Encoding chain on the horizontal line
    \node (m) at (-1.5,0) {$m$};
    \node[block] (A) at (0,0) {Noiseless \\ ID Code  $\tilde{\Prob}_{B^{\tilde{n}}|M}$ };
    \node[block] (B) at (2,0) {Auth. Code \\ $\hat{\Prob}_{X^n|\mathbf{B}}$};
    \node[block] (C) at (4,0) {$\wy_{Y|X,S}$};
    \node[block] (D) at (6,0) {Auth. Decoder\\ $\hat{\Phi}$};
    % Decoding branch: two possibilities from the Auth. Decoder
    \node[block] (E) at (9,0) {ID Decoder\\ $\tilde{\phi}_{\hat{m}}$};
    \node[block] (F) at (11,0) {Final Output\\ $\hat{\Phi}\land \tilde{\phi}_{\hat{m}}$};
    
    % Additional input for the ID Code Decoder
    \node (mhat) at (9,-1) {$\hat{m}$};
    \node (output) at (12.5,0) {$0$ or $1$};
    
    % Arrows along the encoding chain
    \draw[->,thick] (m) -- (A);
    \draw[->,thick] (A) -- node[above] {$B^{\tilde{n}}$} (B);
    \draw[->,thick] (B) -- node[above] {$X^n$} (C);
    \draw[->,thick] (C) -- node[above] {$Y^n$} (D);
    
    % Branching from the Authentication Decoder:
    % If the output is non-zero, it proceeds to the ID Code Decoder.
    \draw[->,thick] (D) -- node[above] {{if $\hat{\Phi}(Y^n)\neq 0$}} (E);
    % If the output is 0, it bypasses the decoder and goes directly to the final output.
    \draw[->,thick] (D) .. controls +(0,1) and +(0,1) .. node[above] {if $\hat{\Phi}(Y^n)= 0$} (F);
    
    % Continue from the ID Code Decoder to the final output.
    \draw[->,thick] (E) -- (F);
    
    % Arrow from the additional input to the ID Code Decoder.
    \draw[->,thick] (mhat) -- (E);
    \draw[->,thick] (F) -- (output);
\end{tikzpicture}
}
\caption{Two-step encoding and decoding procedure for authentication tags. The message $m$ is first encoded using a Noiseless ID Code, followed by an Authentication Code, and then transmitted through a Channel. At the receiver, the Authentication Decoder $\phi$ outputs $0$ (yielding an overall output of $0$) or a nonzero value, in which case the signal, together with an additional input $\hat{m}$, is processed by an ID Code Decoder to produce the final output ($0$ or $1$). }
\label{fig:twostep_horiz}
\end{figure}

Our next lemma shows that whenever positive rates for authentication tags are supported, a two-phase construction based on authentication tags can achieve the no-adversary capacity for the authentication problem. 
\begin{lem}\label{lem:dichotomy}
Whenever $C_{\mathrm{tag}}> 0$, $\cst=C_{\szero}$.	
\end{lem}

\begin{proof} We first note that $\cst$ is always bounded from above by $C_{\szero}$ as the latter is the capacity of the channel $\wy$ when there is no adversary. 

In the following we prove that any rate smaller than $C_{\szero}$ is achievable whenever $C_{\mathrm{tag}}$ is positive.  Our achievability relies on a two-phase scheme. The first phase consists of a code for the (non-adversarial) channel $\wy_{Y|X,S=\szero}$. The second phase is a short phase used to verify that the message has been decoded correctly in the first phase. This architecture was previously shown to be capacity-achieving in~\cite{SangwanBDP:2019a} for channels with oblivious adversaries. In the following, we prove that this property continues to hold even when the adversary has a noisy view of the input codeword. 
 
 Suppose that $C_{\mathrm{tag}}>0$. Let $R<C_{\szero}$. Let $\epsilon>0$ be the target probability of error for the authentication code. For every $n>0$, Let $\tn  =\lceil (\log nR)/C_{\mathrm{tag}}\rceil$. Note that $\tn$ is a function of $n$, but we do not state the dependence on $n$ for notational simplicity. Choose $n$ large enough such that the following are ensured:
 \begin{itemize}
 \item[\em a)] There exists a $(2^{nR},n-\tn)$  channel code with encoder $f:[2^{nR}]\to\cX^{n-\tn}$ and decoder $g:\cY^{n-\tn}\to[2^{nR}]$ for the (non-adversarial) channel $\mathsf{W}_{Y|X,S=\szero}$ with maximal error probability at most $\epsilon/2$.	
 \item[\em b)]There exists a $(2^{nR},\tn)$ authentication tag $(\hat{\Prob}_{X^{\tn}|M},\set{\phi_m})$ for $(\wy,\uz)$ achieving error probabilities $(\epsilon/2,\epsilon/2)$.
  \end{itemize}
  
  We form a $(2^{nR},n)$ authentication code $(\Prob_{X^n|M},\Phi)$ by concatenating the above codes as follows. The encoder transmits the channel code followed by the authentication tag, \emph{i.e.}, 
  \[\Prob_{X^n|M}(x^n|m)= \begin{cases}\hat{\Prob}_{X^{\tn}|M}(x_{n-\tn+1}^n|m)& \mbox{if }x^{n-\tn}=f(m),\\ 0&\mbox{otherwise.}\end{cases}\]
  The decoder first decodes the message from the channel code and then authenticates it using the authentication tag. Let $\hat{m}\triangleq g(y^{n-\tn})$. The decoder $\Phi$ is defined as
  \[\Phi(y^n)=\begin{cases}
\hat{m}& \mbox{if }\phi_{\hat{m}}(y_{n-\tn+1}^n)=\accept\\
\reject & \mbox{otherwise}.	
\end{cases}
\]
We argue that the authentication code thus constructed achieves an error probability of $\epsilon$. Consider the following two cases.
\paragraph{Case 1. No adversary} When there is no adversary, the error event is contained in the union of the event that the channel code $(f,g)$ decodes to an incorrect message and the event that the authentication tag $(\hat{\Prob}_{X^{\tn}|M},\set{\phi_m})$ rejects a correct message from the channel code. Thus, 
\begin{align}
\Perr^{(\mathrm{no-adv})}&\leq \max_{m\in[2^{nR}]}\bigg[\Pr(g(Y^{n-\tn})\neq m|m\mbox{ sent}, S^{n-\tn}=(\szero)^{n-\tn})+ \\& \ 	\Pr(\phi_m(Y^{n-\tn+1})_m=\reject|m\mbox{ sent},S_{n-\tn+1}^n=(\szero)_{n-\tn+1}^n)\bigg]\\
&\leq \epsilon/2+\epsilon/2=\epsilon.
\end{align}
\paragraph{Case 2. Adversary present}  In this case, the error probability can be expressed as follows.  In the following, let $l=n-\tn$ for ease of notation.
{\begin{align}
\lefteqn{\max_{m}\sup_{\ProbJ_{S^n|Z^n}}\Pr_{X^n,Y^n,Z^n,S^n}(\Phi(Y^n)\notin\set{M,\reject}|m\textrm{ sent})}\\
& = \max_{m} \sup_{\ProbJ_{S^n|Z^n}}\bigg[\Pr(\Phi(Y^n)\notin\set{m,\reject}, g(Y^{l})=m|m\textrm{ sent})\\
&\quad +\Pr(\Phi(Y^n)\notin\set{m,\reject}, g(Y^{l})\neq m|m\textrm{ sent})\bigg]\\
& \leq \max_m\sup_{\ProbJ_{S^n|Z^n}}\Pr(\Phi(Y^n)\notin\set{m,\reject}, g(Y^{l})=m|m\textrm{ sent})\ +\nonumber \\
&\quad \max_m\sup_{\ProbJ_{S^n|Z^n}}\Pr(\Phi(Y^n)\notin\set{m,\reject}, g(Y^{l})\neq m|m\textrm{ sent})\label{eq:pesum}
\end{align}}
The first term in~\eqref{eq:pesum} corresponds to the channel code decoding to the correct message. In this case,  
the decoder $\Phi$ outputs either the correct message or $\reject$, both of which are acceptable outcomes when an adversary is present. Thus,
\begin{equation*}
	\Pr_{X^n,Y^n,Z^n,S^n}(\Phi(Y^n)\notin\set{m,\reject}, g(Y^{l})=m|m\textrm{ sent})=0.
\end{equation*} 
Next, we analyze the second term in~\eqref{eq:pesum}. {We perform the error analysis by allowing the adversary to choose their strategy based on the true message $m$ and the channel code output $\hat{m}$. We argue that, in this setting, even when the adversary knows $m$ and $\hat{m}$, it is sufficient for the adversary to choose a strategy based on the subset of observations $Z_{l+1}^n$ (in addition to $m$ and $\hat{m}$) rather than all of $Z^n$.The probability of not rejecting a wrong message is upper bounded as follows.}
\begin{align}
\lefteqn{\max_m\sup_{\ProbJ_{S^n|Z^n}}\Pr(\Phi(Y^n)\notin\set{m,\reject}, g(Y^{l})\neq m|m\textrm{ sent})}\\
%\lefteqn{\sup_{\ProbJ_{S^n|Z^n}}\Pr_{M,X^n,Y^n,Z^n,S^n}(\Phi(Y^n)\notin\set{M,\reject}, g(Y^{l})\neq M))
%}\\
%	&=\sup_{\ProbJ_{S^n|Z^n}}\Pr_{X^n,Y^n,Z^n,S^n}(\phi_{\hat{M}}(Y_{l+1}^n)=\accept|\hat{M}\neq M)\\
	&\label{eq:pe}\leq  \max_{m,\hat{m}\neq m}\sup_{\ProbJ}\Pr_{X^n,Y^n,Z^n,S^n}(\phi_{\hat{m}}(Y_{l+1}^n)=\accept|m\mbox{ sent})\\
	&= \max_{m,\hat{m}\neq m}\sup_{\ProbJ} \sum_{\substack{x^n,y^n,z^n,s^n\\ \mathrm{s.t. } \phi_{\hat{m}}(y_{l+1}^n)=\accept}}  \Prob_{X^n|M}(x^n|m)\uz(z^n|x^n) \ProbJ_{S^n|Z^n}(s^n|z^n)\wy(y^n|x^n,s^n)\\
&= \max_{m,\hat{m}\neq m}\sup_{\ProbJ} \sum_{\substack{x^n,y^n,z^n,s^n\\ \mathrm{s.t. } x^{l}=f(m)\\ \phi_{\hat{m}}(y_{l+1}^n)=\accept}}  \hat{\Prob}_{X^{\tn}|M}(x_{l+1}^n|m)\uz(z^n|x^n)  \ProbJ_{S^n|Z^n}(s^n|z^n)\wy(y^n|x^n,s^n)\\
&= \max_{m,\hat{m}\neq m}\sup_{\ProbJ} \sum_{\substack{x_{l+1}^n,y_{l+1}^n,z^n,s^n\\ \mathrm{s.t. } \phi_{\hat{m}}(y_{l+1}^n)=\accept}}  \hat{\Prob}_{X^{\tn}|M}(x_{l+1}^n|m) \uz(z^l|f(m))\uz(z_{l+1}^n|x_{l+1}^n)\ProbJ_{S^{n}|Z^n,S^n}(s^n|z^n) \wy(y_{l+1}^n|x_{l+1}^n,s_{l+1}^n).\label{eq:ProbJfull}
\end{align}
In the above,~\eqref{eq:pe} follows from our design of the authentication code as a two-phase code (see Lemma~\ref{lem:errorprob} for details). Note that~\eqref{eq:pe} does not require a union bound over all $\hat{m}\neq m$, rather, only bounding in terms of the worst-case $\hat{m}$ suffices. Next, defining
\[\hat{\ProbJ}_{S^{\tn}|Z^{\tn}}(s_{l+1}^n|z_{l+1}^n)=\sum_{z^l,s^l}\uz(z^l|f(m))\ProbJ_{S^{n}|Z^n,S^n}(s^n|z^n),\]
the expression in \eqref{eq:ProbJfull} may be rewritten as 
\begin{align}
	&\max_{m,\hat{m}\neq m}\sup_{\hat{\ProbJ}} \sum_{\substack{x_{l+1}^n,y_{l+1}^n,z_{l+1}^n,s_{l+1}^n\\ \mathrm{s.t. } \phi_{\hat{m}}(y_{l+1}^n)=\accept}}  \hat{\Prob}_{X^{\tn}|M}(x_{l+1}^n|m) \hat{\ProbJ}_{S^{\tn}|Z^{\tn}}(s_{l+1}^n|z_{l+1}^n)\wy(y_{l+1}^n|x_{l+1}^n,s_{l+1}^n)\\
	& = \max_{m,\hat{m}\neq m}\sup_{\hat{\ProbJ}} \Pr_{X^{\tn},Y^{\tn},Z^{\tn},S^{\tn}} (\phi_{\hat{m}}=\accept|m\mbox{ sent})\leq \epsilon/2.
	%	&\leq \epsilon/2.
\end{align} This completes the proof that whenever the authentication tag capacity is strictly positive, the authentication capacity equals the no-adversary capacity.\end{proof}

We now state the following lemma that was used in the previous proof. We state the proof of this lemma in Appendix~\ref{app:errorprob}.
\begin{lem}\label{lem:errorprob} Let $\Phi$, $g$ and $\set{\phi_m}_{m}$ be defined as in the proof of Lemma~\ref{lem:dichotomy}. Then, we have, 
\begin{align*}
{\max_m\sup_{\ProbJ_{S^n|Z^n}}\Pr(\Phi(Y^n)\notin\set{m,\reject}, g(Y^{l})\neq m|m\textrm{ sent})}\leq  \max_{m,\hat{m}\neq m}\sup_{\ProbJ_{S^n|Z^n}}\Pr_{X^n,Y^n,Z^n,S^n}(\phi_{\hat{m}}(Y_{l+1}^n)=\accept|m\mbox{ sent}).
\end{align*}
\end{lem}

We are now ready to state the main result of this section.
\begin{thm}\label{thm:dichotomy} $\ctag= \cst$. Further, whenever these are positive, they equal $C_{\szero}$.	
\end{thm}

\begin{proof} 
The proof of this theorem relies on Lemmas~\ref{lem:ctaggreater} and~\ref{lem:dichotomy}. Consider the following two cases.

\paragraph*{Case (i) \underline{$\ctag>0$} } Noting that any $(N,n)$ authentication tag for $\wy_{Y|X,S}$ is also an identification code for the channel $\wy_{Y|X,S=\szero}$. Thus, the authentication capacity $\ctag$ of the channel $\wy_{Y|X,S}$ is bounded from above by the identification capacity of the channel $\wy_{Y|X,S=\szero}$. Further, we know that the identification capacity of the channel for $\wy_{Y|X,S=\szero}$ is equal to $C_{\szero}$~\cite{AhlswedeD:1989}. Thus, $\ctag$ is bounded from above by $C_{\szero}$. Combining this fact with Lemmas~\ref{app:ctaggreater} and~\ref{lem:dichotomy}, we obtain that that \[C_{\szero}\geq \ctag\geq \cst=C_{\szero},\] this establishing that $\ctag\geq \cst=C_{\szero}$. 
\paragraph*{Case (ii) \underline{$\ctag=0$}} In this case, by Lemma~\ref{lem:ctaggreater} and by the non-negativity of $\cst$, we obtain that $\cst =0$.

This proves that $\ctag$ equals $\cst$, and both equal the no-adversary capacity $C_{\szero}$ whenever they are strictly positive.
\end{proof}

\begin{remark} The proof of Lemma~\ref{lem:dichotomy} suggests a natural two-phase architecture for authentication. The first phase may be considered the ``payload'', which can be transmitted using any reliable code for the channel $\wy_{Y|X, S=\szero}$ without an adversary. The second phase is a short authentication phase. Our analysis shows that the length of the second phase needs to be no larger than logarithmic in the overall block length to achieve a vanishing probability of error at rates achieving the capacity. This aspect of communication with authentication has previously been noted in different adversarial models~\cite{SangwanBDP:2019a,BeemerKKGY:2019b}. 
\end{remark}
\begin{remark} In the proof of Lemma~\ref{lem:dichotomy}, the no-adversary case proceeds identically to the setting with an oblivious adversary~\cite{SangwanBDP:2019a}. However, unlike in the oblivious setting, when the adversary is present, they can correlate the attack strategy in the first and the second phases. Thus, it is not \emph{a priori} clear if the authentication tag phase can be analyzed separately from the communication phase. Our proof shows that even though the adversary may choose a strategy that depends on the entire transmission $\ProbJ_{S^n|Z^n}$, such a strategy is no more powerful than strategies $\hat{\ProbJ}_{S^{\tn}|Z^{\tn}}$ that depend only on the authentication tag phase and the knowledge of the message and the first phase reconstruction.
\end{remark}

\section{\texorpdfstring{$U$}{U}-overwritability}
\label{sec:u_over}

When the adversary is myopic rather than oblivious, we must revise the oblivious adversary overwritability conditions (see Definitions \ref{def:overwritability} and \ref{def:distow}). We may approach overwritability for a myopic adversary in several different ways to help us to characterize the authentication capacity. In this and the following section, we define three different notions of overwritability. Later, in Sections~\ref{section:compare_adversaries} and \ref{sec:examples}, we show how these conditions are related to each other and give examples of channels demonstrating that these are not redundant. The first such condition, called \textit{$U$-overwritability}, is a modification, taking the adversary's channel output into account, of the original overwritability condition for oblivious adversaries.

\begin{definition}[$U$-overwritability]
\label{def:U-overwritable}
We say that an adversarial channel $W_{Y|X,S}$ with no-adversary state $s_0$ is {\em $U$-overwritable}, where $U_{Z|X}$ is a conditional distribution, if there exists a distribution $P_{S\mid X', Z}$ such that
\begin{equation}
\sum_{s, z}U(z\mid x)P_{S\mid X', Z}(s\mid x', z)W(y\mid x,s)=W(y\mid x',s_{0}) \quad \forall x, x', y.
\end{equation}
\end{definition}
If a channel is obliviously overwritable, it is clearly also $U$-overwritable (since the adversary can simply choose to ignore $z$). Example \ref{exmp:bitflip_Uow} shows that $U$-overwritability is strictly weaker than oblivious overwritability.

Let $\mathbbm{1}_{Z \mid X}$ be the deterministic identity distribution: that is, $\mathbbm{1}(z\mid x)=1$ if $z=x$, and zero otherwise. This corresponds to the case of a so-called \textit{omniscient} adversary, and induces a special case of $U$-overwritability, as defined below.

\begin{definition}[$I$-overwritability]
\label{definition:I-over}
An adversarial channel $W_{Y|X,S}$ with no-adversary state $s_0$ is {\em $I$-overwritable} (i.e. {\em Identity-overwritable}) if there exists a distribution $P_{S\mid X', X}$ such that
\begin{equation}
\sum_{s}P_{S\mid X', X}(s\mid x', x)W(y\mid x,s)=W(y\mid x',s_{0}) \quad \forall x, x', y.
\end{equation}
\end{definition}

Our next result shows that non-$U$-overwritability is a necessary condition for achieving positive authentication rates. 
\begin{thm}\label{thm:converse} $\cst=0$ if $\wy$ is $U$-overwritable.	
\end{thm}

\begin{proof}
Suppose $\wy$ is $U$-overwritable, and let $P_{S|X',Z}$ be the distribution guaranteed by Definition \ref{def:U-overwritable}. Consider a sequence of $(2^{nR}, n)$ authentication codes with fixed $R>0$, and let $M:=2^{nR}$. Let $\tilde{J}(\bfs \mid \z):=\frac{1}{M}\sum_{j=1}^{M}\mathbb{E}_{\X'\mid j}\left[P_{S|X',Z}(\bfs \mid \X', \z)\right]$, where the expectation is over the stochastic encoding of message $j$. Then,
\begin{align}
e(\tilde{J}) &= \frac{1}{M}\sum_{i=1}^{M} e(i,\tilde{J}) \\
&= \frac{1}{M^{2}}\sum_{i,j}\mathbb{E}_{\X\mid i}\left[\mathbb{E}_{\X' \mid j}\left[W(\phi^{-1}(\{i,0\})^{c} \mid \X', \bfs_{0})\right]\right] \label{equation:u-over}\\
&\geq \frac{1}{M^{2}}\sum_{j}\sum_{i\neq j}\mathbb{E}_{\X'\mid j}\left[W(\phi^{-1}(j) \mid \X', \bfs_{0})\right] \label{equation:i-j}\\
&= \frac{1}{M^{2}}\sum_{j}\sum_{i\neq j}(1-e(j,J_{\bfs_{0}}))  \\
&\geq \frac{M-1}{M}\left(1-e(J_{\bfs_{0}})\right), 
\end{align} 
where \eqref{equation:u-over} follows from $U$-overwritability and \eqref{equation:i-j} follows because $j\in\{i,0\}^{c}$ as long as $i\neq j$. The deterministic distribution $J_{\bfs_{0}}$ is as defined in Section \ref{section:prelims}. Since $\sup_{J}e(J) \geq e(\tilde{J})$ and $\sup_{J}e(J) \geq e(J_{\bfs_{0}})$,
\[ \sup_{J}e(J) \geq \frac{M-1}{M}\left(1-e(J_{\bfs_{0}})\right) \geq \frac{M-1}{M}\left(1-\sup_{J}e(J)\right). \]
We conclude that  $\sup_{J}e(J) \geq \frac{M-1}{M}\left(1-\sup_{J}e(J)\right)$. This reduces to $\sup_{J}e(J) \geq \frac{M-1}{2M-1}$,
 which approaches $1/2$ as $n\to \infty$. Thus, $\cst=0$.
\end{proof}

\section{\texorpdfstring{$U$}{U}-distribution overwritability}
\label{sec:u_dist_over}

Similar to the notion of distribution overwritability for oblivious channels, our next condition, \textit{\distribution overwritability}, is a relaxation of the $U$-overwritability condition that gives the adversary additional knowledge of the adversary's channel output distribution. We say that input distributions  $P_X,P'_X$ are \emph{myopically indistinguishable} if they induce the same distribution for the adversary's channel output,  \ie,
\[\sum_{x\in\cX}P_X(x)U(z|x)=\sum_{x\in\cX}P'_X(x)U(z|x) \ \forall z\in\cY.\] Let $\fP_{\mathrm{ind}}$ be the collection of equivalence classes formed by grouping indistinguishable input distributions.  
\begin{definition}[\distribution overwritability]\label{def:udistow} We say that $\wy_{Y|X,S}$ with no-adversary state $\szero$ is \distribution overwritable if for every $\cP\in\fP_{\mathrm{ind}}$, there is an adversarial strategy $\Prob_{S|X',Z}$ such that 
\begin{equation}
\sum_{x, s, z}\Prob_X(x) U(z\mid x)P_{S\mid X', Z}(s\mid x', z){\wy(y|x,s)}=\wy(y\mid x',\szero) \quad \forall \Prob_X\in\cP, x'\in\cX, y\in\cY.	\label{eq:def:udistow}
\end{equation}
\end{definition}	

Note that if a channel $W_{Y|X,S}$ is $U$-overwritable, selecting $P_{S\mid X', Z}=\tilde{P}_{S\mid X', Z}$, where the latter distribution is that guaranteed by $U$-overwritability, immediately gives $U$-distribution overwritability.

\begin{remark} \label{rem:nonowchar} Intuitively, the \distribution overwritability condition requires that the myopic adversary be able to mimic legitimate transmission of any symbol $x'$ of their choice when the true input to the channel is drawn from any distribution that leads to a publicly known distribution $P_Z$ for the adversary. Following a similar reasoning as Remark~\ref{rem:nonow}, we note that when $\wy_{Y|X,S}$ is not \distribution overwritable, there exists $\eta>0$ such that 
\begin{equation}\label{eq:nonowmyopic}
\max_{\substack{\cP\in\fP_{\mathrm{ind}}\\ x'\in\cX}}\min_{\ProbJ_{S|Z}}\max_{\Prob_X\in\cP}\mathbb{V}\left(\expect{X,Z,S}{\wy(\cdot|X,S)},\wy(\cdot|x',\szero)\right)>\eta. 
\end{equation} 
\end{remark}

We note that the notion of (oblivious) distribution overwritability defined in Definition \ref{def:distow} is not a special case of the notion of $U$-distribution overwritability. In particular, if we were to consider channels $U$ whose outputs are independent of the input, all input distributions would be myopically indistinguishable distributions. For such channels, satisfying ~\eqref{eq:def:udistow} would imply that the adversary should be able to overwrite all input symbols simultaneously through the same strategy $P_{S|X'}$. Thus, this implies that the channel $W$ is overwritable in the sense of Definition~\ref{def:overwritability}, which would be a stronger condition than $W$ being distribution overwritable (see Example~\ref{ex:azbsc} for an example of this distinction).

Nonetheless, as the following result suggests, the notion of distribution overwritability is equivalent to $U$-distribution overwritability when the channel $U$ is sufficiently noisy but not useless.  We give the proof of this result in Appendix~\ref{app:lem:u_epsilon_k}.
\begin{lem}\label{lem:u_epsilon_k} Let $U_\epsilon$ be an erasure channel with erasure probability $1-\epsilon$; that is, it has input and output alphabets $\cX$ and $\cZ=\cX\cup\{E\}$ respectively and has transition probabilities
 \begin{align*}
 U_\epsilon(z|x)=\begin{cases}
 \epsilon & z=x\\
 1-\epsilon & z=E
 \end{cases}
 \end{align*}$\wy^{\otimes k}$ is distribution overwritable if and only if there exists $\epsilon_k\in(0,1)$ such that $\wy^{\otimes k}$ is $\uz_\epsilon^{\otimes k}$-distribution overwritable for every $\epsilon\in(0,\epsilon_k]$.    
\end{lem}

Theorem~\ref{thm:dichotomy} shows a dichotomy property for the authentication problem with myopic adversaries. Such dichotomies are well known in the AVC literature for both the reliable communication problem (\emph{e.g.},~\cite{Ahlswede:1978,CsiszarN:1988}) as well as the authentication problem~\cite{KosutK:2018}. In order to fully characterize the capacity, it is thus sufficient to characterize the condition under which the capacity is positive. This observation motivates our next result to give a general sufficient condition for the positivity of authentication capacity. We show that if the channel $\wy$ is not $\uz$-distribution overwritable (as per Definition~\ref{def:udistow}), transmitting an authenticated message at a positive rate is possible.  
\begin{thm}\label{thm:achievability} $\cst>0$ if $\wy$ is not \distribution overwritable.	
\end{thm}

To prove Theorem~\ref{thm:achievability}, we first introduce the following $(\uz,\cX)$-distribution overwritability condition that is implied by the \distribution overwritability condition stated in Definition~\ref{def:udistow}.
\begin{definition}[$(\uz,\cX)$-distribution overwritability]\label{def:upow} We say that $\wy$ is $(\uz,\cX)$-distribution overwritable if for every $P\in\mathcal{P}(\cX)$ and $x'\in\cX$, there is an adversarial strategy $\ProbJ_{S|Z}$ such that, for every $y\in\cY$, 
\begin{equation}
\expect{X,Z,S}{\wy(y|X,S)}=\wy(y|x',\szero),	
\end{equation}
where, the expectation is with respect to the joint distribution $\Prob_{X,Z,S}= \Prob_X \uz_{Z|X} \ProbJ_{S|Z}$.
\end{definition}

Note that for channels that are $(\uz,\cX)$-distribution overwritable, as long as the adversary knows the channel input distribution $P_X$, it is able to overwrite the input with any desired symbol $x'$. Conversely, if a channel is not $(\uz,\cX)$-distribution overwritable, there exists an input distribution $P_X$ and an input symbol $x'$ such that inputs drawn using $P_X$ cannot be overwritten as $x'$. This property can be formalized in a manner similar to Remark~\ref{rem:nonowchar}. Our tag construction for channels that are not $(\uz,\cX)$-distribution overwritable relies on selecting such a pair $(P_X,x')$ to construct an authentication tag of the following form.
\begin{definition}[$(P_X,n, \mathfrak{B},\mu)$-authentication tag] \label{def:authtag} Given $P_X\in\mathcal{P}(\cX)$, blocklength $n$,  a family of subsets of $\mathfrak{B}$ of $[n]$,  and a decoding threshold $\rho$, a $(P_X,n, \mathfrak{B},\mu)$-authentication tag is defined through the encoder map
\[\Prob_{X^n|M}(x^n|m)=\prod_{i\notin\cB}\Prob_X(x)\prod_{i\in \cB}\mathbbm{1}_{x'}(x), \quad x^n\in\cX^n,\]
 and the decoder maps $\set{\phi_m}$ specified below.
Given $y\in\cY$, $y^n\in\cY^n$, and $\cB\in\mathfrak{B}$, let \[\hat{P}_{y^n}(y|\cB)\triangleq \frac{\{i\in\cB:y_i=y\}}{|\cB|}.\]
For each $m\in[N]$, and  $y^n\in\cY^n$ the decoder $\phi_m$ outputs according to the following rule
\[\phi_m(y^n)=\begin{cases}
	\accept & \mathrm{if}\ \mathbb{V}(\hat{P}_{y^n}(\cdot|\cB_{m}),\wy(\cdot|x',\szero))<\rho\\
	\reject&\mathrm{otherwise.} 
\end{cases}\]
\end{definition}

Next, in order to show the existence of authentication tags at positive rates and achieve small enough error probabilities, we rely on the following lemmas. Lemma~\ref{lem:overlappingsets} shows the existence of an appropriate subset $\mathfrak{B}$ and Lemma~\ref{lem:typicality} shows that, for any adversarial strategy $J$, the probability that the subset $\cB$ of input symbols drawn from $P_X$ can be overwritten as $x'$ is exponentially small.
\begin{lem} \label{lem:overlappingsets}Let $n\in\mathbb{N}$ and $\alpha,\beta\in(0,1)$. Let $R<\min\set{\beta^2\alpha^2/6,\beta^2\alpha(1-\alpha)/4}$. Then, there exists a family  $\mathfrak{B}$ of subsets of $[t]$ satisfying:
\begin{enumerate}
	\item $ \alpha (1-\beta) n\leq|\cB|\leq \alpha (1+\beta) n$ for all $\cB\in\mathfrak{B}$,
	\item $|\cB\cap\cB'|<\alpha^2(1+\beta) n$ for all $\cB,\cB'\in\mathfrak{B}$,
	\item $|\cB'\setminus\cB|>\alpha (1-\alpha)(1+\beta) n$ for all $\cB,\cB'\in\mathfrak{B}$, and
	\item $|\mathfrak{B}|\geq 2^{R n}$.
	\item For every $y^n\in\cY^n$, \begin{equation}
	\mathbb{V}\left(\hat{P}_{y^n}(\cdot|\cB'),\wy(\cdot|x',\szero)\right) \geq  \frac{(1-\alpha)(1-\beta)}{1+\beta}\mathbb{V}\left(\hat{P}_{y^n}(\cdot|\cB'\setminus\cB),\wy(\cdot|x',\szero)\right)-2 \frac{\alpha(1+\beta)}{(1-\beta)}.
\end{equation}
\end{enumerate}
\end{lem}

\begin{lem}\label{lem:typicality} Let $(X^n,Y^n,Z^n,S^n)$ be drawn from  $\prod_i P_{X}(\cdot)\prod_{i}\uz (\cdot|\cdot) \ProbJ_{S^n|Z^n}(\cdot|\cdot)\prod_{i}\wy(\cdot|\cdot,\cdot)$.  For any $P_X\in\mathcal{P}(\cX)$ and $x'\in\cX$, let
\begin{equation}
 \kappa(P_X,x')\triangleq \min_{Q_{S|Z}}\mathbb{V}\left(\sum_{x,z,s}P_X(x)\uz(z|x)Q_{S|Z}(s|z)\wy(\cdot|x,s),\wy(\cdot|x',\szero)\right). \label{eq:kappa}
\end{equation} Then, there exists $\gamma>0$, such that for $n$ large enough and for all $\ProbJ$, 
\[\Pr(\mathbb{V}(\hat{P}_{Y^n},\wy(\cdot|x',\szero))<\kappa(P_X,x')/2)<2^{-n\gamma}.\]
\end{lem}

We prove the above lemmas in Appendix~\ref{app:typicality}. 

We can now turn to our proof of Theorem~\ref{thm:achievability}. We proceed by considering two classes of channels -- channels that are $(\uz,\cX)$-distribution overwritable and channels that are \distribution overwritable but not $(\uz,\cX)$-distribution overwritable. We construct an authentication tag for the former class of channels using the construction described in Definition~\ref{def:authtag}. For the latter class of channels, we use a construction that concatenates together tags of the form described in Definition~\ref{def:authtag} for different choices of $P_X$.
\begin{proof}[Proof of Theorem~\ref{thm:achievability}] We prove that $\cst>0$ for channels that are not $\uz$-distribution writable by showing the existence of authentication tags with positive rates for such channels and invoking Theorem~\ref{thm:dichotomy}. Suppose that $(\wy,\uz)$ satisfy~\eqref{eq:nonow} for some $\nu>0$ and let $(\cP,x')$ be a pair achieving the maxima in~\eqref{eq:nonow}. Let $\alpha,\beta>0$ be small enough so as to satisfy
\begin{equation}
	\frac{1+\beta}{(1-\alpha)(1-\beta)}\left(\mu/4+2\frac{\alpha(1+\beta)}{1-\beta}\right)\leq \frac{\mu}{2},\label{eq:alphabeta}
\end{equation} 
and choose $\mathfrak{B}=\set{\cB_m:m\in[N]}$ as per Lemma~\ref{lem:overlappingsets}. 

\paragraph*{\underline{Case 1: $|\cP|=1$}} Let $\cP=\set{P_X}$. Note that, in this case, $\wy$ is $(\uz,\cX)$-overwritable (as stated in~Definition~\ref{def:upow}). Following along a similar reasoning in Remark~\ref{rem:nonowchar}, we note that if $\wy$ is not  $(\uz,\cX)$-overwritable, there exist $P_X\in\mathcal{P}(\cX)$ and $x'\in\cX$ and $\mu>0$ such that 
\begin{equation}
\min_{\ProbJ_{S|Z}}\mathbb{V}\left(\expect{X,Z,S}{\wy(\cdot|X,S)},\wy(\cdot|x',\szero)\right)> \mu.\label{eq:nonupow}\end{equation}
Next, we invoke the Theorem~\ref{thm:dichotomy} to note that it is sufficient to show the existence of authentication tags of positive rate for large enough blocklengths. To this end, let $(\lambda_1,\lambda_2)$ be the desired false alarm and missed detection probabilities for the authentication tag. 

 Consider a $(P_X,n,\mathfrak{B},\mu/4)$-tag $(P_{X^n|M},\set{\phi_m})$ as given in Definition~\ref{def:authtag}.
We first analyze the false alarm probability.  Note that 
\begin{align}
&\Pr_{X^n,Y^n}(\authdec_m(Y^n)=\reject\vert m\mbox{ sent},S^n=\szero^n)	\\
& = \Pr_{X^n,Y^n}\left(\mathbb{V}\left(\hat{P}_{Y^n}(\cdot|\cB_m),\wy(\cdot|x',\szero\right)>\mu/4\right)\\
&= \Pr_{Y_{\cB_m}}\left(\mathbb{V}\left(\hat{P}_{Y^n}(\cdot|\cB_m),\wy(\cdot|x',\szero\right)>\mu/4 \big| x_{\cB_m}=(x')^{|\cB_m|}\right).
\end{align}
Now, noting that for each $i\in\cB_m$, $x_i=x'$, and hence, $Y_i\sim \wy(\cdot|x',\szero)$, the probability of false rejection is simply the probability that a sequence of length $|\cB_m|$ is not typical when each symbol in the sequence is drawn i.i.d. from $\wy(\cdot|x',\szero)$. Thus, by the Weak Law of Large Numbers, for $n$ large enough, the false alarm probability for our construction is smaller than $\lambda_1$.

Now, we analyze the missed detection probability. To this end, note that, for any $m$, $\hat{m}$ s.t. $\hat{m}\neq m$ and any adversarial strategy $Q_{S^n|Z^n}$, we have,  
\begin{align} 
&	\Pr_{X^n,Y^n,Z^n,S^n}(\authdec_{\hat{m}}(Y^n)=\accept\vert m\mbox{ sent})\\
& = \Pr_{X^n,Y^n,Z^n,S^n}\left(\mathbb{V}\left(\hat{P}_{Y^n}(\cdot|\cB_{\hat{m}}),\wy(\cdot|x',\szero\right)<\mu/4\right)\\
& \stackrel{(a)}{\leq} \Pr_{X^n,Y^n,Z^n,S^n}\Bigg[\mathbb{V}\left(\hat{P}_{Y^n}(\cdot|\cB_{\hat{m}}\setminus\cB_{m}),\wy(\cdot|x',\szero\right)<\frac{1+\beta}{(1-\alpha)(1-\beta)}\left(\mu/4+2\frac{\alpha(1+\beta)}{1-\beta}\right)\Bigg]\\
& \stackrel{(b)}{\leq}  \Pr_{X^n,Y^n,Z^n,S^n}\Bigg[\mathbb{V}\left(\hat{P}_{Y^n}(\cdot|\cB_{\hat{m}}\setminus\cB_{m}),\wy(\cdot|x',\szero\right)<\frac{\mu}{2}\Bigg]\label{eq:lastprob}.
\end{align}
In the above, $(a)$ follows from Property 5) of Lemma~\ref{lem:overlappingsets}. $(b)$ follows from the condition in~\eqref{eq:alphabeta}. Finally, we note that for all $i\in\cB_{\hat{m}}\setminus\cB_m$, $X_i$ is drawn i.i.d. from $P_X$. Thus, application of Lemma~\ref{lem:typicality} gives that the probability in~\eqref{eq:lastprob} is bounded from above by $2^{-n\alpha(1-\alpha)(1+\beta)\gamma}$ for some $\gamma>0$. This proves that the missed detection probability is smaller than $\lambda_2$ for $n$ large enough.

 \paragraph*{\underline{Case 2: $|\cP|>1$}} Now, we extend the result to the case when the set $\cP$ achieving the maximum contains more than one element. First, fix $\delta>0$ and let $\cP_{\delta}$ be a $\delta$-net covering $\cP$, \emph{i.e.}, for all $P\in\cP$, there is $P'\in\cP_{\delta}$ such that $\mathbb{V}(P',P)<\delta$. Note that $\cP_\delta$ may be chosen to have a finite number of elements. In particular, we can always find $\cP_{\delta}$ such that $|\cP_\delta|<\frac{1}{\delta^{|\cX|}}$. Let $\cP_{\delta}=\set{P^{(1)},P^{(2)},\ldots,P^{(K)}}$. For each $k\in[K]$, let $(P^{(k)}_{X^n|M},\set{\phi^{(k)}_m})$ be a $(P^{(k)}_X,n,\mathfrak{B},\mu/4)$-authentication tag  as given in Definition~\ref{def:authtag}.
 
Our achievability relies on an authentication tag consisting of several sub-blocks. Consider a $(N,nL)$-authentication tag $(P_{X^{nL}|M},\set{\phi_m})$ with encoder and decoder maps defined below. For each $l\in[L]$, the encoder uniformly picks $k$ from $[K]$ and encodes the sub-block $X_{(l-1)n}^{ln}$ according to $P_{X^n|M}^{(k)}$. Thus,
  \[P_{X^{nL}|M}(x^{nL}|m)= \prod_{l=1}^L \sum_{k=1}^K\frac{1}{K}P_{X^n|M}^{(k)}(x_{(l-1)n+1}^{ln}|m).\]
  The decoder first decodes each sub-block and outputs an $\accept$ only if, for each block, the corresponding decoder outputs $\accept$, \emph{i.e.}, 
\[\phi_m(y^{nL})=\begin{cases}\accept & \mbox{if }\phi_m^{(k)}(y_{(l-1)n+1}^{ln})=\accept\ \forall\ l\in[L]\\ \reject&\mbox{otherwise}.
	\end{cases}
\]Note that, the decoder doesn't know \emph{a priori} the value of $k$ for each sub-block (since the value of $k$ is randomly and privately chosen by the transmitter). In fact, the decoder $\set{\phi_m^{(k)}}$ for any $(P^{(k)}_X,n,\mathfrak{B},\mu/4)$-authentication tag only needs the knowledge of $\mathfrak{B}$ and not of $P^{(k)}_X$. Let $\kappa^{(k)}=\kappa(P_X^{(k)},x')$ (as defined in~\eqref{eq:kappa}). Note that, over the uniform choice of $k$, $\Pr(\kappa^{(k)}>\mu/2)\geq \frac{1}{K}$ as long as $\delta$ is small enough. Thus, by Lemma~\ref{lem:typicality}, with probability at least $1/K$, $\phi_{\hat{m}}(\cdot)=\reject$ when $m\neq \hat{m}$. We first let the number of sub-blocks to be large enough and then the length of each sub-block to be large enough to conclude that the error probabilities can be made as small as desired.
\end{proof}

\section{Overwritability and positivity under stochastic degradation}
\label{section:compare_adversaries}

A myopic adversary spans the gap between omniscient and oblivious adversaries; as a result, there may be a level of myopia at which the authentication capacity becomes positive for a channel whose authentication capacity under an omniscient adversary is zero. Since not every pair of potential channels to the adversary are directly comparable, we address the transition from oblivious to omniscient with the use of a partial ordering on channels $U_{Z|X}$ given by stochastic degradation.

\begin{definition}\label{def:stoch_degraded}
We say that a channel $\tilde{V}_{\tilde{A}|B}$, whose input and output alphabets are $\cB$ and $\tilde{\cA}$ respectively, is {\em stochastically degraded} with respect to $V_{A|B}$, whose input and output alphabets are $\cB$ and $\cA$ respectively, if there exists $P_{\tilde{A}|A}$ such that for all $b\in \mathcal{B},\tilde{a}\in \tilde{\mathcal{A}}$:
\begin{equation}\tilde{V}_{\tilde{A}|B}(\tilde{a}|b)=\sum_{a\in \mathcal{A}}P_{\tilde{A}|A}(\tilde{a}|a)V_{A|B}(a|b).\label{eq:def:degraded}\end{equation}
\end{definition}

 Formally, we say that  $\tilde{V}_{\tilde{A}|B} \leq V_{A|B}$ if and only if $\tilde{V}_{\tilde{A}|B}$ is stochastically degraded with respect to $V_{A|B}$. Examples \ref{example:max_min} and \ref{example:totally_ordered_BSC} illustrate some features of this partial order.

\begin{exmp}[Extremes of the partial order: oblivious and omniscient]
\label{example:max_min}
The channel to the adversary in the oblivious case (\emph{i.e.}, the case in which $Z$ is independent of $X$) is stochastically degraded with respect to any other channel $U_{Z|X}$. On the other hand, every channel $U_{Z|X}$ is stochastically degraded with respect to the omniscient channel ($U_{Z|X}=\mathbbm{1}_{Z|X}$). Thus, these extremes give the unique maximum (omniscient adversary) and minimum (oblivious adversary) elements of our partial order. 
\end{exmp}

\begin{exmp}[BSCs as a totally ordered chain]
\label{example:totally_ordered_BSC}
As an example of a totally ordered chain within the partial order, consider the set of binary symmetric channels with crossover probability at most $1/2$. The BSC($p$) is stochastically degraded with respect to the BSC($q$) if and only if $p\geq q$.  The maximum element of the chain is the BSC($0$) (omniscient adversary), and the minimum is the BSC($0.5$) (oblivious adversary). See Figure \ref{figure:BSC_example}.
\end{exmp}

Using the partial order of stochastic degradation, we may compare different channels to the adversary and then use these comparisons to draw conclusions about positivity of the capacity. 

\begin{thm}
\label{theorem:stoch_deg_posititivity}
If $
\tilde{U}_{\tilde{Z}|X}$ is stochastically degraded with respect to $U_{Z|X}$, and the stochastic (resp. deterministic) authentication capacity of $(W,U)$ is positive, then the stochastic (resp. deterministic) authentication capacity of $(W,\tilde{U})$ is positive.
\end{thm}

We prove the above theorem in Appendix~\ref{app:theorem:stoch_deg_posititivity}. The partial ordering of the positivity of the authentication capacity under stochastic degradation naturally implies a corresponding partial ordering of the specific channel condition that characterizes the positivity. Although we do not have a general single letter characterization of the exact positivity condition in our setting, we may still compare channels that are stochastically degraded with respect to our necessary and sufficient single-letter conditions for the (non) positivity of the authentication capacity, \emph{i.e.}, $U$-overwritability and $U$-distribution overwritability.

\begin{thm}
\label{theorem:stoch_deg_overwritability}
If $\tilde{U}_{\tilde{Z}|X}$ is stochastically degraded with respect to $U_{Z|X}$, and $W_{Y|X,S}$ is $\tilde{U}$-overwritable (resp. $\tilde{U}$-distribution overwritable), then $W_{Y|X,S}$ is also $U$-overwritable (resp. $U$-distribution overwritable). 
\end{thm}

We give the proof of the above theorem in Appendix~\ref{app:theorem:stoch_deg_overwritability}. Theorem \ref{theorem:stoch_deg_overwritability} implies that overwritability is a stronger condition than $U$-(distribution) overwritability, formalized in the following corollary.

\begin{co}\label{cor:re}
If a channel $W_{Y|X,S}$ is overwritable, then it is $U$-overwritable, and thus $U$-distribution overwritable, for every channel $U_{Z|X}$.
\end{co}

\begin{proof}
The channel to an oblivious adversary is stochastically degraded with respect to every channel $U_{Z|X}$. 
The result follows by Theorem \ref{theorem:stoch_deg_overwritability}, and the fact that $U$-overwritability implies $U$-distribution overwritability.
\end{proof}

Clearly, if the adversary can successfully generate a false message given a noisy version of a transmission, it is also successful in the noiseless case. This is formalized in the following corollary.

\begin{co}
\label{cor:implies_I-over}
If a channel $W_{Y|X,S}$ is $U$-distribution overwritable for some $U_{Z|X}$, then it is also $I$-overwritable.
\end{co}

\begin{proof}
Every channel $U_{Z|X}$ is stochastically degraded with respect to $\mathbbm{1}_{Z \mid X}$. The result follows by Theorem \ref{theorem:stoch_deg_overwritability}.
\end{proof}

On the other hand, $I$-overwritability does not imply $U$-distribution overwritability for the channel $U_{Z|X}$ to the adversary, as is demonstrated in Example \ref{exmp:bitflip_IvsUdistow}.

\begin{remark}
In general, it is straightforward to see that the existence of a distribution $P_{S|X',Z}$ (resp. $P_{S|X'X}$) satisfying the equality in Definition \ref{def:U-overwritable} (resp. Definition \ref{definition:I-over}), and thus the $U$-overwritability (resp. $I$-overwritability) of a channel, may be determined via linear programming. Once $U$-overwritability has been determined for a particular choice of $U_{Z|X}$, Theorem \ref{theorem:stoch_deg_overwritability} allows us to determine $U$-overwritability for less degraded channels.
\end{remark}

\section{Multi-letter distribution overwritability conditions}
\label{sec:multi_letter}

So far, the overwritability conditions that we have considered (Definitions~\ref{def:U-overwritable} and~\ref{def:udistow}) are single-letter characterizations. We now turn our attention to their adaptations to multi-letter versions of the channels. 

\begin{definition}[$k$-letter channel] \label{def:kletterchannel} Given a channel $(\wy,\uz)$, and a positive integer $k$, the $k$-letter channel $(\wy^{\otimes k},\uz^{\otimes k})$ consists of the channels defined through transition probabilities 
\begin{align}
    \wy^{\otimes k} (y^k|x^k,s^k) & = \prod_{i=1}^k \wy (y_i|x_i,s_i),\mbox{ and}\\
    \uz^{\otimes k} (z^k|x^k) & = \prod_{i=1}^k \uz (z_i|x_i).
\end{align} 
\end{definition}
Our next two results derive relations between overwritability conditions for the single letter and multi-letter versions of the channel. The proofs of the following theorems are given in Appendices~\ref{app:thm:koverwritability} and~\ref{app:thm:kdistoverwritability}. 
\begin{thm}\label{thm:koverwritability}
     For every pair of positive integers $k$ and $k'$ such that $k'\geq k$, $\wy^{\otimes k}$ is $\uz^{\otimes k}$-overwritable if and only if $\wy^{\otimes k'}$ is $\uz^{\otimes k'}$-overwritable. 
\end{thm}
\begin{thm}\label{thm:kdistoverwritability}
    For every pair of positive integers $k$ and $k'$ such that $k'\geq k$, $\wy^{\otimes k}$ is $\uz^{\otimes k}$-distribution overwritable if $\wy^{\otimes k'}$ is $\uz^{\otimes k'}$-distribution overwritable. 
\end{thm}

Theorem~\ref{thm:koverwritability} shows that the set $\uz^{\otimes k}$-overwritable channels coincides with the set of $\uz^{\otimes k'}$-overwritable channels for every $k,k'$. However, Theorem~\ref{thm:kdistoverwritability} only shows that the set of $\uz^{\otimes k'}$-distribution overwritable channels is contained in the set of $\uz^{\otimes k}$-distribution overwritable channels for $k<k'$. In fact, Example~\ref{exmp:BSCZ} shows that this inclusion is strict. 

Using the multi-letter distribution overwritability condition, we now give an asymptotic exact characterization of the positivity condition. To this end, we define the following overwritability condition.
\begin{definition}[$\uz^{\otimes\infty}$-distribution overwritability]\label{def:infinitedimow} We say that $\wy$ is $\uz^{\otimes\infty}$-distribution overwritable if, for each positive integer $k$, $\wy^{\otimes k}$ is $\uz^{\otimes k}$-distribution overwritable.
\end{definition}

\begin{thm}\label{thm:infinitedim}
    $\cst=C_{\szero}$ if and only if $\wy$ is not $\uz^{\otimes\infty}$-distribution overwritable.
\end{thm}

\begin{proof}
Note that it suffices to prove that $\cst>0$ if and only if $\wy$ is not $\uz^{\otimes\infty}$-distribution overwritable, since by Theorem~\ref{thm:dichotomy}, $\cst$ equals $C_{\szero}$ whenever it is strictly positive.

To this end, first assume that $\wy$ is not $\uz^{\otimes\infty}$-distribution overwritable. Thus, there exists a positive integer $k$ such that $\wy^{\otimes k}$ is not $\uz^{\otimes k}$-distribution overwritable. Using Theorem~\ref{thm:achievability}, we can construct a positive rate stochastic authentication code for the channel $(\wy^{\otimes k},\uz^{\otimes k})$ and transmit it over the channel $(\wy,\uz)$.

Next, suppose that $\wy$ is $\uz^{\otimes\infty}$-distribution overwritable. Consider an arbitrary stochastic code of blocklength $n$ with the encoder conditional distribution $P_{X^n|M}$. Note that the adversary knows the encoding strategy but not the message. Thus, from the point of the adversary, the codeword is distributed as per the average distribution $P_{X^n}(x^n)\triangleq \frac{1}{M}\sum_{i\in[M]} P_{X^n|M}(x^n|i)$. As $\wy^{\otimes n}$ is $\uz^{\otimes n}$-distribution overwritable, for every $i'\in[M]$ there exists a conditional distribution $P_{S^n|\tilde{X}^n,Z^n}$ such that 
\begin{equation}
    \sum_{x^n,z^n,s^n}P_{X^n}(x^n) U^{\otimes n}(z^n|x^n)P_{S^n|\tilde{X}^n,Z^n}(s^n|\tilde{x}^n,z^n)W^{\otimes n}(y^n|x^n,s^n)= W^{\otimes n}(y^n|\tilde{x}^n,\szero^n)
\end{equation}
for every $\tilde{x}^n$ and $y^n$. The adversary first uniformly picks $i'\in[M]$. Consider the jamming strategy $\ProbJ_{S^n|Z^n}$ defined as
\begin{equation}
    \ProbJ_{S^n|Z^n}(s^n|z^n) = \sum_{\tilde{x}^n} P_{X^n|M}(\tilde{x}^n|i') P_{S^n|\tilde{X}^n,Z^n}(s^n|\tilde{x}^n,z^n).
\end{equation}
It follows that, under the above jamming strategy, the decoder observes an output that is distributed exactly as if the actual message is $i'$ and there is no adversary. \emph{i.e.}, $s^n=\szero^n$. Following along the lines of the proof of Theorem~\ref{thm:converse}, we conclude that the error probability is bounded from below by $1/2$ as $n\to\infty$. Thus, no positive rate of authenticated transmission is possible.
\end{proof}

\section{Examples}
\label{sec:examples}

In this section, we present a collection of examples that distinguish between our channel overwritability conditions, and begin an in-depth study of a myopic binary adversarial channel (which is concluded in Section \ref{sec:det}).

\subsection{Gaps between overwritability conditions}
Here, we establish the strict containments between overwritability conditions. We refer the reader to Table \ref{fig:capacity_summary}. First, we exhibit a channel that is $U$-overwritable, but not (oblivious) overwritable.

\begin{exmp}[$U$-overwritable but not  overwritable]\label{exmp:bitflip_Uow}
Let $\cX=\cY=\cS=\{0,1\}$, $\szero=0$, and let $Y=X\oplus S$ with probability $1$. To see that this channel is not overwritable (Def.~\ref{def:overwritability}), note that any $P_{S\mid X'}$ satisfying~\eqref{eq:def:overwritability} would need $P_{S\mid X'}(0\mid 1)=0$ (from $x=0,x'=1,y=0$) and $P_{S\mid X'}(0\mid 1)=1$ (from $x=1,x'=1,y=0$), which would be a contradiction. However, taking $U(z\mid x)$ such that $Z=X$ with probability $1$ (i.e., the adversary observes the input exactly) and setting $S=X'\oplus Z$ gives $\sum_{s,z}U(z\mid x)P_{S\mid X',Z}(s\mid x',z)W(y\mid x,s) = W(y\mid x, x'\oplus x) = \mathbbm{1}[y=x']$, so $W$ is $U$-overwritable. Hence $U$-overwritability is strictly weaker than overwritability.
\end{exmp}

In order to distinguish between single and multi-letter distribution overwritability, we introduce the following channel, which acts as a BSC without an adversary's interference, and as a $Z$-type channel when the adversary acts. The same channel gives an example of the gap between (oblivious) distribution overwritability and overwritability.

\begin{definition}[Adversarial Z/Binary-Symmetric Channel AZBSC$(\delta,\alpha)$] \label{def:azbsc} We say that the channel $W$  with  input and output alphabets $\cX=\cY=\{0,1\}$ and state space $\cS=\{\szero,s_{1,0},s_{1,1}\}$ is AZBSC$(\delta,\alpha)$ if
\begin{equation}
    \wy(y\mid x,s) = \begin{cases}
    1-\delta& s=\szero, y=x\\
    \delta & s=\szero, y\neq x\\
    0 & s=s_{1,x}, y \neq  x\\
    1 & s = s_{1,x}, y=x\\
    \alpha & s = s_{1,x\oplus 1}, y = x\\
    1-\alpha & s = s_{1,x\oplus 1}, y\neq x .
    \end{cases}
\end{equation}
See Figure~\ref{fig:azbsc} for a visual representation of the transition probabilities.
\end{definition}

\definecolor{nodecolor}{rgb}{0.30196078431372547,0.30196078431372547,1}
\begin{figure}[!ht]\centering 
  \resizebox{0.7\columnwidth}{!}{
  \begin{tikzpicture}[line cap=round,line join=round,>=triangle 45,x=1cm,y=1cm]
  \clip(-6.5,-5.06) rectangle (12.87,5.42);

  % ================== State s_0 : BSC(δ) ==================
  % Edges
  \draw [line width=1pt] (-6,2) -- (-2,2);   % 0 -> 0
  \draw [line width=1pt] (-6,-2) -- (-2,-2); % 1 -> 1
  \draw [line width=1pt] (-6,-2) -- (-2,2);  % 1 -> 0
  \draw [line width=1pt] (-6,2) -- (-2,-2);  % 0 -> 1

  % Input/output labels
  \draw (-6.60,2.25) node[anchor=north west] {\LARGE $0$};
  \draw (-6.60,-1.75) node[anchor=north west] {\LARGE $1$};
  \draw (-1.83,2.25) node[anchor=north west] {\LARGE $0$};
  \draw (-1.87,-1.75) node[anchor=north west] {\LARGE $1$};

  % Probabilities for BSC(δ)
  \draw (-4.49,2.90) node[anchor=north west] {\LARGE $1-\delta$}; % 0->0
  \draw (-4.47,-2.15) node[anchor=north west] {\LARGE $1-\delta$}; % 1->1
  \draw (-4.25,0.85) node[anchor=north west] {\LARGE $\delta$};    % 1->0
  \draw (-4.25,-0.40) node[anchor=north west] {\LARGE $\delta$};   % 0->1

  % ================== State s_{1,0} ==================
  % Edges: 0->0 (prob 1), 1->1 (prob α), 1->0 (prob 1-α)
  \draw [line width=1pt] (2-0.5,2) -- (6-0.5,2);      % 0 -> 0
  \draw [line width=1pt] (2-0.5,-2) -- (6-0.5,-2);    % 1 -> 1
  \draw [line width=1pt] (2-0.5,-2) -- (6-0.5,2);     % 1 -> 0

  % Input/output labels
  \draw (1.45-0.5,2.25) node[anchor=north west] {\LARGE $0$};
  \draw (1.45-0.5,-1.75) node[anchor=north west] {\LARGE $1$};
  \draw (6.15-0.5,2.25) node[anchor=north west] {\LARGE $0$};
  \draw (6.15-0.5,-1.75) node[anchor=north west] {\LARGE $1$};

  % Probabilities for s_{1,0}
  \draw (2.2+1,2.90) node[anchor=north west] {\LARGE $1$};          % 0->0
  \draw (3.51-0.5,-2.15) node[anchor=north west] {\LARGE $\alpha$};   % 1->1
  \draw (3.87-0.5,-.10) node[anchor=north west] {\LARGE $1-\alpha$}; % 1->0

  % ================== State s_{1,1} ==================
  % Edges: 1->1 (prob 1), 0->0 (prob α), 0->1 (prob 1-α)
  \draw [line width=1pt] (8,2) -- (12,2);      % 0 -> 0
  \draw [line width=1pt] (8,-2) -- (12,-2);    % 1 -> 1
  \draw [line width=1pt] (8,2) -- (12,-2);     % 0 -> 1

  % Input/output labels
  \draw (7.45,2.25) node[anchor=north west] {\LARGE $0$};
  \draw (7.45,-1.75) node[anchor=north west] {\LARGE $1$};
  \draw (12.15,2.25) node[anchor=north west] {\LARGE $0$};
  \draw (12.11,-1.75) node[anchor=north west] {\LARGE $1$};

  % Probabilities for s_{1,1}
  \draw (9.53,2.90) node[anchor=north west] {\LARGE $\alpha$};       % 0->0
  \draw (8.83,-.10) node[anchor=north west] {\LARGE $1-\alpha$};    % 0->1
  \draw (9.60,-2.15) node[anchor=north west] {\LARGE $1$};           % 1->1

  % ================== Labels under panels ==================
  \draw (-5.3,-3.7)  node[anchor=north west] {\huge $W(\cdot|\cdot,s_0)$};
  \draw (2.0,-3.7)   node[anchor=north west] {\huge $W(\cdot|\cdot,s_{1,0})$};
  \draw (8.5,-3.7)   node[anchor=north west] {\huge $W(\cdot|\cdot,s_{1,1})$};

  % Vertical separator (optional, as in original)
 % \draw [line width=1pt] (-0.25,3) -- (-0.25,-4);

  % ================== Nodes ==================
  \begin{scriptsize}
    % s_0
    \draw [fill=nodecolor] (-6,2) circle (2.5pt);
    \draw [fill=nodecolor] (-2,2) circle (2.5pt);
    \draw [fill=nodecolor] (-6,-2) circle (2.5pt);
    \draw [fill=nodecolor] (-2,-2) circle (2.5pt);

    % s_{1,0}
    \draw [fill=nodecolor] (2-0.5,2) circle (2.5pt);
    \draw [fill=nodecolor] (6-0.5,2) circle (2.5pt);
    \draw [fill=nodecolor] (2-0.5,-2) circle (2.5pt);
    \draw [fill=nodecolor] (6-0.5,-2) circle (2.5pt);

    % s_{1,1}
    \draw [fill=nodecolor] (8,2) circle (2.5pt);
    \draw [fill=nodecolor] (12,2) circle (2.5pt);
    \draw [fill=nodecolor] (8,-2) circle (2.5pt);
    \draw [fill=nodecolor] (12,-2) circle (2.5pt);
  \end{scriptsize}

  \end{tikzpicture}
  }
  \caption{The AZBSC($\delta,\alpha$) in states $s_0$, $s_{1,0}$, and $s_{1,1}$.}\label{fig:azbsc}
\end{figure}

In the following, we establish that the above channel is distribution overwritable, but not overwritable.

\begin{exmp}[Distribution overwritable but not overwritable]\label{ex:azbsc} Let   $0\leq \delta\leq 1/2$ and $0\leq \alpha\leq \delta$ and let $\wy_{Y|X,S}$ be the transition probability for AZBSC($\delta,\alpha)$ (as defined in Definition~\ref{def:azbsc}). We claim that $\wy$ is distribution overwritable but not overwritable.

To show that that $\wy$ is distribution overwritable, we can verify that for any input distribution $P_X$ with $P_X(0)=p$, the adversary can successfully overwrite with $x'=0$ by selecting the state as per $P_S$ with 
 
 \begin{equation}
     P_S(s)=\begin{cases} 0& s=\szero \\ \frac{1-\delta-p\alpha}{1-\alpha}& s = s_{1,0} \\  \frac{\delta-(1-p)\alpha}{1-\alpha} & s = s_{1,1}.    \end{cases}
 \end{equation} 
Note that by symmetry, the strategy for $x'=1$ case follows in a similar manner.

Next, suppose that there exists a $P_{S|X}'$ such that $W$ satisfies~\eqref{eq:def:overwritability}. Consider $x'=0$. Then, the following must hold:
 \[W(0|0,s_0)=\sum_{s} P_{S|X'}(s|0)W(0|0,s)=\sum_{s} P_{S|X'}(s|0)W(0|1,s).\]
Let \[P_{S|X'}(s|0) = \begin{cases}1-q_0-q_1& s = s_0\\ q_0& s = s_{1,0}\\q_1& s = s_{1,1}.\end{cases}\]
The above condition implies that 
\[1-\delta= (1-q_0-q_1)(1-\delta) + q_0+q_1\alpha \]
and 
\[1-\delta= (1-q_0-q_1)\delta +    q_0(1-\alpha). \]
Subtracting the second equation above from the first equation gives us 
\[q_0+q_1 = \frac{1-2\delta}{1-2\delta-\alpha} ,\] which is either greater than $1$ or smaller than $0$ for $0<\alpha<\delta<1/2$, which gives a contradiction. Hence, the given channel is not overwritable.

\end{exmp}

The same channel exhibits the gap between $U$-distribution overwritability and $U^{\otimes 2}$-distribution overwritability, as well as the corresponding oblivious adversary statement.

     \begin{exmp} \label{exmp:BSCZ} [Distribution overwritable but not 2-letter distribution overwritable] Let $0\leq \delta\leq 1/2$ and $0\leq \alpha\leq \delta$ and let $W_{\alpha,\delta}$ be the transition probability for AZBSC($\delta,\alpha)$ (as defined in Definition~\ref{def:azbsc}). We first show that distribution overwritability of $W_{\alpha,\delta}$ does not imply distribution overwritability of $W_{\alpha,\delta}^{\otimes 2}$ (with an oblivious adversary). Subsequently, we show this property for a myopic adversary.  
 Note that, for any $k\in\mathbb{N}$,  if $W_{\alpha,\delta}^{\otimes k}$ is distribution overwritable, then so is $W_{\alpha',\delta}^{\otimes k}$ for each $\alpha'<\alpha$ (See Lemma~\ref{lem:alhpastar} in the appendix for a proof). We therefore define $\alpha_k^*(\delta)$ to be the largest $\alpha$ such that $W_{\alpha,\delta}^{\otimes k}$ is distribution overwritable. 
 
First, to show that that $\wy_{\alpha,\delta}$ is distribution overwritable for every $\alpha\leq \delta$, we can verify that for any input distribution $P_X$ with $P_X(0)=p$, the adversary can successfully overwrite with $x'=0$ by selecting the state as per $P_S$ with 
 
 \begin{equation}\label{eq:ex:azbscadvstrat}
     P_S(s)=\begin{cases} 0& s=\szero \\ \frac{1-\delta-p\alpha}{1-\alpha}& s = s_{1,0} \\  \frac{\delta-(1-p)\alpha}{1-\alpha} & s = s_{1,1}.    \end{cases}
 \end{equation}

This shows that $\alpha_1^*(\delta)=\delta$. Next, we consider the two-letter  channel $W_{\alpha,\delta}^{\otimes 2}$ formed from $\wy_{\alpha,\delta}$. For this channel, we numerically obtain an upper bound on $\alpha_2^*(\delta)$. In particular, we consider a randomly chosen subset ${\hat\cP}\subseteq\cP(\cX^2)$ of two-letter input distributions. For each input distribution $P_{X_1X_2}\in\hat{\cP}$, we check whether ~\eqref{eq:distow_nonmyopic} is satisfied for some choice of $P_{S_1S_2|X_1'X_2'}$. Note that this is amounts to checking whether, for every $P_{X_1X_2}$ and $(x_1',x_2')$ the no-adversary output distribution $W_{\alpha,\delta}^{\otimes 2}(\cdot|x_1'x_2',\szero,\szero)$ lies in the convex hull of distributions $\{\hat{Q}^{(s_1,s_2)}\}$, where, 
 \begin{equation}
     \hat{Q}^{(s_1,s_2)}(y_y,y_2)\triangleq \sum_{x_1,x_2}P_{X_1X_2}(x_1,x_2){\wy(y_1,y_2|x_1,x_2,s_1,s_2)}.
 \end{equation}
 If $W_{\alpha,\delta}^{\otimes 2}(\cdot|x_1'x_2',\szero,\szero)$ does not lie in the convex hull above for some $P_{X_1X_2}\in\hat{\cP}$, the corresponding $W_{\alpha,\delta}^{\otimes 2}$ is not distribution overwritable, thus giving us an upper bound on $\alpha_2^*$ is not distribution overwritable. Figure~\ref{fig:alphaplot} shows this plot.
\begin{figure}
\centering
\begin{tikzpicture}
  \begin{axis}[
    xlabel={$\delta$},
    ylabel={$\alpha$},
    xmin=0, xmax=0.5,
    ymin=0, ymax=0.5,
    grid=both,
    grid style={line width=.1pt, draw=gray!20},
    legend cell align={left},
    legend pos=north west,
    axis on top
  ]

    % First curve: Shading (Hatching)
    \addplot[
      draw=none,
      pattern=north east lines,
      pattern color=red!40!blue!10,
      smooth,
      forget plot
    ] table[
      x=delta,
      y=alpha1,
      col sep=comma
    ] {plot.csv} \closedcycle;

    % First curve: Line and Markers
    \addplot[
      color=blue!80!red,
      thin,
      smooth,
      mark=*,
      mark size=0.5pt,
      mark options={solid, fill=blue!20, draw=blue!90}
    ] table[
      x=delta,
      y=alpha1,
      col sep=comma
    ] {plot.csv};
    \addlegendentry{$\alpha^*_1(\delta)$}

    % Second curve: Shading (Dots)
    \addplot[
      draw=none,
      pattern=dots,
      pattern color=yellow!20!red!20,
      smooth,
      forget plot
    ] table[
      x=delta,
      y=alpha2,
      col sep=comma
    ] {plot.csv} \closedcycle;

    % Second curve: Line and Markers
    \addplot[
      color=yellow!20!red!80,
      thin,
      smooth,
      mark=*,
      mark size=0.5pt,
      mark options={solid, fill=red!20, draw=yellow!20!red!90}
    ] table[
      x=delta,
      y=alpha2,
      col sep=comma
    ] {plot.csv};
    \addlegendentry{$\widehat{\alpha^*_2}(\delta)$}

  \end{axis}
\end{tikzpicture}
\caption{The figure shows that the set of two-letter distribution overwritable channels is strictly contained in the set of single-letter distribution overwritable channels. In the above, $\alpha^*_1(\delta)$, and $\widehat{\alpha^*_2}(\delta)$, respectively denote the largest $\alpha$ such that $W_{\alpha,\delta}$ is distribution overwritable and an upper bound  on the largest $\alpha$ such that $W_{\alpha,\delta}^{\otimes 2}$ is distribution overwritable.}\label{fig:alphaplot}
\end{figure}

Next, we consider the myopic setting with the following channel to the adversary. Let $\cZ=\cX\cup\{E\}$, where, we let $E\notin \cX$ without loss of generality. Consider $U_\epsilon$ defined as
 \begin{align*}
 U_\epsilon(z|x)=\begin{cases}
 \epsilon & z=x\\
 1-\epsilon & z=E.
 \end{cases}
 \end{align*}

By Lemma~\ref{lem:u_epsilon_k}, because the single-letter channel $W_{\alpha,\delta}$ is distribution overwritable, there exists a threshold $\epsilon_1 > 0$ such that $W_{\alpha,\delta}$ is $U_\epsilon$-distribution overwritable for all $\epsilon \in (0, \epsilon_1]$.

Conversely, we analyze the two-letter channel under the adversary's joint observation channel $U_\epsilon^{\otimes 2}$. By Lemma~\ref{lem:u_epsilon_k}, because $W_{\alpha,\delta}^{\otimes 2}$ is not distribution overwritable, it cannot be $U_\epsilon^{\otimes 2}$-distribution overwritable for arbitrarily small $\epsilon$. Thus, there exists a threshold $\epsilon_2 > 0$ such that for all $\epsilon \in (0, \epsilon_2]$, $W_{\alpha,\delta}^{\otimes 2}$ is not $U_\epsilon^{\otimes 2}$-distribution overwritable. Letting $\epsilon = \min(\epsilon_1, \epsilon_2)$, we obtain a channel $U_\epsilon$ under which the single-letter channel $W_{\alpha,\delta}$ is $U_\epsilon$-distribution overwritable, but the two-letter channel $W_{\alpha,\delta}^{\otimes 2}$ is not $U_\epsilon^{\otimes 2}$-distribution overwritable. 

This establishes that for a myopic adversary, the inclusion established in Theorem~\ref{thm:kdistoverwritability} is strict.
 \end{exmp}

Finally, both $U$-overwritability and $U$-distribution overwritability are stronger conditions than $I$-overwritability. The  following illustrates that $U$-distribution overwritability is a {strictly} stronger condition than $I$-overwritability.

\begin{exmp}[$I$-overwritable but not $U$-distribution\ overwritable]\label{exmp:bitflip_IvsUdistow}
Continuing with the bit-flip channel of Example~\ref{exmp:bitflip_Uow}, where $W(y\mid x,s) = \mathbbm{1}[y = x\oplus s]$, we show that $W$
is $I$-overwritable but is not $U_\epsilon$-distribution overwritable
for any $\epsilon\in(0,1)$, where $U_\epsilon$ is the erasure channel of
Lemma~\ref{lem:u_epsilon_k}. Combined with Corollary~\ref{cor:implies_I-over},
this shows that the implication ``$U$-distribution overwritable
$\Rightarrow I$-overwritable'' is strict.

\smallskip\noindent\textit{$I$-overwritability.} Take $P_{S\mid X',X}(s\mid x',x) = \mathbbm{1}[s = x'\oplus x]$. Then $\sum_s P_{S\mid X',X}(s\mid x',x)\,W(y\mid x,s) = W(y\mid x,\, x'\oplus x) = \mathbbm{1}[y = x']$, matching $W(y\mid x',\szero)$.

\smallskip\noindent\textit{Failure of $U_\epsilon$-distribution overwritability for $\epsilon<1$.} Under $U_\epsilon$, every input distribution is in its own myopic-indistinguishability class (see the proof of Lemma~\ref{lem:u_epsilon_k}). Take $\Prob_X$ uniform on $\{0,1\}$. Substituting $W(y\mid x,s) = \mathbbm{1}[s = x\oplus y]$ into~\eqref{eq:def:udistow} and using $\sum_s P_{S\mid X',Z}(s\mid x',E) = 1$, the left-hand side becomes
\begin{equation*}
   \tfrac{\epsilon}{2}\bigl[P_{S\mid X',Z}(y\mid x',0) + P_{S\mid X',Z}(1\oplus y\mid x',1)\bigr] + \tfrac{1-\epsilon}{2}.
\end{equation*}
Setting this equal to $W(y\mid x',\szero) = \mathbbm{1}[y = x']$ at $(x',y) = (0,0)$ yields
\begin{equation*}
   P_{S\mid X',Z}(0\mid 0,0) + P_{S\mid X',Z}(1\mid 0,1) \;=\; \tfrac{1+\epsilon}{\epsilon},
\end{equation*}
which exceeds $2$ whenever $\epsilon<1$ and is therefore unachievable since each probability is at most $1$. Hence $W$ is not $U_\epsilon$-distribution overwritable for any $\epsilon\in(0,1)$.
\end{exmp}

\subsection{Myopic Binary Adversarial Channel}
\label{subsection:MBAC}

Next, we examine in detail a binary model (depicted in Figure~\ref{mbac_channel_model}). In this model, the adversary views the transmitted codeword through a BSC$(q)$, and decides on a binary state sequence $\bfs$, with $\bfs_0$ being the all zero vector, which is added (modulo 2) to the transmission $\x$. The sequence $\x\oplus\bfs$ is then transmitted across a BSC$(p)$. We call this the myopic binary adversarial channel with parameters $p$ and $q$, and denote it by MBAC$_{p,q}$.

{\footnotesize\begin{figure*}[htb!]
\centering
\tikzstyle{block} = [draw, rectangle, minimum height=30pt, minimum width=70pt, text centered]
\tikzstyle{chanblock} = [draw,thick, rectangle, minimum height=56, minimum width=70, text centered]
\tikzstyle{circ} = [draw, circle, text centered]
\begin{tikzpicture}[every node/.style={transform shape},scale=.75]
    %all the blocks
    \node[circ] (add) at (4,0) {$\oplus$};
    \node[block] (enc) at (0,0) {$\begin{array}{c} \text{Transmitter} \\ F \end{array}$};
    \node[block] (dec) at (10.5,0) {$\begin{array}{c} \text{Receiver} \\ \phi \end{array}$};
    \node[block] (adv) at (4,2.5) {$\begin{array}{c} \text{Adversary} \\ J(\bfs\mid \z) \end{array}$};
    \foreach \px/\py/\pt/\pl in {0/2.5/q/adv_chan,7/0/p/chan}    
    {
        \node[chanblock] (\pl) at (\px,\py) {};
        \node at (\px-.75,\py+0.5) {$0$};
        \node at (\px-.75,\py-0.5) {$1$};
        \node at (\px+.75,\py+0.5) {$0$};
        \node at (\px+.75,\py-0.5) {$1$};
        \node at (\px-.35,\py) {$\pt$};
        \draw[->,thick] (\px-.55,\py+0.5) -- node[midway,above] {$1-\pt$} (\px+.55,\py+0.5);
        \draw[->,thick] (\px-.55,\py-0.5) -- (\px+.55,\py+0.5);
        \draw[->,thick] (\px-.55,\py+0.5) -- (\px+.55,\py-0.5);
        \draw[->,thick] (\px-.55,\py-0.5) -- node[midway,below] {$1-\pt$} (\px+.55,\py-0.5);
    }

    %all the arrows
    \draw[->,thick] (-2,0) node[above] {$i \in [M]$} -- (enc); 
    \draw[->,thick,dashed] (enc)  -- node[left,midway] {$\x_{i}$}  (adv_chan); 
    \draw[->,thick] (enc) -- node[above,midway] {$\x_i$} (add); 
    \draw[->,thick,dashed] (adv_chan) -- node[above,midway] {$\z$} (adv);
    \draw[->,thick] (adv) -- node[right,midway] {$\bfs$} (add);
    \draw[->,thick] (add) -- node[above,midway] {$\x_i \oplus \bfs$} (chan);
    \draw[->,thick] (chan) -- node[above,midway] {$\y$} (dec);
    \draw[->,thick] (dec) -- node[above right]{$\hat{i}\in [M]\cup \{0\}$} (12.5,0);
    
\end{tikzpicture}
\caption{\label{mbac_channel_model} An illustration of the MBAC$_{p,q}$; note that the transmitted sequence $\mathbf{x}_{i}$ may be a random variable if the encoder $F$ is stochastic. The dashed lines indicate that the link from the transmitter to adversary is non-causal.
}
\end{figure*}}

To begin to consider the channel, suppose $q>0$ and $p=0.5$. In this case, the  MBAC$_{p,q}$ has non-adversarial capacity $1-H(0.5)=0$, and thus its authentication capacity is also zero. On the other hand, when $q=0$, the channel is $U$-overwritable (recall that $U$ in this context is the BSC$(0)$). Equivalently, the channel is $I$-overwritable in this case: indeed, choose $P_{S |X',Z=X}$ to be deterministic such that $s=x'\oplus x$. 

As we consider other values of $p$ and $q$, we note that if the channel is $U$-overwritable for some $q$, the results of Section~\ref{section:compare_adversaries} allow us to conclude that the channel is $U$-overwritable for all smaller $q$ as well:

\begin{exmp}[Monotonicity of $U$-overwritability along a BSC chain]\label{exmp:UowBSCchain}
Extending Example \ref{example:totally_ordered_BSC}, we see that if a channel $W_{Y|X,S}$ is $U$-overwritable for $U_{Z|X}=\text{BSC}(q_1)$, then it is $U$-overwritable for $U_{Z|X}=$BSC($q_2$) for all $q_2\leq q_1 \le 1/2$. As was shown in Theorem \ref{thm:converse}, $U$-overwritability implies zero authentication capacity, so the authentication capacity of this channel would be equal to zero whenever the channel to the adversary is a BSC($q_2$) for $q_2\leq q_1 \le 1/2$. See Figure \ref{figure:BSC_example}. 
\begin{figure}
    \centering
    \begin{tikzpicture}
        \draw[fill=black] (0,0) circle (.05);
        \node at (-.1,0) [left] {BSC($0.5$)};
        \node at (-.1,-0.5) [left] {(oblivious)};
        \draw[fill=black] (4,0) circle (.05);
        \node at (4.1,0) [right] {BSC($0$)};
        \node at (4.1,-0.5) [right] {(omniscient)};
        \draw[fill=black] (1.5,0) circle (.05);
        \node at (1.5,0) [above] {BSC($q_1$)};
        \draw[thick] (0,0) -- (4,0);
        \draw [thick,dashed,decorate,decoration={brace,amplitude=10pt,mirror},xshift=0.4pt,yshift=-0.4pt](1.5,0) -- (4,0) node[black,midway,yshift=-0.6cm] {$\cst=0$};
    \end{tikzpicture}
    \caption{The binary symmetric channels are totally ordered by stochastic degradation. If $W_{Y|X,S}$ is $U$-overwritable for $U_{Z|X}=$BSC($q_1$), then the authentication capacity $\cst$ is equal to zero for that $U_{Z|X}$ and all binary symmetric channels with smaller crossover probability.}
    \label{figure:BSC_example}
\end{figure}
\end{exmp}

In the case of the MBAC$_{p,q}$, however, we show in the following lemma that the channel is not $U$-distribution overwritable as long as the honest channel is not useless (i.e., $p\ne 0.5$), and the channel to the adversary is not perfect (i.e., $0<q<1$).

\begin{lem}\label{lem:MBAC-U-dist-over}
The MBAC$_{p,q}$ is not $U$-distribution overwritable for any $0<q<1$ and $p\ne 1/2$.
\end{lem}

\begin{proof}
Due to the symmetry of the channel, we may assume that $p\le 1/2$, so by assumption that $p \ne 1/2$ we may take any $p<1/2$. Suppose that this MBAC$_{p,q}$ is $U$-distribution overwritable. Consider the equivalence class $\cP$ of input distributions that are myopically indistinguishable from the $\text{Ber}(1/2)$ distribution. Certainly $\text{Ber}(1/2)\in \cP$, so there must be an adversarial strategy $P_{S|X',Z}$ such that
\begin{equation}\label{mbac_u_dist_condition}
    \sum_{x,s,z} \frac{1}{2} U(z|x) P_{S|X',Z}(s|x',z) W(y|x,s)=W(y|x',s_0),\ \forall x',y.
\end{equation}
Consider $x'=y=0$. For notation convenience, for $z=0,1$, define $t_z=P_{S|X',Z}(1|0,z)$, so also $\bar{t}_z=P_{S|X',Z}(0|0,z)$. The summation in \eqref{mbac_u_dist_condition} has 8 terms, which are listed in the following table:
\begin{equation}
    \begin{array}{ccccccc}
x & z & s & P_X(x) & U(z|x) & P_{S|X',Z}(s|0,z) & W(0|x,s) \\ \hline
0 & 0 & 0 & 1/2 & \bar{q} & \bar{t}_0 & \bar{p}\\
0 & 0 & 1 & 1/2 & \bar{q} & t_0 & p \\
0 & 1 & 0 & 1/2 & q & t_1 & \bar{p} \\
0 & 1 & 1 & 1/2 & q & \bar{t}_1 & p \\
1 & 0 & 0 & 1/2 & q & \bar{t}_0 & p\\
1 & 0 & 1 & 1/2 & q & t_0 & \bar{p}\\
1 & 1 & 0 & 1/2 & \bar{q} & t_1 & p\\
1 & 1 & 1 & 1/2 & \bar{q} & \bar{t}_1 & \bar{p}.
    \end{array}
\end{equation}
Adding these together and combining similar terms, we find
\begin{align}
\sum_{x,s,z} \frac{1}{2} U(z|x) P_{S|X',Z}(s|0,z) W(0|x,s)
    & = \frac{1}{2}\left[
    (\bar{t}_0+\bar{t}_1)(\bar{q}\bar{p} + q p)+ (t_0+t_1)(\bar{q}p+q\bar{p})
    \right].
\end{align}
To be $U$-distribution overwritable, this summation must match $W(0|0,s_0)=\bar{p}$. Noting that $t_0,t_1\in[0,1]$, we have
\begin{align}
    \sum_{x,s,z} \frac{1}{2} U(z|x) P_{S|X',Z}(s|0,z) W(0|x,s) \le \max\{\bar{q}\bar{p}+qp,\ \bar{q}p+q\bar{p}\}
    = \bar{p}+\max\{q (2p-1),\ \bar{q}(2p-1)\}<\bar{p}
\end{align}
where the last step follows from the assumptions that $0<q<1$, and that $p<1/2$.
\end{proof}
The following is an immediate corollary of the previous lemma.
\begin{co}
    \label{lem:MBAC-U-over}
The MBAC$_{p,q}$ is not $U$-overwritable provided $0<q<1$ and $p\neq 1/2$. 
\end{co}

By Theorem~\ref{thm:converse} and the comments above, if $q=0$ or $p=1/2$, the authentication capacity of the MBAC$_{p,q}$ is zero, regardless of the type of encoder. However, allowing a stochastic encoder, we can achieve the non-adversarial capacity of the channel as long as $q>0$ and $p<0.5$: that is, as long as the channel to the adversary is not noiseless, and the honest channel is not useless.

\begin{thm}
\label{theorem:MBACp>q}
Suppose $0< q<1$ and $p \neq 1/2$. For the MBAC$_{p,q}$, $\cst=1-H(p)$.
\end{thm}

\begin{proof}
    We see that the channel is not $U^{\otimes \infty}$-distribution overwritable from the fact that it is not $U$-distribution overwritable. Then Theorem \ref{thm:infinitedim} gives us that $\cst=C_{s_{0}}$. Since the honest channel is a BSC($p$), we have that $\cst=1-H(p)$.
\end{proof}

In fact, $\cdet = C_{s_{0}}$ for this channel whenever $0\leq p<q\leq 1/2$. We examine deterministic authentication capacity of the MBAC$_{p,q}$ in the next section.

\section{Deterministic Capacity}
\label{sec:det}

We now turn our attention to deterministic authentication codes. Our first results shows that, when the channel is not $I$-overwritable (see \textup{(Definition~\ref{definition:I-over})}), the deterministic authentication capacity is positive.

\begin{thm}[Non-$I$-overwritability implies positive deterministic capacity]
\label{theorem:nonIover_positive}
If $W_{Y|X,S}$ is not $I$-overwritable, then $\cdet>0$. 
\end{thm}

\begin{proof}
We first note that, by Theorem~\ref{theorem:stoch_deg_posititivity}, it suffices to prove this result for an omniscient adversary, \emph{i.e.}, the channel when $U$ is the identity channel (\emph{i.e.}, $U=X$). In the following, we assume that the channel $U$ is the identity channel.

Recall from Definition~\ref{definition:I-over} that $W_{Y|X,S}$ is $I$-overwritable if and only if, for every input $x$ and adversarial choice $x'$, there is a state distribution $P_{S|X',X}(\cdot\mid x',x)$ with $\sum_{s}P_{S|X',X}(s\mid x',x)\wy(\cdot\mid x,s)=\wy(\cdot\mid x',\szero)$; equivalently, $\wy(\cdot\mid x',\szero)\in\cK_{x}$, where $\cK_{x}:=\mathrm{conv}\{\wy(\cdot\mid x,s):s\in\cS\}$ is the set of output distributions the adversary can produce from input $x$. Since $W_{Y|X,S}$ is not $I$-overwritable, there is an input $a$ and a choice $a'\neq a$ for which $\wy(\cdot\mid a',\szero)\notin\cK_{a}$. The set $\cK_{a}$ is a compact convex subset of the simplex over $\cY$, so
\begin{equation}
\label{eq:nonIover:margin}
\mu:=\min_{\nu\in\cK_{a}}\mathbb{V}\big(\nu,\wy(\cdot\mid a',\szero)\big)>0.
\end{equation}
\emph{Codebook:} Fix a block length $k$, to be chosen below, and write $\mathsf{A}:=a^{k}$ and $\mathsf{B}:=(a')^{k}$ for the two length-$k$ blocks in $\cX^{k}$. Fix also $\delta\in(0,1/4)$. By a Gilbert--Varshamov argument over the constant-composition set, for all sufficiently small $\delta$ there is a constant $R'=R'(\delta)>0$ such that, for all large $L$, there exists a code $\C\subseteq\{\mathsf{A},\mathsf{B}\}^{L}$ in which every codeword has exactly $\lfloor L/2\rfloor$ blocks equal to $\mathsf{A}$, any two codewords differ in at least $2\delta L$ blocks, and $|\C|\geq 2^{LR'}$~\cite{CsiszarK:2011}. For a message $i\in[N]$, the encoder transmits the concatenation $\x_{i}\in\cX^{n}$ of the $L$ symbols from the above code, such that $n=Lk$. Let $\cB_{1},\dots,\cB_{L}\subseteq[n]$ with $\cB_b=\{bk+1,\ldots (b+1)k\}$ denote the coordinate blocks derived from this concatenation and let $c_{b}(i)\in\{a,a'\}$ denote the symbol that $\x_{i}$ carries on block $\cB_{b}$, and let $\hat{P}_{\y}(\cdot\mid\cB_{b})$ denote the empirical distribution of $\y$ over the coordinates in $\cB_{b}$.

\emph{Decoder:} Set $\tau:=\mu/3$. Given the received word $\y$, the receiver decodes to the message $i\in[N]$ for which $\mathbb{V}\big(\hat{P}_{\y}(\cdot\mid\cB_{b}),\wy(\cdot\mid c_{b}(i),\szero)\big)\leq\tau$ on every block $\cB_{b}$, provided such an $i$ is unique; otherwise, it declares adversarial interference by outputting $\reject$.

\emph{No-adversary error:} Under the no-adversary state $\bfs=\bfs_{0}$, the $k$ coordinates of block $\cB_{b}$ of the transmitted $\x_{i}$ are independent draws from $\wy(\cdot\mid c_{b}(i),\szero)$. Applying the Chernoff bound to each of the $|\cY|$ output letters, we have $\mathbb{V}(\hat{P}_{\y}(\cdot\mid\cB_{b}),\wy(\cdot\mid c_{b}(i),\szero))>\tau$ on a given block with probability at most $2|\cY|2^{-2k(\tau/|\cY|)^{2}}$, so the transmitted message passes every block test except with probability $L2^{-\Omega(k\tau^{2})}$. On the other hand, any competing message $j\neq i$ differs from $i$ on at least $2\delta L$ blocks, and on each such block the no-adversary output distributions of $c_{b}(i)$ and $c_{b}(j)$ are at distance at least $\mu$ by~\eqref{eq:nonIover:margin}; hence $\hat{P}_{\y}(\cdot\mid\cB_{b})$ lies at distance more than $\tau$ from $\wy(\cdot\mid c_{b}(j),\szero)$, and $j$ fails its test. We conclude that $e(i,J_{\bfs_{0}})\leq L2^{-\Omega(k\tau^{2})}$, which vanishes as $L\to\infty$.

\emph{Error under an active adversary:} Now fix a message $i$ and an arbitrary omniscient state sequence $\bfs\in\cS^{n}$. The decoder accepts a wrong message, \ie, $\y\in\phi^{-1}(\{i,\reject\})^{c}$, only if some $j\neq i$ passes every block test. Identifying $\mathsf{A}$ and $\mathsf{B}$ with $0$ and $1$ makes $\C$ a binary code, and the blocks on which $\x_{i}$ and $\x_{j}$ differ form the support of the modulo-$2$ sum of their codewords; by the minimum distance of $\C$ there are at least $2\delta L$ of them. Since the two codewords have the same number of $\mathsf{A}$-blocks, this support splits evenly between blocks carrying $(\mathsf{A},\mathsf{B})$ and blocks carrying $(\mathsf{B},\mathsf{A})$, so on at least $\delta L$ blocks $\x_{i}$ carries $\mathsf{A}=a^{k}$ while $\x_{j}$ carries $\mathsf{B}=(a')^{k}$. Fix such a block $\cB_{b}$. The transmitted symbols on $\cB_{b}$ are all $a$, so the conditional mean of $\hat{P}_{\y}(\cdot\mid\cB_{b})$ given $\bfs$ is $\bar{\nu}_{b}:=\frac{1}{k}\sum_{t\in\cB_{b}}\wy(\cdot\mid a,s_{t})\in\cK_{a}$, and hence $\mathbb{V}(\bar{\nu}_{b},\wy(\cdot\mid a',\szero))\geq\mu$ by~\eqref{eq:nonIover:margin}. For $j$ to pass the test on $\cB_{b}$, we would need $\mathbb{V}(\hat{P}_{\y}(\cdot\mid\cB_{b}),\wy(\cdot\mid a',\szero))\leq\tau$, and therefore a deviation of at least $\mu-\tau=2\mu/3$ from $\bar{\nu}_{b}$; by the Chernoff bound, this has probability at most $2|\cY|2^{-2k(2\mu/(3|\cY|))^{2}}=:2^{-k\beta}$ with $\beta=\Omega(\mu^{2})$. Conditioned on $\bfs$ the blocks are independent, so $j$ passes the test on all of these at least $\delta L$ blocks with probability at most $2^{-\delta Lk\beta}$. Taking a union bound over the at most $2^{LR'}$ messages, we obtain
\begin{equation}
\label{eq:nonIover:err}
\wy\big(\phi^{-1}(\{i,\reject\})^{c}\,\big|\,\x_{i},\bfs\big)\leq 2^{LR'}2^{-\delta Lk\beta}=2^{-L(\delta k\beta-R')},
\end{equation}
uniformly in $i$ and $\bfs$. Choosing $k>R'/(\delta\beta)$ makes the exponent positive, so this probability vanishes as $L\to\infty$, uniformly over all omniscient, and hence all myopic, adversaries.

\emph{Rate:} Combining the two error bounds, we have $\sup_{J}e(J)\to0$, while the rate satisfies $\frac{1}{n}\log_{2}|\C|\geq R'/k>0$. We conclude that $\cdet\geq R'/k>0$.

\end{proof}

Such a channel can achieve $\cdet=C_{s_{0}}$, as is demonstrated in the following example.

\begin{exmp}[Not $I$-overwritable with $\cdet=C_{s_{0}}$]\label{exmp:BECcounterex}
Consider the following AVC: $\mathcal{X}=\F_{2}$, $\mathcal{Y}=\F_{2}\cup \{ \varepsilon\}$, $\mathcal{S}=\{0,1,s_{0}\}$, and $W(y\mid x, s)$ is a binary erasure channel (BEC) operating on $x$ with erasure probability $s$ if $s=0,1$. If $s=s_{0}$, we define the erasure probability to be equal to $p$.

First, we observe that if the channel is $I$-overwritable, then $p=1$: indeed, if there exists $P_{S|X',X}$ such that 
\[\sum_{s}P_{S|X',X}(s\mid x',x)W(y\mid x,s)=W(y\mid x',s_{0})\] 
for all $x,x',y$, then, letting $x'=y=0$ and $x=1$, we have $0=1-p$; we conclude $p=1$. Thus, for all $p<1$, the channel is not $I$-overwritable. Recall that it cannot then be $U$-overwritable for any channel $U_{Z|X}$ to the adversary.

In fact, we can show that the {deterministic} authentication capacity is equal to $C_{\szero}=1-p$ for any value of $p$ and any channel to the adversary. The converse follows from the non-adversarial capacity of the underlying BEC($p$). To see achievability, choose any deterministic code family designed for a BEC($p$). The receiver will observe a channel output consisting of reliable $0$'s and $1$'s and (possibly) some erasures. We call an observed output sequence \textit{consistent} with a codeword of the code if the output and codeword match on all non-erased bits. We decode as follows:
\begin{enumerate}[(1)]
\item If there is a single codeword consistent with the observed sequence, decode to that codeword. 
\item If there is more than one codeword consistent with the observed sequence, declare adversarial interference.
\end{enumerate}

The only error that may occur is when the adversary is not present, and we declare an error via step (2). However, this would constitute a regular decoding error in our code for the BEC($p$). 
\end{exmp}

Next, we show that, with deterministic encoding, when the adversary's channel is less noisy than the receiver's the adversary can decode the transmitted message and cause maximal damage, \emph{i.e.}, force the authentication capacity to be equal to zero.

Recall that the encoder $f$ of a $(2^{nR},n)$ authentication code is deterministic if each message $i\in [2^{nR}]$ has a single associated codeword $f(i):=\x_{i}\in \mathcal{X}^{n}$ that will be transmitted across the channel $W_{Y|X,S}$. That is, if $f$ is a function from $[2^{nR}]$ to $\mathcal{X}^{n}$. For deterministic encoders, knowing the message is synonymous with knowing the transmitted sequence. Thus, in this scenario, if the adversary can reliably decode the message, we can consider them as an essentially omniscient adversary. We show for the set of deterministic encoders that if the no-adversary channel between legitimate users is {stochastically degraded} (see Definition~\ref{def:stoch_degraded}) with respect to the channel to the adversary, then $I$-overwritability guarantees $\cdet=0$.

Theorem~\ref{theorem:nonIover_positive} gives a general sufficient condition for the positivity of the authentication capacity with deterministic codes. The following result shows that this condition is also necessary when the channel $W_{Y|X, S=s_{0}}$  is stochastically degraded with respect to $U_{Z|X}$.

\begin{thm}
\label{theorem:essentially_omniscient}
Suppose that $W_{Y|X, S=s_{0}}$ is stochastically degraded with respect to $U_{Z|X}$.
Then,  $\cdet=0$ if and only if $W_{Y|X,S}$ is $I$-overwritable,.
\end{thm}

\begin{proof} When $W_{Y|X,S}$ is not $I$-overwritable, Theorem~\ref{theorem:nonIover_positive} shows that  $\cdet>0$, proving the reverse direction of the claim. Now, to prove the forward direction, suppose that $W_{Y|X,S}$ is $I$-overwritable, and let $P_{S|X',X}$ be the distribution guaranteed by $I$-overwritability. Because $W_{Y|X, S=s_{0}}$ is stochastically degraded with respect to $U_{Z|X}$, there exists $P_{Y|Z}$ such that for all $x,y$, we have $\sum_{z} U(z|x) P_{Y|Z}(y|z)=W(y|x,s_0)$. Consider a sequence of $(2^{nR}, n)$ authentication codes with $R>0$, and let $M:=2^{nR}$. With an abuse of notation, we will let $\phi(\y)$ denote either a message or a codeword; they are in one-to-one correspondence due to the assumption of a deterministic decoder. Let \[\tilde{J}(\bfs \mid \z)=\frac{1}{M}\sum_{j=1}^{M}\sum_{\y}P_{S|X',X}(\bfs \mid \x_{j}, \phi(\y))P_{Y|Z}(\y\mid \z),\] 
where $\phi$ is the non-adversarial decoding rule of the legitimate receiver.  Then $e(\tilde{J}) $ is lower bounded by
\begin{align}
e(\tilde{J}) &\geq 
\frac{1}{M^{2}}\sum_{i,j,\bfs,\z,\y}U(\z \mid \x_{i}) P_{S|X',X}(\bfs \mid \x_{j}, \phi(\y)) \nonumber\\
& \quad\cdot P_{Y|Z}(\y\mid \z)W(\phi^{-1}(\{i,0\})^{c} \mid \x_{i}, \bfs)\\
&=\frac{1}{M^{2}}\sum_{i,j,\bfs,\y} W(\y\mid \x_{i},\bfs_{0}) P_{S|X',X}(\bfs \mid \x_{j}, \phi(\y))W(\phi^{-1}(\{i,0\})^{c} \mid \x_{i}, \bfs)\label{eqn:stoch_deg_cond}\\
&\geq \frac{1}{M^{2}}\sum_{\substack{i,j,\bfs, \y \\ \mid \phi(\y)= i}}  W(\y\mid \x_{i},\bfs_{0}) P_{S|X',X}(\bfs \mid \x_{j}, \phi(\y)) W(\phi^{-1}(\{i,0\})^{c} \mid \x_{i}, \bfs) \\
&= \frac{1}{M^{2}}\sum_{i,j,\bfs} \left[1-e(i,J_{\bfs_{0}})\right] P_{S|X',X}(\bfs \mid \x_{j}, \x_{i})W(\phi^{-1}(\{i,0\})^{c} \mid \x_{i}, \bfs)\label{equation:deterministic-necessary} \\
&=\frac{1}{M^{2}}\sum_{i,j}\left[1-e(i,J_{\bfs_{0}})\right]W(\phi^{-1}(\{i,0\})^{c} \mid \x_{j}, \bfs_{0}) \label{equation:i-over}\\
&\geq \frac{1}{M}\sum_{i}\left[1-e(i,J_{\bfs_{0}})\right]\left(\frac{1}{M}\sum_{j\neq i}W(\phi^{-1}(j) \mid \x_{j}, \bfs_{0})\right)  \\
&\geq \frac{1}{M}\sum_{i}\left[1-e(i,J_{\bfs_{0}})\right]\left(\frac{1}{M}\sum_{j\neq i}\left[1-e(j,J_{\bfs_{0}})\right]\right)  \\
&\geq \left[1-e(J_{\bfs_{0}})\right]^{2} - \frac{1}{M}
\end{align}
where \eqref{eqn:stoch_deg_cond} follows from stochastic degradation, and \eqref{equation:i-over} follows from $I$-overwritability. Thus, $\sup_{J}e(J)> \left[1-e(J_{\bfs_{0}})\right]^{2} -\frac{1}{M} \geq \left[1-\sup_{J}e(J))\right]^{2} -\frac{1}{M}$.  
Solving, we conclude that $\sup_{J}e(J)$ is bounded away from zero as $n\to \infty$, and thus that $\cdet=0$.
\end{proof}

Suppose that the non-adversarial channel between legitimate users is the worse channel in an MBAC($p,q$). Theorem \ref{theorem:MBACp>q} states that, provided the channel to the adversary is not completely noiseless, the authentication capacity in this case is equal to the non-adversarial communication capacity over a BSC($p$) if the encoder is allowed to be stochastic. However, when the encoder is forced to be deterministic, the capacity drops to zero:

\begin{thm}
\label{thm:MBAC_deterministic_q<p}
If $0\leq q\leq p \leq 1/2$, and the encoder is deterministic, $\cdet=0$.
\end{thm}

\begin{proof}
It is straightforward to show that if $0\leq q\leq p\leq 1/2$, 
$U_{Z|X}$ is stochastically degraded with respect to $W_{Y|X, S=s_{0}}$
and the channel is $I$-overwritable: to see $I$-overwritability, choose $P_{S |X',X}$ to be deterministic such that $s=x'\oplus x$. Then, by Theorem \ref{theorem:essentially_omniscient}, $\cdet=0$.
\end{proof}

Next, we turn to the case in which the channel to the adversary is stochastically degraded with respect to the non-adversarial channel $W_{Y|X,S=s_{0}}$ between legitimate users. That is, $0\leq p<q\leq 1/2$. 
By Lemma~\ref{lem:MBAC-U-dist-over}, 
the channel is never BSC($q$)-disttribution overwritable (and thus, not BSC($q$)-overwritable; see Corollary~\ref{lem:MBAC-U-over}). In fact, we show that, not only is the authentication capacity nonzero when $q>p$,  surprisingly, we can achieve the non-adversarial capacity with a deterministic encoder.

\begin{thm}
\label{theorem:simple_case}
If $0\leq p<q\leq 1/2$, the authentication capacity $\cst$ is equal to the non-adversarial capacity $C_{\szero}=1-H(p)$. Moreover, this rate can be achieved with a deterministic encoder. That is, $\cdet = C_{\szero}$.
\end{thm}

\begin{proof}[Proof Sketch]
For the converse, notice that since any authentication code must also be an error-correcting code for the underlying non-adversarial channel, we have $\cdet\leq C_{\text{BSC($p$)}} =1-H(p)$.

The proof of achievability, given in full in Appendix \ref{app:theorem:simple_case}, shows that even if the adversary is allowed to know the exact error pattern that will be introduced over the main BSC($p$) and the exact distance from the transmitted codeword to their noisy version, we can reliably authenticate with a deterministic code family of rate approaching $1-H(p)$. More specifically, letting $R$ approach $1-H(p)$ while being strictly larger than $1-H(q)$, we show that we can select $M=2^{nR}$ codewords from the type class $\tau_{X}$ where $X\sim$ Bernoulli(1/2) such that (1) for a word $\mathbf{z}$ received by the adversary over the BSC($q$), there are a large number of codewords at distance approximately $qn$ from $\mathbf{z}$, and (2) few codewords can be reliably nudged to another valid codeword by a fixed adversarial choice of state. We show that the proportion of the set of codewords at the right distance from $\mathbf{z}$ that can be moved to another codeword using any particular state sequence is vanishingly small; as a result, the adversary cannot reliably convince the receiver that a different valid codeword was sent.
\end{proof}

\begin{remark}
\label{remark:model_diff}
The authors of \cite{DeyJL:2019} examine communication capacity under the following model: the adversary views the transmitted codeword through a BSC$(q)$, and decides on a state sequence $\bfs$ such that $\|\bfs\|\leq tn$ for fixed parameter $t$ (to avoid confusion, we use ``$t$'' rather than ``$p$'' as used in \cite{DeyJL:2019}). This state sequence is added to $\x$ and sent noiselessly to the receiver. This differs from our model in two significant ways: (1) the adversary is power constrained and the no-adversary case is noiseless, and (2) error correction rather than authentication is considered. 

Consider this model in the authentication setting. Because the channel is not overwritable for any power constraint $t$, the oblivious case of $q=1/2$ yields an authentication capacity of $\cst=1$, the non-adversarial capacity, by Theorem \ref{KK18_theorem}. In fact, Theorem \ref{theorem:simple_case} shows that even if the adversary can flip any number of bits (i.e. $t=1$), as long as $q>0$, the authentication capacity is equal to 1. Interestingly, the power constraint $t$, which is instrumental in the general AVC case, is immaterial in the authentication case.
\end{remark}

%%%%%%%%%%%%%%%%%%%%%%%%%%%%

{\section{Conclusions}\label{section:conclusions}

We have studied keyless authentication over arbitrarily varying
channels with a myopic adversary. We introduced a hierarchy of
channel conditions characterizing when authentication is possible.
On the converse side, $U$-overwritability is a sufficient condition
for zero authentication capacity for both stochastic and
deterministic codes. For stochastic codes, a finer condition,
$U$-distribution overwritability, governs positivity: when the
channel is not $U$-distribution overwritable, the capacity under
stochastic codes equals the non-adversarial channel capacity. We
further showed that single-letter distribution overwritability does
not imply multi-letter distribution overwritability, so the
relevant condition is multi-letter. We established a tight
connection between authentication and identification via
channels~\cite{AhlswedeD:1989}: the stochastic authentication
capacity coincides with the authentication tag capacity, and
whenever positive, both equal the non-adversarial channel capacity
(Theorem~\ref{thm:dichotomy}). This dichotomy shows that, with
stochastic encoders, every channel admits authentication either at
the full non-adversarial rate or not at all. We illustrated this
framework on the myopic binary adversarial channel, giving a
complete characterization: the authentication capacity with
stochastic codes equals $1-H(p)$ whenever $q>0$ and $p<1/2$, while
stochastic encoding is necessary to achieve this capacity when
$q \leq p$. Finally, we established a partial ordering on adversary
channels via stochastic degradation under which both overwritability
conditions are monotone, providing a tool for transferring
zero-capacity results across adversary models.

Several open questions remain. First, tight conditions for
positivity of the deterministic authentication capacity in the
general myopic setting are not yet known. Second, it would be of
interest to determine the infinite-letter limit of the multi-letter
$U$-distribution overwritability condition; we note that in the
oblivious setting, this limit coincides with the standard
overwritability condition, as can be seen by combining the
characterization of authentication capacity in~\cite{KosutK:2018}
with Theorem~\ref{thm:infinitedim} specialized to the oblivious case.
We conjecture that the infinite-letter limit of the multi-letter
$U$-distribution overwritability condition also coincides with the $U$-overwritability condition. Finally, extending these results to continuous-alphabet channels
is a natural direction for future work.}

\appendices
\section{Proof of Lemma~\ref{lem:ctaggreater}}\label{app:ctaggreater}
Let $\mathsf{W}_{id}$ denote the noiseless binary channel whose input equals the output with probability one. Fix $R<\cst$. Let $\hat{R}\in(R,\cst)$ and $\tilde{R}=R/\hat{R}$. Let $\lambda_1,\lambda_2\in(0,1/2)$ be desired misauthentication probabilities. 
Following~\cite[Theorem~1]{AhlswedeD:1989}, for all $\tilde{n}$ large enough, there exist $(N^{\mathrm{(id)}},\tilde{n})$-identification codes $\left(\tilde{\Prob}_{B^{\tilde{n}}|M},\set{	\tilde{\phi}_m}_{m\in[N^{\mathrm{(id)}}]}\right)$ for the channel $\mathsf{W}_{id}$ that achieve misidentification probabilities $(\lambda_1/2,\lambda_2/2)$  and $N^{\mathrm{(id)}}\geq \lfloor 2^{2^{\tilde{n}\tilde{R}}}\rfloor$.  Note that the number of codewords for such a code is at most $N^{\mathrm{(id)}}_{out}\triangleq 2^{\tilde{n}}$. Next, from the definition of authentication capacity, for all $n$ large enough, there exist $(N^{\mathrm{(auth)}},n)$-authentication codes $\left(\hat{\Prob}_{X^{n}|\mathbf{B}},\hat{\Phi}\right)$ for the channel $(\wy,\uz)$ with error probability at most $\max\set{\lambda_1/2,\lambda_2/2}$ and $N^{\mathrm{(auth)}}\geq \lfloor 2^{n \hat{R}}\rfloor$. 

Now, we compose the two codes in the following manner. First, let $n$ be large enough so that there exist both an $(N^{\mathrm{(id)}},\tilde{n})$-identification code and an $(N^{\mathrm{(auth)}},n)$-authentication code of the form above with $N^{\mathrm{(auth)}}=N^{\mathrm{(id)}}_{out}$.  Consider an $(N^{\mathrm{(id)}},n)$-authentication tag for the channel $(\wy,\uz)$ defined through the encoder 
\[\Prob_{X^n|M}(x^n|m)=\sum_{b^{\tilde{n}}\in\set{0,1}^{\tilde{n}}}\hat{\Prob}_{X^n|\mathbf{B}}(x^n|b^{\tilde{n}})\tilde{\Prob}_{B^{\tilde{n}}|M}(b^{\tilde{n}}|m),\]
	and the decoders
	\[\phi_{\hat{m}}(y^n)=\begin{cases}\reject& \mbox{if } \Phi(y^n)=\reject \\
	\reject & \mbox{if } \Phi(y^n)=\hat{\mathbf{b}}\in\{0,1\}^{\tilde{n}}\mbox{ and }\\ &\quad \tilde{\phi}_{\hat{m}}(\mathbf{b})=\reject\\
	\accept&\mbox{otherwise},\end{cases}\]
	for all $m,\hat{m}\in[N^{\mathrm{(id)}}]$, $x^n\in\cX^n$, and $y^n\in\cY^n$.
	We note that the rate of this code is 
	\begin{align*}
 	\frac{1}{n}\log\log N^{\mathrm{(id)}} &\geq \frac{\tilde{n}{\tilde{R}}}{n}\\
 	&\geq \frac{\tilde{n}{\tilde{R}}}{(1/\hat{R})\log\tilde{N}} = \tilde{R} \hat{R}\\
 	&=R.
 \end{align*}
Further, the error probabilities for the authentication tag $(\tilde{\Prob}_{X^n|M},\set{\phi_{\hat{m}}})$ are bounded from above by the sum of the corresponding misidentification probabilities for the identification code $(\Prob_{B^{\tilde{n}}|M},\set{\tilde{\phi}_m})$ and the error probability for the authentication code $(\hat{\Prob}_{X^n|\mathbf{B}},\hat{\Phi})$. This shows that $C_{\mathrm{tag}}\geq \cst$.\endproof

\section{Proof of Lemma~\ref{lem:errorprob}}\label{app:errorprob}
 Consider the following chain of inequalities.
	\begin{align}
	\lefteqn{\Pr(\Phi(Y^n)\notin\set{m,\reject}, g(Y^{l})\neq m|m\textrm{ sent})\label{eq:step1}}\\
	&= \Pr_{X^n,Y^n,Z^n,S^n}(\phi_{g(Y^l)}(Y_{l+1}^n)=\accept, g(Y^l)\neq m|m\mbox{ sent}) \label{eq:step2}\\
	&= \sum_{m'\neq m}\Pr_{X^n,Y^n,Z^n,S^n}(\phi_{m'}(Y_{l+1}^n)=\accept|g(Y^l)=m',m\mbox{ sent})\eqbr{\qquad\times} \Pr_{X^n,Y^n,Z^n,S^n}(g(Y^l)=m'|m\mbox{ sent}) \label{eq:step3}\\
	&= \sum_{m'\neq m}\Pr_{X^n,Y^n,Z^n,S^n}(\phi_{m'}(Y_{l+1}^n)=\accept|m\mbox{ sent})\eqbr{\qquad\times} \Pr_{X^n,Y^n,Z^n,S^n}(g(Y^l)=m'|m\mbox{ sent}) \label{eq:step4}\\
	& \leq \max_{\hat{m}\neq m}\Pr_{X^n,Y^n,Z^n,S^n}(\phi_{\hat{m}}(Y_{l+1}^n)=\accept|m\mbox{ sent})\eqbr{\qquad\times} \sum_{m'}\Pr_{X^n,Y^n,Z^n,S^n}(g(Y^l)=m'|m\mbox{ sent}) \label{eq:step5}\\
	& \leq \max_{\hat{m}\neq m}\Pr_{X^n,Y^n,Z^n,S^n}(\phi_{\hat{m}}(Y_{l+1}^n)=\accept|m\mbox{ sent}).
\end{align}

In the above,~\eqref{eq:step2} follows from the definition of $\Phi$, and~\eqref{eq:step4} is obtained by noting that $Y^l-(X^n,Z^n,S^n,M)-Y_{l+1}^n$ is a Markov chain due to our two-phase coding scheme. Finally, taking the suprema with respect to $m$ and $\ProbJ_{S^n|Z^n}$ gives us the bound in the lemma statement. 

\section{Proof of Lemma~\ref{lem:u_epsilon_k}}\label{app:lem:u_epsilon_k}

  We first prove the result for $k=1$. 	We first note that under $U_\epsilon$, the set of myopically indistinguishable distributions is comprised of singletons (\emph{i.e.}, no two distinct distributions are indistinguishable). This follows from the observation that, with probability $\epsilon$, the output of channel $U_\epsilon$ is identical to the input. Thus, distinct input distributions induce distinct distributions at the output of the channel $U_\epsilon$.
	
	Now, if  $W_{Y|X,S}$ is distribution overwritable, it follows that $W_{Y|X,S}$ is $U_\epsilon$-distribution overwritable for each $\epsilon$. In particular, as no two distinct distributions are myopically indistinguishable, the adversarial attack used to claim distribution overwritability may be used to show $U_\epsilon$-distribution overwritability.
	
	Conversely, suppose that there exists $\epsilon_1\in(0,1)$ such that $W_{Y|X,S}$ is $U_\epsilon$-distribution overwritable for all $\epsilon\in(0,\epsilon_1]$. Consider an input distribution $P_X$. By the definition of $U_\epsilon$-distribution overwritability, there exist strategies $P^{\epsilon}_{S|X',Z}$ and $P^{\epsilon}_{S|X'}$ such that, for each $x',y$, we have 
	\begin{align}
		\wy(y\mid x',\szero) & = \epsilon \sum_{x, s}\Prob_X(x) P^{\epsilon}_{S\mid X', Z}(s\mid x', x){\wy(y|x,s)}+ (1-\epsilon)\sum_{x, s}\Prob_X(x) P^{\epsilon}_{S\mid X'}(s\mid x'){\wy(y|x,s)}.
	\end{align}
	Rearranging the terms above, we get
	\begin{align}
		(1-\epsilon)\left(\sum_{x, s}\Prob_X(x) P^{\epsilon}_{S\mid X'}(s\mid x'){\wy(y|x,s)} -\wy(y\mid x',\szero)\right) & = -\epsilon \left(\sum_{x, s}\Prob_X(x) P^{\epsilon}_{S\mid X', Z}(s\mid x', x){\wy(y|x,s)}- \wy(y\mid x',\szero)\right).
	\end{align}
	Noting that $0<\epsilon<1$, the above gives
	\begin{align}
		\sum_y\left\vert\sum_{x, s}\Prob_X(x) P^{\epsilon}_{S\mid X'}(s\mid x'){\wy(y|x,s)} -\wy(y\mid x',\szero)\right\vert & = \frac{\epsilon}{1-\epsilon}\sum_y\left\vert\sum_{x, s}\Prob_X(x) P^{\epsilon}_{S\mid X', Z}(s\mid x', x){\wy(y|x,s)}- \wy(y\mid x',\szero)\right\vert\\
		& \leq \frac{2\epsilon}{1-\epsilon},\label{eq:tri_in_up_bd}
	\end{align}
	where \eqref{eq:tri_in_up_bd} follows by the triangle inequality. Since the above holds for each $P_X$ and all $\epsilon\in(0,\epsilon_1]$, it follows that 
    
	\begin{align}
		\max_{P_X}\min_{P_{S\mid X'}}\sum_y\left\vert\sum_{x, s}\Prob_X(x) P_{S\mid X'}(s\mid x'){\wy(y|x,s)} -\wy(y\mid x',\szero)\right\vert=0.
	\end{align}
    This proves that $W_{Y|X,S}$ is distribution overwritable.
    
    The proof for $k>1$ proceeds similarly. First, we note that all input distributions are myopically indistinguishable as, with probability $\epsilon^k$, the adversary observes the transmitted sequence $x^k$. Now, when $\wy^{\otimes k}$ is distribution overwritable for any $\epsilon$, a myopic adversary with the channel $\uz_\epsilon^{\otimes k}$ can overwrite the channel by using the adversarial strategies from the oblivious channel (without relying on the observed $z^k$). This shows that $\wy^{\otimes k}$ is $\uz_\epsilon^{\otimes k}$-distribution overwritable. 
    Conversely, suppose that there exists $\epsilon_k\in(0,1)$ such that $\wy^{\otimes k}$ is $\uz_\epsilon^{\otimes k}$-distribution overwritable for each $\epsilon\in(0,\epsilon_k]$. Consider an input distribution $P_{X^k}$. By the definition of $U_\epsilon^{\otimes k}$-distribution overwritability, there exist a  strategy $P^{\epsilon}_{S^k|(X')^k,Z^k}$ such that, for each $(x')^k,y^k$, we have 
	\begin{align}
		\wy(y^k\mid (x')^k,\szero^k) & =  \sum_{x^k, s^k, z^k}\Prob_{X^k}(x^k) U_\epsilon^{\otimes k}(z^k|x^k) P^{\epsilon}_{S^k\mid (X')^k, Z^k}(s^k\mid (x')^k, z^k){\wy(y^k|x^k,s^k)}.\label{eq:ex:kow1}
	\end{align}
    Note that, with probability $\epsilon^k$, $Z^k$ equals $E^k$. Conditioned on this event, the adversarial strategy takes the form $P^{\epsilon}_{S^k\mid (X')^k}(s^k\mid (x')^k)$, \emph{i.e.}, the adversary acts obliviously. Thus,~\eqref{eq:ex:kow1} leads to 
    \begin{align}
		\wy(y^k\mid (x')^k,\szero^k) & = \epsilon^k \sum_{x^k, s^k}\Prob_{X^k}(x^k)  P^{\epsilon}_{S^k\mid (X')^k}(s^k\mid (x')^k){\wy(y^k|x^k,s^k)}\\ & \qquad +\sum_{\substack{x^k\\ s^k\\ z^k\neq E^k}}\Prob_{X^k}(x^k) U_\epsilon^{\otimes k}(z^k|x^k) P^{\epsilon}_{S^k\mid (X')^k, Z^k}(s^k\mid (x')^k, z^k){\wy(y^k|x^k,s^k)}.\label{eq:ex:kow2}
	\end{align}
	Rearranging the terms above, we get
	\begin{align}
		& (1-\epsilon^k)\left(\sum_{\substack{x^k\\ s^k\\ z^k\neq E^k}}\Prob_{X^k}(x^k) U_\epsilon^{\otimes k}(z^k|x^k) P^{\epsilon}_{S^k\mid (X')^k, Z^k}(s^k\mid (x')^k, z^k){\wy(y^k|x^k,s^k)}-\wy(y^k\mid (x')^k,\szero^k) \right) \\ & \qquad = -\epsilon^k \left(\sum_{x^k, s^k}\Prob_{X^k}(x^k)  P^{\epsilon}_{S^k\mid (X')^k}(s^k\mid (x')^k){\wy(y^k|x^k,s^k)}- \wy(y^k\mid (x')^k,\szero^k)\right).
	\end{align}
	Noting that $0<\epsilon<1$, the above gives
	\begin{align}
		&\sum_{y^k}\left\vert\sum_{\substack{x^k\\ s^k\\ z^k\neq E^k}}\Prob_{X^k}(x^k) U_\epsilon^{\otimes k}(z^k|x^k) P^{\epsilon}_{S^k\mid (X')^k, Z^k}(s^k\mid (x')^k, z^k){\wy(y^k|x^k,s^k)}-\wy(y^k\mid (x')^k,\szero^k)\right\vert \\ \qquad & = \frac{\epsilon^k}{1-\epsilon^k}\sum_{y^k}\left\vert\sum_{x^k, s^k}\Prob_{X^k}(x^k)  P^{\epsilon}_{S^k\mid (X')^k}(s^k\mid (x')^k){\wy(y^k|x^k,s^k)}- \wy(y^k\mid (x')^k,\szero^k)\right\vert\\
		& \leq \frac{2\epsilon^k}{1-\epsilon^k}.
	\end{align}
	Since the above holds for each $P_{X^k}$ and for all $\epsilon\in(0,\epsilon_k]$, it follows that 
	\begin{align}
		\max_{P_{X^k}}\min_{P_{S^k\mid (X')^k}}\sum_{y^k}\left\vert\sum_{x^k, s^k}\Prob_{X^k}(x^k) P_{S^k\mid (X')^k}(s^k\mid (x')^k){\wy(y^k|x^k,s^k)} -\wy(y^k\mid (x')^k,\szero^k)\right\vert=0.
	\end{align}
    This proves that $W_{Y|X,S}^{\otimes k}$ is distribution overwritable.

\section{Proof of Lemma~\ref{lem:overlappingsets}}\label{app:overlappingsets}

Choose $\mathfrak{B}$ by picking $N$ subsets of $[t]$ such that each  $\cB\in\mathfrak{B}$ by independently including every element of $[t]$ with probability $\alpha$ each. Then, we have, \begin{align}
&\expect{}{|\cB|}=\alpha n,\\
&\expect{}{|\cB\cap\cB'|}=\alpha^2n,\ \mathrm{and}\\
&\expect{}{|\cB\setminus\cB'|}=\alpha(1-\alpha)n
\end{align} 
	for all $\cB,\cB'\in\mathfrak{B}$. Next, we apply Chernoff bound to conclude that 
	\begin{align}
	& \Pr(\left| |\cB|-\alpha n\right|>\alpha\beta n)\leq 2e^{-\beta^2\alpha n/3},\\
	& \Pr(|\cB\cap\cB'|>\alpha^2(1+\beta)n)\leq e^{-\beta^2\alpha^2 n/3}, \ \mathrm{and}\\
	& \Pr(|\cB\setminus\cB'|<\alpha(1-\alpha)(1-\beta)n)\leq e^{-\beta^2\alpha(1-\alpha) n/2}	
	\end{align}
for all $\cB,\cB'\in\mathfrak{B}$. Taking a union bound over all $N$ sets in the first inequality and $N^2$ pairs of sets in the second and third, we see that this random choice of $\mathfrak{B}$ satisfies Properties 1) to 3) with probability at least
$1-Ne^{-\beta^2\alpha n/2}-N^2e^{-\beta^2\alpha^2 n/3}-N^2e^{-\beta^2\alpha(1-\alpha) n/2}$. Thus, as long as $N<e^{Rn}$, our procedure outputs a set $\mathfrak{B}$ satisfying Properties 1) through 3) with high probability. Further, we may let $N=2^{Rn}$, thus also satisfying Property 4). 

Next, given $\cB,\cB'\subseteq [n]$ such that $|\cB'\setminus\cB|>0$, for every $y^n\in\cY^n$ and $y\in\cY$, we have 
\begin{align} 
\hat{P}_{y^n}(y|\cB')&= \frac{\{i\in\cB':y_i=y\}}{|\cB'|}	\\
&= \frac{|\{i\in\cB'\setminus\cB:y_i=y\}|+|\{i\in\cB\cap\cB':y_i=y\}|}{|\cB'|}\\
&= \frac{|\cB'\setminus\cB|}{|\cB'|}\hat{P}_{y^n}(y|\cB'\setminus\cB)+\frac{|\cB\cap\cB'|}{|\cB'|}\hat{P}_{y^n}(y|\cB'\cap\cB).
\end{align}
Thus, by the triangle inequality for variational distance, we have
\begin{align}
	&\mathbb{V}\left(\hat{P}_{y^n}(\cdot|\cB'),\wy(\cdot|x',\szero)\right)\\
	&\geq\frac{|\cB'\setminus\cB|}{|\cB'|}\mathbb{V}\left(\hat{P}_{y^n}(\cdot|\cB'\setminus\cB),\wy(\cdot|x',\szero)\right) - \frac{|\cB\cap\cB'|}{|\cB'|}\mathbb{V}\left(\hat{P}_{y^n}(\cdot|\cB'\cap\cB),\wy(\cdot|x',\szero)\right)\\
	& \geq  \frac{|\cB\setminus\cB'|}{|\cB'|}\mathbb{V}\left(\hat{P}_{y^n}(\cdot|\cB'\setminus\cB),\wy(\cdot|x',\szero)\right)-2 \frac{|\cB'\cap\cB|}{|\cB'|} .
\end{align}

Property 5) now follows from applying properties 1) to 3) to $\cB$ and $\cB'$.\endproof

\section{Proof of Lemma~\ref{lem:typicality}}\label{app:typicality}

\emph{Notation:} In the following, we borrow notation from~\cite[Chapter~1]{CsiszarK:2011}. We let $\cP^{(n)}=\cP^{(n)}(\cY\times\cZ\times\cS)$ be the set of all types on $\cY\times\cZ\times\cS$. For a joint type $\Prob\in\cP^{(n)}$, let $\tc_\Prob$ denotes the corresponding type class, \emph{i.e.}, all sequences whose empirical distribution is $\Prob$. Similarly, we define the set of types and type classes for all other combinations random variables.

Suppose that $(z^n,s^n)\in \tc_{Q_{ZS}}$ for some $Q_{ZS}\in\cP^{(n)}(\cZ\times \cS)$. Let $V_{Y|ZS}$ be the conditional probability of $Y$ given $(Z,S)$ when $X$ is drawn as per $P_X$ and $S$ is drawn as per $Q_{S|Z}$, \emph{i.e.},\[V_{Y|ZS}(y|zs)=\frac{\sum_{x}P_X(x)\uz(z|x)Q_{S|Z}(s|z)\wy(y|x,s)}{\sum_{x}P_X(x)\uz(z|x)Q_{S|Z}(s|z)}.\]   Let $Q_{YZS}=V_{Y|ZS}Q_{ZS}$ and let $Q_Y(=Q_Y^{(z^n,s^n)})$ be the marginal distribution induced by $Q_{YZS}$ on $\cY$. Note that this distribution is completely determined by the joint type of $(z^n,s^n)$, the given channel conditional probabilities, and the given input distribution. 

\emph{Proof overview:} The proof of this lemma proceeds by first showing in~\eqref{eq:jttype} that, for any $(z^n,s^n)$,  the joint type of $(Y^n,z^n,s^n)$ is close to $Q_{YZS}=V_{Y|ZS}Q_{SZ}$,  with high probability over the generation of $Y^n$. Next, we show that, with a high probability over $Z^n$, the joint type of $(Y^n,Z^n,S^n)$ is close to the single letter distribution $V_{Y|ZS}P_ZQ_{S|Z}$, where $P_Z$ is the probability distribution of the random variable $Z$ when $X$ is drawn from $P_X$ and $Z$ is the output of the channel $\uz$. This allows us to bound the probability that the variational distance in the lemma statement exceeds $\kappa(P_X,x')/2$ by a similar expression in terms of the type $Q_Y$ of $Y$ (Eq.~\eqref{eq:lastbigeqn}). Finally, we note that $Q_Y$ approximately satisfies the form required in the definition~\eqref{eq:kappa} to conclude that the lemma statement holds.  

\emph{Proof details:} Let $\kappa=\kappa(P_X,x')$. Let $(z^n,s^n)\in \tc_{Q_{ZS}}$ for some $Q_{ZS}\in\cP^{(n)}(\cZ\times \cS)$ and let $Y^n\sim\prod_i\Probv_{Y|Z,S}(\cdot|z_i,s_i)$. Following~\cite[Lemma~2.6]{CsiszarK:2011}, we first note that, for all $Q'_{YZS}\in\cP^{(n)}(\cY\times\cZ\times \cS)$ such that $ Q'_{ZS}=Q_{ZS}$,\[\Pr((Y^n,z^n,s^n)\in \tc_{Q'_{YZS}})\leq 2^{-nD(\Prob'_{Y|ZS}||\Probv_{Y|ZS})|\Prob_{ZS})}.\]
	Thus, for the given $(z^n,s^n)$, and $\eta>0$,
	\begin{align}
	& \Pr( D({\hat{P}}_{Y^n,z^n,s^n}||\Probv_{Y|ZS}Q_{ZS})>\eta    | z^n,s^n) \\
	\quad	& = \sum_{\substack{Q'_{YZS}\in\cP^{(n)}(\cY\times\cZ\times\cS)\\
	s.t.\ Q'_{ZS}=Q_{ZS}\\ D(Q'_{YZS}||V_{Y|ZS}Q_{ZS})>\eta}}\Pr((Y^n,z^n,s^n)\in \tc_{P'_{YZS}})\\
	\quad & \leq \sum_{\substack{Q'_{YZS}\in\cP^{(n)}(\cY\times\cZ\times\cS)\\
	s.t.\ Q'_{ZS}=Q_{ZS}\\ D(Q'_{YZS}||V_{Y|ZS}Q_{ZS})>\eta}} e^{-nD(Q'_{Y|ZS}||\Probv_{Y|ZS}|Q_{ZS})}\\
	&  \leq (n+1)^{|\cY|} 2^{-n\eta} < 2^{-n\eta/2}.\label{eq:jttype}
	\end{align}
Now, by Pinsker's inequality, whenever $D(\hat{P}_{Y^n,z^n,s^n}||Q_{YZS})<\eta$, we have $\mathbb{V}(\hat{P}_{Y^n,z^n,s^n},Q_{YZS}))<\sqrt{2\eta}$. Note that 

\begin{align}
\mathbb{V}(\hat{\Prob}_{Y^n},\wy(\cdot|x',\szero))&\geq 	\mathbb{V}(Q_Y,\wy(\cdot|x',\szero))-\mathbb{V}(\hat{\Prob}_{Y^n},Q_Y)\\
& \geq  	\mathbb{V}(Q_Y,\wy(\cdot|x',\szero))-\mathbb{V}(\hat{P}_{Y^n,z^n,s^n},Q_{YZS})\\
& \geq  	\mathbb{V}(Q_Y,\wy(\cdot|x',\szero))-\sqrt{2\eta}.\label{eq:pinsker1}\end{align}
Let $P_Z$ be the distribution of each $Z_i$ under the conditions of the lemma, \emph{i.e.} $P_Z(\cdot)=\sum_{x\in\cX}P_X(x)\uz(\cdot|x)$. We have,
\begin{align}
&\Pr\bigg(\mathbb{V}(\hat{P}_{Y^n},\wy(\cdot|x',\szero))<\frac{\kappa}{2}\bigg)	\\
&= \sum_{z^n,s^n}\bigg[\Pr\bigg(\mathbb{V}(\hat{P}_{Y^n},\wy(\cdot|x',\szero))<\frac{\kappa}{2}\bigg|z^n,s^n\bigg) \prod_i\Prob_Z(z_i) Q_{S^n|Z^n}(s^n|z^n)\bigg] \\
& = \sum_{Q_{ZS}\in\cP^{(n)}} \sum_{(z^n,s^n)\in \tc_{Q_{ZS}}}\bigg[\Pr\bigg(\mathbb{V}(\hat{P}_{Y^n},\wy(\cdot|x',\szero))<\frac{\kappa}{2}\bigg|z^n,s^n\bigg)  \prod_i\Prob_Z(z_i) Q_{S^n|Z^n}(s^n|z^n)\bigg] \\
& \stackrel{(a)}{\leq} \sum_{Q_{ZS}\in\cP^{(n)}} \sum_{\substack{(z^n,s^n)\\ \;\;\in \tc_{Q_{ZS}}}}\bigg[\Pr\bigg(\mathbb{V}(Q_Y,\wy(\cdot|x',\szero))-\sqrt{2\eta}<\frac{\kappa}{2}\bigg|z^n,s^n\bigg) \prod_i\Prob_Z(z_i) Q_{S^n|Z^n}(s^n|z^n)\bigg]+2^{-n\eta/2} \\
& \stackrel{(b)}{\leq} \sum_{Q_{ZS}\in\cP^{(n)}} \sum_{\substack{(z^n,s^n)\\ \;\;\in \tc_{Q_{ZS}}}}\bigg[\Pr\bigg(\mathbb{V}(Q_Y,\wy(\cdot|x',\szero))-\sqrt{2\eta}<\frac{\kappa}{2}\bigg|z^n,s^n\bigg) 2^{-nD(Q_Z||P_Z)-nH(Q_Z)} Q_{S^n|Z^n}(s^n|z^n)\bigg]+2^{-n\eta/2} \\
& \stackrel{(c)}{\leq} \sum_{\substack{Q_{ZS}\in\cP^{(n)}\\D(Q_Z||P_Z)<\xi}} \sum_{\substack{(z^n,s^n)\\ \;\;\in \tc_{Q_{ZS}}}}\bigg[\Pr\bigg(\mathbb{V}(Q_Y,\wy(\cdot|x',\szero))-\sqrt{2\eta}<\frac{\kappa}{2}\bigg|z^n,s^n\bigg) 2^{-nD(Q_Z||P_Z)-nH(Q_Z)} Q_{S^n|Z^n}(s^n|z^n)\bigg]+2^{-n\xi/2}+2^{-n\eta/2}. \label{eq:lastbigeqn}
\end{align}

In the above, $(a)$ follows from~\eqref{eq:pinsker1}. $(b)$ follows from the method of types. To obtain $(c)$, we bound the probability contribution from all $z^n$ in type classes with $D(Q_Z||P_Z)>\xi$ by $2^{-n\xi/2}$. Let $\tilde{Q}_Y$ be the marginal distribution on $\cY$ induced by the joint distribution $P_ZQ_{S|Z}V_{Y|ZS}$. Note that $\tilde{Q}_Y$ is of the form
\begin{align}
	\tilde{Q}_Y(y)&=\sum_{z,s}P_Z(z)Q_{S|Z}(s|z) V_{Y|ZS}(y|z,s)\\
	& = \sum_{x,z,s}P_X(x)\uz(z|x)Q_{S|Z}(s|z) W_{Y|XS}(y|x,s)\label{eq:tildeq}.
\end{align}

Now, note that  for all $Q_{S|Z}$, $D(Q_Z||P_Z)<\xi$ implies that $\mathbb{V}(Q_Y,\tilde{Q}_Y)<\sqrt{2\xi}$, since 
\[D(Q_Y||\tilde{Q}_Y)\leq D(Q_ZQ_{S|Z}V_{Y|ZS}||P_ZQ_{S|Z}V_{Y|ZS}) = D(Q_Z||P_Z), \] 
and Pinsker's inequality gives us $\mathbb{V}(Q_Y,\tilde{Q}_Y)\leq \sqrt{2D(Q_Y||\tilde{Q}_Y)}$. As a consequence, we have, for all $Q_Y$ that obtained from a joint distribution $Q_ZQ_{S|Z}V_{Y|ZS}$ with $D(Q_Z||P_Z)<\xi$, we have, \[\mathbb{V}(Q_Y,\wy(\cdot|x',\szero))\geq \mathbb{V}(\tilde{Q}_Y,\wy(\cdot|x',\szero))-\sqrt{2\xi}.\]
Applying this bound to~\eqref{eq:lastbigeqn} gives
\begin{align}
&\Pr\bigg(\mathbb{V}(\hat{P}_{Y^n},\wy(\cdot|x',\szero))<\frac{\kappa}{2}\bigg)\\
& \leq  \Pr\bigg(\mathbb{V}(\tilde{Q}_Y,\wy(\cdot|x',\szero))<\sqrt{2\xi}+\sqrt{2\eta}+\frac{\kappa}{2}\bigg) +2^{-n\xi/2}+2^{-n\eta/2}.
\end{align}
Finally, comparing~\eqref{eq:tildeq} with~\eqref{eq:kappa}, we note that $\tilde{Q}_Y$  satisfies $\mathbb{V}(\tilde{Q}_Y,\wy(\cdot|x',\szero))\geq \kappa$. Thus, as long as $\sqrt{2\xi}+\sqrt{2\mu}<\mu/2$, the conclusion of the lemma follows with $\gamma=\min\set{\xi,\eta}/2$.

\section{Proof of Theorem~\ref{theorem:stoch_deg_posititivity}}\label{app:theorem:stoch_deg_posititivity}

%\internalmark{Move proof to appendix?}

    We prove this result for stochastic coding. Similar arguments follow for deterministic codes. 
 We will prove that the capacity for the channel $(W,U)$ is zero if the capacity for $(W,\tilde{U})$ is zero. To this end, let $R>0$ and consider a sequence of $(2^{nR},n)$ stochastic codes with stochastic encoders $\{P^{(n)}_{\X|M}\}_{n=1}^\infty$ and decoders $\{\phi^{(n)}\}_{n=1}^\infty$. 
    Suppose that the capacity for the channel $(W,\tilde{U})$ is zero. By definition, this implies that there exists a sequence of adversarial strategies $\{\tilde{J}^{(n)}_{\bfS|\tilde{\Z}}\}_{n=1}^\infty$ such that the error probability for the channel $(W,\tilde{U})$ does not vanish, \emph{i.e.}, there exists $\epsilon>0$ such that  
   \[ \limsup_{n\to\infty}\frac{1}{2^{nR}}\sum_{i=1}^{2^{nR}}\mathbb{E}_{\X \mid i}\left[\sum_{\tilde{\z}}\tilde{U}(\tilde{\z} \mid \X) \tilde{J}^{(n)}(\bfs_{0} \mid \tilde{\z})W((\phi^{(n)})^{-1}(i)^{c} \mid \X, \bfs_{0})+\right. 
 \left.\sum_{\bfs \neq \bfs_{0}}\sum_{\z}U(\z \mid \X) \tilde{J}^{(n)}(\bfs\mid \z)W((\phi^{(n)})^{-1}(\{i,0\})^{c} \mid \X, \bfs)\right]=\epsilon.\]

 Let $P_{\tilde{Z}|Z}$ be chosen as per~\eqref{eq:def:degraded}. Consider the sequence of adversary strategies $\{{J}^{(n)}_{\bfS|{\Z}}\}_{n=1}^\infty$ with 
 \[{J}^{(n)}_{\bfS|{\Z}}(s^n|z^n)= \sum_{\tilde{z}^n\in\tilde{Z}^n} \tilde{J}^{(n)}_{\bfS|\tilde{\Z}} (s^n|\tilde{z}^n) \prod_{j=1}^n  P_{\tilde{Z}|Z}(\tilde{z}_j|z_j).\]
Under the above above adversarial strategy, the error probability for the sequence of $(2^{nR},n)$ stochastic codes with  encoders $\{P^{(n)}_{\X|M}\}_{n=1}^\infty$ and decoders $\{\phi^{(n)}\}_{n=1}^\infty$ is given by 
\begin{align}
&\limsup_{n\to\infty}\frac{1}{2^{nR}}\sum_{i=1}^{2^{nR}}\mathbb{E}_{\X \mid i}\left[\sum_{{\z}}{U}({\z} \mid \X) {J}^{(n)}(\bfs_{0} \mid {\z})W((\phi^{(n)})^{-1}(i)^{c} \mid \X, \bfs_{0})+\right. 
 \left.\sum_{\bfs \neq \bfs_{0}}\sum_{\z}U(\z \mid \X) {J}^{(n)}(\bfs\mid \z)W((\phi^{(n)})^{-1}(\{i,0\})^{c} \mid \X, \bfs)\right]\nonumber\\
 & = \limsup_{n\to\infty}\frac{1}{2^{nR}}\sum_{i=1}^{2^{nR}}\mathbb{E}_{\X \mid i}\left[\sum_{{\z}}{U}({\z} \mid \X) \sum_{\tilde{\z}} \tilde{J}^{(n)} (\bfs|\tilde{\z}) \prod_{j=1}^n  P_{\tilde{Z}|Z}(\tilde{z}_j|z_j)W((\phi^{(n)})^{-1}(i)^{c} \mid \X, \bfs_{0})+\right.\nonumber\\
&\qquad\qquad  \left.\sum_{\bfs \neq \bfs_{0}}\sum_{\z}U(\z \mid \X) \sum_{\tilde{\z}} \tilde{J}^{(n)} (\bfs|\tilde{\z}) \prod_{j=1}^n  P_{\tilde{Z}|Z}(\tilde{z}_j|z_j)W((\phi^{(n)})^{-1}(\{i,0\})^{c} \mid \X, \bfs)\right]\nonumber\\
 & = \limsup_{n\to\infty}\frac{1}{2^{nR}}\sum_{i=1}^{2^{nR}}\mathbb{E}_{\X \mid i}\left[\sum_{{\z,\tilde{\z}}}\prod_{j=1}^n  P_{\tilde{Z}|Z}(\tilde{z}_j|z_j){U}({\z} \mid \X)  \tilde{J}^{(n)} (\bfs|\tilde{\z}) W((\phi^{(n)})^{-1}(i)^{c} \mid \X, \bfs_{0})+\right.\nonumber\\
&\qquad\qquad  \left.\sum_{\bfs \neq \bfs_{0}}\sum_{{\z,\tilde{\z}}}\prod_{j=1}^n  P_{\tilde{Z}|Z}(\tilde{z}_j|z_j){U}({\z} \mid \X)  \tilde{J}^{(n)} (\bfs|\tilde{\z}) W((\phi^{(n)})^{-1}(\{i,0\})^{c} \mid \X, \bfs)\right]\nonumber\\
& = \limsup_{n\to\infty}\frac{1}{2^{nR}}\sum_{i=1}^{2^{nR}}\mathbb{E}_{\X \mid i}\left[\sum_{{\tilde{\z}}}\tilde{U}(\tilde{\z} \mid \X)  \tilde{J}^{(n)} (\bfs|\tilde{\z}) W((\phi^{(n)})^{-1}(i)^{c} \mid \X, \bfs_{0})+\right.  \left.\sum_{\bfs \neq \bfs_{0}}\sum_{{\tilde{\z}}}\tilde{U}({\tilde{\z}} \mid \X)  \tilde{J}^{(n)} (\bfs|\tilde{\z}) W((\phi^{(n)})^{-1}(\{i,0\})^{c} \mid \X, \bfs)\right]\nonumber \\
& =\epsilon>0.
\end{align}
As the above holds for any sequence of encoders $\{P^{(n)}_{\X|M}\}_{n=1}^\infty$ and decoders $\{\phi^{(n)}\}_{n=1}^\infty$, it follows that the capacity for the myopic channel $(W,U)$ is zero.

\section{Proof of Theorem~\ref{theorem:stoch_deg_overwritability}}\label{app:theorem:stoch_deg_overwritability}

We first prove the implication from $\tilde{U}$-overwritability.  Let $P_{S|X',Z'}$ be the distribution guaranteed by $\tilde{U}$-overwritability, and let $P_{S|X',Z}(s|x',z)=\sum_{\tilde{z}}P_{S|{X'},\tilde{Z}}(s|x',z')P_{\tilde{Z}|Z}(\tilde{z}|z)$, where $P_{\tilde{Z}|Z}$ is the distribution guaranteed by stochastic degradation. Then, 
\begin{align*}
&\sum_{s,z}U^(z|x)P_{S|X',Z}(s|x',z)W(y|x,s) 
    \notag \\ &\quad
    = 
\sum_{s,z,\tilde{z}}P_{\tilde{Z}|Z}(\tilde{z}|z)U^{1}(z|x)P_{S|X',\tilde{Z}}(s|x',\tilde{z})W(y|x,s) 
    \\&\quad
    = 
\sum_{s,\tilde{z}}\tilde{U}(\tilde{z}|x)P_{S|X',\tilde{Z}}(s|x',\tilde{z})W(y|x,s) 
    \\&\quad 
    = 
W(y | x', s_{0}).
\end{align*}
We conclude that $W_{Y|X,S}$ is $U$-overwritable.

{ The implication from $\tilde{U}$-distribution overwritability follows similarly by using the observation that if $P_X$ and $P'_X$ are myopically indistinguishable under $U$, then there are also myopically indistinguishable under $\tilde{U}$.}

\section{Proof of Theorem~\ref{thm:koverwritability}}\label{app:thm:koverwritability}
We note that it is sufficient to prove that, for each positive integer $k$, $\wy^{\otimes k}$ is $\uz^{\otimes k}$-overwritable if and only if $\wy$ is $\uz$-overwritable. 

First, suppose that $\wy$ is $\uz$-overwritable. Let $P_{S|X',Z}$ be the conditional distribution guaranteed by Definition~\ref{def:U-overwritable}. Let $P_{S^k|\tilde{X}^k,Z^k}$ be the conditional distribution defined as 
\begin{equation}
    P_{S^k|\tilde{X}^k,Z^k}(s^k|\tilde{x}^k,z^k)=\prod_{i=1}^k P_{S|X',Z}(s_i|\tilde{x}_i,z_i).
\end{equation}
Then, for any $x^k$, $\tilde{x}^k$, and $y^k$, we have
\begin{align}
    \sum_{s^k,z^k} \uz^{\otimes k}(z^k|x^k)  P_{S^k|\tilde{X}^k,Z^k}(s^k|\tilde{x}^k,z^k) \wy^{\otimes k}(y^k|x^k,s^k) & = \sum_{s^k,z^k} \prod_{i=1}^k\uz(z_i|x_i)  P_{S|X',Z}(s_i|\tilde{x}_i,z_i) \wy(y_i|x_i,s_i)\\
    & =  \prod_{i=1}^k \sum_{s_i,z_i} \uz(z_i|x_i)  P_{S|X',Z}(s_i|\tilde{x}_i,z_i) \wy(y_i|x_i,s_i)\\
    & \overset{(a)}{=} \prod_{i=1}^k \wy(y_i|\tilde{x}_i,\szero)\\
    & = \wy^{\otimes k}(y^k|\tilde{x}^k,\szero^k).
\end{align}
In the above, $(a)$ follows from Definition~\ref{def:U-overwritable}. This shows that $\wy^{\otimes k}$ is $\uz^{\otimes k}$-overwritable. 

Now, suppose that $\wy^{\otimes k}$ is $\uz^{\otimes k}$-overwritable for some integer $k>1$. Let $P_{S^k|\tilde{X}^k,Z^k}$ be a conditional distribution (guaranteed to exists by Definition~\ref{def:U-overwritable} applied to the $k$-letter channel) that satisfies 
\begin{equation}
     \sum_{s^k,z^k} \uz^{\otimes k}(z^k|x^k)  P_{S^k|\tilde{X}^k,Z^k}(s^k|\tilde{x}^k,z^k) \wy^{\otimes k}(y^k|x^k,s^k) = \wy^{\otimes k}(y^k|\tilde{x}^k,\szero^k)
\end{equation}
for every $x^k$, $\tilde{x}^k$, and $y^k$. Let $P_{X^{k-1}}$ be any distribution on $\cX^{k-1}$. Consider $P_{S|X',Z}$ defined as
\begin{equation}
    P_{S|X',Z}(s|x',z)=\sum_{x^{k-1},\tilde{x}^{k-1}} P_{X^{k-1}}(x^{k-1})P_{X^{k-1}}(\tilde{x}^{k-1})\uz^{\otimes (k-1)}(z^{k-1}|x^{k-1}) P_{S^k|\tilde{X}^k,Z^k}((s,s^{k-1})|(x',\tilde{x}^{k-1}),(z,z^{k-1})). 
\end{equation}
Then, for any $x$, $x'$, and $y$, we have 
\begin{align}
\lefteqn{\sum_{s,z} \uz(z|x)  P_{S|X',Z}(s|x',z) \wy(y|x,s)}\\
& = \sum_{s,z} \uz(z|x)  \sum_{x^{k-1},\tilde{x}^{k-1}} P_{X^{k-1}}(x^{k-1})P_{X^{k-1}}(\tilde{x}^{k-1})\uz^{\otimes (k-1)}(z^{k-1}|x^{k-1}) P_{S^k|\tilde{X}^k,Z^k}((s,s^{k-1})|(x',\tilde{x}^{k-1}),(z,z^{k-1})) \wy(y|x,s)\\
& = \sum_{x^{k-1},\tilde{x}^{k-1}} P_{X^{k-1}}(x^{k-1})P_{X^{k-1}}(\tilde{x}^{k-1}) \uz^{\otimes k}((z,z^{k-1})|(x,x^{k-1})) P_{S^k|\tilde{X}^k,Z^k}((s,s^{k-1})|(x',\tilde{x}^{k-1}),(z,z^{k-1}))\wy(y|x,s)\\
& \overset{(a)}{=} \sum_{x^{k-1},\tilde{x}^{k-1},y^{k-1}} P_{X^{k-1}}(x^{k-1})P_{X^{k-1}}(\tilde{x}^{k-1}) \uz^{\otimes k}((z,z^{k-1})|(x,x^{k-1})) P_{S^k|\tilde{X}^k,Z^k}((s,s^{k-1})|(x',\tilde{x}^{k-1}),(z,z^{k-1}))\wy^{\otimes k}((y,y^{k-1})|(x,x^{k-1}),(s,s^{k-1}))\\
& \overset{(b)}{=} \sum_{x^{k-1},\tilde{x}^{k-1},y^{k-1}} P_{X^{k-1}}(x^{k-1})P_{X^{k-1}}(\tilde{x}^{k-1}) \wy^{\otimes k}((y,y^{k-1})|(x',\tilde{x}^{k-1}),\szero^k)\\
& = \wy(y|x',\szero).
\end{align}
In the above, $(a)$ is obtained by noting that $\sum_{y^{k-1}}\wy^{\otimes (k-1)}(y^{k-1}|x^{k-1},s^{k-1})=1$, and $(b)$ follows from $\wy^{\otimes k}$ being $\uz^{\otimes k}$-overwritable. This shows that $\wy$ is $\uz$-overwritable.

\section{Proof of Theorem~\ref{thm:kdistoverwritability}}\label{app:thm:kdistoverwritability}
Let $k<k'$. Fix a set $\cP^{(k)}$ of distributions over $\cX^k$ that are myopically indistinguishable. Fix a distribution $\hat{P}_{X^{k'-k}}$ over $\cX^{k'-k}$. Note that, if $P_{X^k},P'_{X^k}\in \cP^{(k)}$, and $\hat{P}_{X^{k'-k}}$ is any distribution over $\cX^{k'-k}$, then the distributions $P_{X^k}\circ \hat{P}_{X^{k'-k}}$ and $P'_{X^k}\circ \hat{P}_{X^{k'-k}}$ are also myopically indistinguishable, \emph{i.e.}, there exists a set $\cP^{(k')}$ of myopically indistinguishable distributions over $\cX^{k'}$ that contains the distributions $P_{X^k}\circ \hat{P}_{X^{k'-k}}$ and $P'_{X^k}\circ \hat{P}_{X^{k'-k}}$.

Suppose that, $\wy^{\otimes k'}$ is $\uz^{\otimes k'}$-distribution overwritable. By Definition~\ref{def:udistow}, there exists $P_{S^{k'}|\tilde{X}^{k'},Z^{k'}}$ such that 
\begin{equation}
\sum_{x^{k'}, s^{k'}, z^{k'}}\Prob_{X^{k'}}(x^{k'}) U^{\otimes k'}(z^{k'}\mid x^{k'})P_{S^{k'}\mid \tilde{X}^{k'}, Z^{k'}}(s^{k'}\mid \tilde{x}^{k'}, z^{k'}){\wy^{\otimes k'}(y^{k'}|x^{k'},s^{k'})}=\wy^{\otimes k'}(y^{k'}\mid \tilde{x}^{k'},\szero^{k'}) \quad \forall \Prob_{X^{k'}}\in\cP^{(k')}, \tilde{x}^{k'}\in\cX^{k'}, y^{k'}\in\cY^{k'}.	
\end{equation}
Now, consider the conditional distribution $P_{S^{k}|\tilde{X}^{k},Z^{k}}$ defined as 
\begin{equation}
     P_{S^{k}|\tilde{X}^{k},Z^{k}}(s^k|\tilde{x}^k,z^k) = \sum_{\hat{x}^{k'-k},\tilde{x}_{k+1}^{k'},y_{k+1}^{k'},s_{k+1}^{k'}} \hat{P}_{X^{k'-k}}(\hat{x}^{k'-k})  \hat{P}_{X^{k'-k}}(\tilde{x}_{k+1}^{k'})  U^{\otimes {k'-k}}(z_{k+1}^{k'}|\hat{x}^{k'-k}) P_{S^{k'}|\tilde{X}^{k'},Z^{k'}}(s^{k'}|\tilde{x}^{k'},z^{k'})\wy^{\otimes {k'-k}}(y_{k+1}^{k'}|\hat{x}^{k'-k},s_{k+1}^{k'}).
\end{equation}

Now, for any distribution $P_{X^k}\in\cP^{(k)}$, ,$\tilde{x}^k\in\cX^k$ and $y^k\in\cY^k$, we have, 
\begin{align}
\lefteqn{    \sum_{x^{k}, s^{k}, z^{k}}\Prob_{X^{k}}(x^{k}) U^{\otimes k}(z^{k}\mid x^{k})P_{S^{k}\mid \tilde{X}^{k}, Z^{k}}(s^{k'}\mid \tilde{x}^{k}, z^{k}){\wy^{\otimes k}(y^{k}|x^{k},s^{k})}}\\
& =  \sum_{x^{k}, s^{k'}, z^{k},\hat{x}^{k'-k},\tilde{x}_{k+1}^{k'},y_{k+1}^{k'}}\Prob_{X^{k}}(x^{k}) \hat{P}_{X^{k'-k}}(\hat{x}^{k'-k})  \hat{P}_{X^{k'-k}}(\tilde{x}_{k+1}^{k'})  U^{\otimes {k}}(z^{k'}|( x^{k},\hat{x}^{k'-k})) P_{S^{k'}|\tilde{X}^{k'},Z^{k'}}(s^{k'}|\tilde{x}^{k'},z^{k'})\wy^{\otimes {k'}}(y^{k'}|(x^{k},\hat{x}^{k'-k}),s^{k'})\\
& \overset{(a)}{=}  \sum_{\tilde{x}_{k+1}^{k'},y_{k+1}^{k'}}  \hat{P}_{X^{k'-k}}(\tilde{x}_{k+1}^{k'}) \wy^{\otimes {k'}}(y^{k'}|\tilde{x}^{k'},\szero^{k'})\\
& =  \wy^{\otimes {k'}}(y^{k}|\tilde{x}^{k},\szero^{k}).
\end{align}
In the above, $(a)$ is a consequence of $\wy^{\otimes k'}$ being $\uz^{\otimes k'}$-distribution overwritable and our previous observation that the distribution $P_{X^k}\circ \hat{P}_{X^{k'-k}}$ lies in $\cP^{(k')}$. This proves that $\wy^{\otimes k}$ is $\uz^{\otimes k}$-distribution overwritable.\\

\section{Auxilliary Lemma for Example~\ref{exmp:BSCZ}}
\begin{lem}\label{lem:alhpastar} Let $W_{\alpha,\delta}$ be the transition probability for AZBSC($\alpha,\delta$). Let $0<\alpha'<\alpha\leq\delta$ and $k\in\mathbb{N}$. Then,$W_{\alpha',\delta}^{\otimes k}$ is distribution overwritable if $W_{\alpha,\delta}^{\otimes k}$ is distribution overwritable.
\end{lem}
\begin{proof}
    First, we note that, for $k=1$, the claim of the lemma follows from the fact that, whenever, $0\leq\alpha\leq \delta$, $W_{\alpha,\delta}$, for every input distribution parameter $p$, there exists an adversarial strategy $Q_{(\alpha,\delta,p)}(s|x')$  (which equals $P_S$ as per~\eqref{eq:ex:azbscadvstrat} for $x'=0$) that shows that $W_{\alpha,\delta}$ is distribution overwritable.

    We observe that if $Q_{(\alpha,\delta,p)}(s|0)$ and $Q_{(\alpha',\delta,p)}(s|0)$ both satisfy the following relations:
    \begin{align}
        0&=Q_{(\alpha,\delta,p)}(s_0|0)=Q_{(\alpha',\delta,p)}(s_0|0), \mbox{ and}\\
        p&= \frac{1-\delta-(1-\alpha) Q_{\alpha,\delta,p}(s_{1,0}|0)}{\alpha} = \frac{1-\delta-(1-\alpha') Q_{\alpha',\delta,p}(s_{1,0}|0)}{\alpha'}. 
    \end{align}
    In particular, we observe that knowing the strategy  $Q_{\alpha,\delta,p}(s|0)$, we can compute $Q_{\alpha',\delta,p}(s|0)$ even without requiring the value of $p$. A similar argument applies when $x'   =1$. Thus, there exists a mapping $M^{\alpha\to\alpha'}_{x'}: \mathcal{P}(\cS)\to\mathcal{P}(\cS)$ such that $Q_{\alpha',\delta,p}(\cdot|x')=M^{\alpha\to\alpha'}_{x'}(Q_{\alpha,\delta,p}(\cdot|x'))$ is the adversarial strategy under $\alpha'$ that achieves the same  output distribution as $Q_{\alpha,\delta,p}(\cdot|x')$ does  under $\alpha$.
    
    Using this observation, we now prove the claim of the lemma for $k>1$ in the following. Fix an input distribution $P_{X^k}$ and a target sequence $\tilde{x}^k\in \mathcal{X}^k$. 
 Since $W_{\alpha,\delta}^{\otimes k}$ is distribution overwritable, for any input distribution $P_{X^k}$, there exists an oblivious adversarial strategy $Q^{(0)}(s^k)$ such that the joint output distribution satisfies:
\begin{equation}
P(y^k) = W^{\otimes k}(y^k \mid \tilde{x}^k, s_0^k).
\end{equation}

For $j=0,1,2,\ldots,k$, we let $V^{(j)}$ be a $k$-letter channel with input alphabet $\mathcal{X}^k$, output alphabet $\cY^k$ and the state set $\cS^k$ such that the first $j$ coordinates operate as $W_{\alpha',\delta}$ and the remaining $k-j$ coordinates operate as $W_{\alpha,\delta}$. We proceed by induction, constructing a sequence of strategies $Q^{(j)}$ that overwrite $V^{(j)}$.

Assume there exists $Q^{(i-1)}$ that overwrites $V^{(i-1)}$. Under $(Q^{(i-1)}, V^{(i-1)})$, the joint output distribution $P(y^k)$ equals $W^{\otimes k}(y^k \mid \tilde{x}^k, s_0^k)$. Thus, $P(y_i \mid y^{i-1})$ equals $W(y_i \mid \tilde{x}_i, s_0)$.

Expanding this conditional probability differently, we have 
\begin{equation}
P(y_i \mid y^{i-1}) = \sum_{x_i, s_i} P(x_i, s_i \mid y^{i-1}) W_{\alpha,\delta}(y_i \mid x_i, s_i).
\end{equation}

As established for the $k=1$ case, there exists a mapping $M^{\alpha\to\alpha'}_{\tilde{x}_i}: \mathcal{P}(\cS)\to\mathcal{P}(\cS)$ that computes the required state marginal under $\alpha'$ directly from the state marginal under $\alpha$, while exactly preserving the target output probability $W(y_i \mid \tilde{x}_i, s_0)$. Crucially, this mapping operates independently of the conditional input distribution $P(x_i \mid y^{i-1})$.

Because the adversary is oblivious, the state sequence $S^k$ is generated independently of the input sequence $X^k$. Therefore, conditioned on $S^{i-1}$, the current state $S_i$ is independent of $(X_i, Y^{i-1})$. This allows us to strictly factor the joint conditional distribution through the state history:
\begin{equation}
P(x_i, s_i \mid y^{i-1}) = \sum_{s^{i-1}} P(x_i, s^{i-1} \mid y^{i-1}) Q^{(i-1)}(s_i \mid s^{i-1}).
\end{equation}

We construct the strategy $Q^{(i)}$ for $V^{(i)}$ by replacing the conditional distribution of $s_i$ given $s^{i-1}$ with the distribution transformed by $M^{\alpha\to\alpha'}_{x'_i}$, \emph{i.e.},
\begin{equation}
Q^{(i)}(s^k) = Q^{(i-1)}(s^{i-1}) \left[ M^{\alpha\to\alpha'}_{\tilde{x}_i}\left( Q^{(i-1)}(\cdot \mid s^{i-1}) \right)(s_i) \right] Q^{(i-1)}(s_{i+1}^k \mid s^i).
\end{equation}

Note that, under the pair $(Q^{(i)}, V^{(i)})$, the joint distribution of $(X^k, S^{i-1}, Y^{i-1})$ is identical to that under $(Q^{(i-1)}, V^{(i-1)})$. In particular, the marginal probability $P(y^{i-1})$ and the conditioning distribution $P(x_i, s^{i-1} \mid y^{i-1})$ are also invariant. Thus,  under the pair $(Q^{(i)}, V^{(i)})$ we have,
\begin{equation}
P'(y_i \mid y^{i-1}) = \sum_{x_i, s_i} \left[ \sum_{s^{i-1}} P(x_i, s^{i-1} \mid y^{i-1}) M^{\alpha\to\alpha'}_{\tilde{x}_i}\left(Q^{(i-1)}(\cdot \mid s^{i-1})\right)(s_i) \right] W_{\alpha',\delta}(y_i \mid x_i, s_i).
\end{equation}

By the definition of $M^{\alpha\to\alpha'}_{\tilde{x}_i}$, this equals 
\begin{equation}
     \sum_{x_i, s_i}  \sum_{s^{i-1}} P(x_i, s^{i-1} \mid y^{i-1}) Q^{(i-1)}(s_i \mid s^{i-1}) W_{\alpha,\delta}(y_i \mid x_i, s_i) = W(y_i \mid x'_i, s_0) 
\end{equation}
irrespective of the conditional input distribution parameter $P(X_i = x_i \mid s^{i-1}, y^{i-1})$. Because the subsequent channel parameters and conditional strategies are unmodified, the full joint output probability remains exactly $W^{\otimes k}(y^k \mid \tilde{x}^k, s_0^k)$.

By induction, the final strategy $Q^{(k)}$ successfully overwrites the target channel $V^{(k)} = W_{\alpha',\delta}^{\otimes k}$, establishing that it is distribution overwritable.
\end{proof}

\section{Proof of the achievability part of Theorem~\ref{theorem:simple_case}}
\label{app:theorem:simple_case}

Per the MBAC$_{p,q}$ model, the adversary can see a noisy version of the transmitted codeword. In our proof of achievability, we will strengthen the adversary in order to simplify some arguments. Proving achievability for a stronger adversary simultaneously proves the result for any weaker adversary. Specifically, we introduce an oracle who will reveal to the adversary the exact distance $d$ of the transmitted codeword, $\x_{i}$, from the received word $\z$. Generally speaking, even given this information, there remain enough potentially transmitted words to make the adversary's task difficult. We will also allow the adversary to be aware of the exact error pattern, $\mathbf{e}$, of the BSC$(p)$ between the transmitter and receiver, so that it can design the state knowing the exact difference between the transmission and what the receiver will see. That is, the adversary can add the state $\mathbf{s}'=\mathbf{s}\oplus\mathbf{e}$, so that the receiver will see $\x_{i}\oplus\bfs'\oplus\mathbf{e}=\x_{i}\oplus\bfs$. To simplify our analysis, we will simply assume that when the adversary adds $\bfs\neq \bfs_{0}$, there is no additional channel noise from the BSC$(p)$. However, the decoder will still assume there has been some channel noise and decode appropriately (i.e. the receiver does not have any increased knowledge). Indeed, when the output is conditioned on $\bfs=\bfs_{0}$, the channel is a noisy BSC($p$) and the receiver must be able to decode reliably.

We first present two lemmas that will allow us to choose a good codebook. The first is a variation of Lemma 3 from \cite{CsiszarN:1988}. In each, let $M:=2^{nR}$.

\begin{lem}[\cite{CsiszarN:1988}]
\label{lemma:csiszar_narayan}
Let $\epsilon_{1}>0$ and let $\X_1,\ldots,\X_M$ be uniform over the type class of type $P_X$. With high probability, the realizations of $\X_1,\ldots,\X_M$ (i.e., $\x_1,\ldots,\x_M$) satisfy the following: for any type class $\tau_{XX'S}$ and any sequence $\bfs$, 
\begin{equation}
\label{equation:csiszar_narayan}
|\{i:\exists j\neq i \text{ s.t. } (\x_i,\x_j,\bfs)\in \tau_{XX'S}\}|
\leq 2^{n\left|R-I(X;X'S)+\left|R-I(X';S)\right|^+\right|^+ + n\epsilon_{1}}.
\end{equation}
\end{lem}

\begin{lem}
\label{lemma:sufficiently_myopic}
Let $\X_1, \dots ,\X_M$ be uniform over the type class $\tau_X$.
For each type class $\tau_Z$, and each $\z \in \tau_Z$, with high probability in $n$ the number of messages $i$ such that $\|\x_{i}\oplus \z\|=d$, where $\mathbf{x}_{i}$ is the realization of $\X_i$, is bounded below by 
\begin{equation*}
\left \lfloor \frac{1}{(n+1)^8} 2^{n(R-I(X;Z))} \right \rfloor ,
\end{equation*}
where we let $p_{X|Z}$ be the conditional distribution of pairs of words in their respective type classes that are distance $d$ apart, as follows:
\begin{align*}
p_{X|Z}(1|0) &= \frac{d}{2 np_{Z}(0)} + \frac{1}{2} - \frac{p_{X}(0)}{2p_{Z}(0)}\\
p_{X|Z}(0|1) &= \frac{d}{2 np_{Z}(1)} + \frac{1}{2} - \frac{p_{X}(1)}{2p_{Z}(1)}.
\end{align*}
\end{lem}

\begin{proofoutline}
The proof is a simple type counting exercise. See, for instance,~\cite[Lemma~10.1]{CsiszarK:2011}.
\end{proofoutline}

Combining Lemma~\ref{lemma:sufficiently_myopic} with the fact that
$|H\left( t \right) - H(t\pm\epsilon) | \leq - \epsilon \log_2(\epsilon/2)$
for $\epsilon < 1/4$ (see~\cite[Lemma~2.7]{CsiszarK:2011}) yields the following corollary. 
\begin{co}
\label{co:sufficiently_myopic}
Given $\epsilon < 1/4$ and $X,Z\sim \mathrm{Bernoulli}(1/2)$, if $\mathbf{z} \in \mathcal{T}_{\epsilon}^{(n)}(Z)$, $|d - nq | \leq n \epsilon $, and $R = 1 - H(p) + 5 \epsilon \log_2 (\epsilon/2) $, then the number of %codewords 
messages distance $d$ from $\z$ is, with high probability, bounded below by
\begin{equation*}
\label{equation:denominator_bound:co}
\left \lfloor \frac{2^{n\left(R-1+H(d/n) + \epsilon \log_2 \epsilon /2 \right)}}{(n+1)^8}  \right \rfloor \geq   \left \lfloor \frac{2^{n\left(H(q) - H(p)    + 7 \epsilon \log_2 \epsilon/2 \right)}}{(n+1)^8}  \right \rfloor .
\end{equation*}
\end{co}

We now prove achievability with a deterministic encoder.
\begin{proof}[\textbf{Theorem \ref{theorem:simple_case} achievability proof}]

Let $\tilde \delta_n $ be an arbitrary sequence such that $n\tilde \delta_n \geq \sqrt{n} \log_2 n $ and $\lim_{n\rightarrow \infty} \tilde \delta_n = 0$.  
For convenience, set $\delta_n := -\tilde \delta_n \log_2 \tilde \delta_n/2$.

Suppose $0\leq p<q\leq 1/2$. For $n$ sufficiently large, and without loss of generality, let $R = 1 - H(p) - 5 \delta_n  > 1 - H(q) + \frac{8}{n} \log_2 (n+1) + 2 \delta_n$. We construct a $(M:=2^{nR}, n)$ code family with vanishing probability of error.

By Lemmas \ref{lemma:csiszar_narayan} and \ref{lemma:sufficiently_myopic}, for $n$ sufficiently large, there exist codewords $\x_{1},\ldots, \x_{M}$ from the type class $\tau_X$, where $X \sim \text{Bernoulli}(1/2)$, such that: (1) for every $\z \in \mathcal{T}^{(n)}_{\tilde \delta_{n}}(Z)$ where $Z\sim$ Bernoulli($1/2$), the number of codewords that are distance $d$ from $\z$, where $|d-nq|\leq  n\tilde \delta_n$, is bounded below by $2^{n\left(H(q)-H(p)-7\delta_n- \frac{8}{n} \log_2 (n+1) \right)}$, and (2) for any type class $\tau_{XX'S}$ and any sequence $\bfs$, \eqref{equation:csiszar_narayan} holds.\\

\noindent\textit{Encoding:}   Given a message $i\in [M]$, transmit $\x_{i}$.\\

\noindent\textit{Decoding:} %Let $\epsilon>0$ be sufficiently small. 
Given an output $\y\in \{0,1\}^{n}$, decode to message $i\in [M]$ if $i$ is unique with the property that $\|\x_{i}\oplus\y\|<n(p+\tilde \delta_n)$. Otherwise, declare adversarial interference by outputting ``0''.\\

\noindent\textit{Probability of error analysis:} 
Define
\[S(\z,d):=\{ i \in [M] : \|\z\oplus \x_{i}\|=d\},\]
\[E(\bfs):=\{i \in [M] : \exists j\neq i \text{ s.t. } \|\x_{i}\oplus \bfs \oplus \x_{j}\|< n(p+\tilde \delta_n) \}.\]
That is, $S(\z,d)$ is the set of messages in $[M]$ whose corresponding codewords are distance $d$ from $\z \in \{0,1\}^{n}$,  and $E(\bfs)$ is the set of messages $i$ in $[M]$ such that adding $\bfs$ to $\x_{i}$ results in the decoder potentially  confusing the intended message with a false message $j$. 
Let $J_{\mathbf{S}|\mathbf{Z},D}$ be adversary's choice of distribution given knowledge of the distance to the transmitted codeword, which is given by the random variable $D$. 
For fixed $J$, $e(i,J)$ is the probability of decoding error given that message $i$ was sent. Then,
\begin{align}
e(J)  
&= \frac{1}{M}\sum_{i=1}^{M}P(\text{error} \mid i ) \\
&=\sum_{i,\bfs, \z, d}P(\text{error} \mid i , \bfs, \z, d)P(i \mid \bfs, \z, d)J(\bfs\mid \z, d)P(\z, d) \\
&=\sum_{i,\bfs, \z, d}P(\text{error} \mid i , \bfs, \z, d)P(i \mid \z, d)J(\bfs\mid \z, d)P(\z, d) \label{eq:OHLAWDY}
\end{align}
Since every message whose codeword is distance $d$ from $\z$ is equiprobable
\[P(i \mid \z, d)=
\begin{cases}
\frac{1}{|S(\z,d)|} & \text{ if }  i \in S(\z,d), \\
0 & \text{ otherwise.}
\end{cases}
\]
Observe that $P(\text{error} \mid i, \bfs, \z, d)=0$ if for every $j\neq i$, $\|\x_{i}\oplus \bfs \oplus \x_{j}\| \geq  n(p+\tilde \delta_n)$. In other words, $P(\text{error} \mid i , \z, \bfs, d)=0$ if $i\notin E(\bfs)$. Then, \eqref{eq:OHLAWDY} is bounded as follows:
\begin{align}
\sum_{\bfs, \z, d}&J(\bfs\mid \z, d)P(\z, d)\sum_{i=1}^{M}P(\text{error} \mid i ,\bfs, \z, d)P(i \mid \z, d) \nonumber \\
&\leq \sum_{\bfs, \z, d }J(\bfs\mid \z, d)P(\z, d) \cdot \frac{|S(\z,d)\cap E(\bfs)|}{|S(\z,d)|}\label{eq:prehappysnowman0}\\
&\leq \sum_{\substack{\bfs, \\ d:  |d-nq| \leq n \tilde \delta_n \\ \z\in \mathcal{T}_{ \tilde \delta_n}^{(n)}} }J(\bfs\mid \z, d)P(\z, d) \cdot \frac{|S(\z,d)\cap E(\bfs)|}{|S(\z,d)|}  \label{eq:prehappysnowman1} \\
&\quad+ \sum_{\substack{\bfs, d \\ \z\notin \mathcal{T}_{ \tilde \delta_n}^{(n)}} }J(\bfs\mid \z, d)P(\z, d)  \label{eq:prehappysnowman2}\\
&\quad+ \sum_{\substack{\bfs, \z \\ d: |d-nq| > n \tilde \delta_n} }J(\bfs\mid \z, d)P(\z, d)  \label{eq:prehappysnowman3},
\end{align}
since $\frac{|S(\z,d)\cap E(\bfs)|}{|S(\z,d)|} \leq 1.$

Now, observe that we can choose $\tilde \delta_n$ such that \eqref{eq:prehappysnowman2} and~\eqref{eq:prehappysnowman3} are typicality conditions, and more specifically
\begin{equation}
\lim_{n\to\infty} \left(\eqref{eq:prehappysnowman2}+\eqref{eq:prehappysnowman3}\right) 
\leq \lim_{n\to \infty}\Pr \left(\mathbf{Z} \notin \mathcal{T}_{\tilde\delta_n}^{(n)}\right) + \lim_{n\to\infty}\Pr \left( \left|D-nq\right| > n \tilde{\delta}_n \right),
\end{equation}
where this upper bound is equal to $0$. On the other hand, to show that \eqref{eq:prehappysnowman1} (and thus \eqref{eq:OHLAWDY}) also approaches zero as $n$ increases, we make use of the following lemma, whose proof appears in Appendix~\ref{app:numerator_bound}.
\begin{lem}
\label{lemma:numerator_bound}
Let the codewords $\x_{1},\ldots,\x_{M}$ be constructed as above. Then, for any $\z\in \mathcal{T}_{\tilde \delta_{n}}^{(n)}$ and $\bfs\in \F_{2}^{n}$, $n$ sufficiently large, and $d$ such that $|d-nq|\leq \tilde \delta_n n$,
\begin{equation}
\label{equation:num_lemma}
|S(\z,d)\cap E(\bfs)|\leq 2^{n\left(H(q)-H(p)-7.5\delta_n\right)}.
\end{equation}
\end{lem}
Combining Corollary \ref{co:sufficiently_myopic} and Lemma \ref{lemma:numerator_bound}, we see that 
\begin{align}
\frac{|S(\z,d)\cap E(\bfs)|}{|S(\z,d)|} \leq (n+1)^8 2^{-0.5 n\delta_n} \label{eqn:ratio_ub}
\end{align}
which converges to $0$ since $0.5 n\delta_n = -0.5n\tilde \delta_n \log_2 \tilde \delta_n /2 \geq 0.5\sqrt{n} \log_2 n.$ Since \eqref{eqn:ratio_ub} also serves as an upper bound for \eqref{eq:prehappysnowman1}, we are done.

\end{proof}
\section{Proof of Lemma~\ref{lemma:numerator_bound}} \label{app:numerator_bound}

\begin{proof}[Proof of Lemma \ref{lemma:numerator_bound}]
We can upper bound the left-hand side of \eqref{equation:num_lemma} by splitting into different type classes. Namely,
\begin{equation*}
|S(\z,d)\cap E(\bfs)|
\leq \!\!\!\sum_{X,X',S,Z}\!\!\! |\{i: (\x_i,\x_j,\bfs,\z)\in \tau_{XX'SZ}\text{ for some }j\neq i\}|,
\end{equation*}
where the binary random variables $X,X',S,Z$ satisfy
\begin{equation}
\label{equation:rv_constraints}
\mathbb{E}[X]=\mathbb{E}[X']=\frac{1}{2},~~ 
\mathbb{E}[X\oplus Z]=\frac{d}{n},~~
\mathbb{E}[X\oplus S\oplus X']\leq p+\tilde{\delta_{n}}.
\end{equation}
Let $\epsilon_{1,n}>0$ be such that $\epsilon_{1,n}<0.5\delta_{n}$. Applying Lemma \ref{lemma:csiszar_narayan} with $(S,Z)$ in place of $S$, we have
\begin{align}
&|\{i: (\x_i,\x_j,\bfs,\z)\in \tau_{XX'SZ}\text{ for some }j\neq i\}|
    \notag \\ & \quad
     \leq 
2^{n\left|R-I(X;X'SZ)+\left|R-I(X';SZ)\right|^+\right|^+ + n\epsilon_{1,n}}.
    \label{equation:type_bd}
\end{align}
We now want to bound this for random variables satisfying \eqref{equation:rv_constraints}. In particular, consider two cases. First, if $R\leq I(X';SZ)$, then we have that $\left|R-I(X;X'SZ)+\left|R-I(X';SZ)\right|^+\right|^+$ is equal to
\begin{align}
&=\left|R-I(X;X'SZ)\right|^+ \nonumber \\
&\leq \left|R-I(X;S\oplus X')\right|^+\label{equation:suff_small_1a}\\
&\leq  \left|R-(1-H(p+\tilde{\delta_{n}}))\right|^+\label{equation:suff_small_1b}\\
& \leq {|R - 1 +H(p) + \delta_n|^+}\\
&= 0, \label{equation:suff_small_1c}
\end{align}
where \eqref{equation:suff_small_1a} follows from the data processing inequality, \eqref{equation:suff_small_1b} from the last condition of \eqref{equation:rv_constraints}, and \eqref{equation:suff_small_1c} holds by the definition of $R$. Second, if $R>I(X';SZ)$, then we have that $\left|R-I(X;X'SZ)+\left|R-I(X';SZ)\right|^+\right|^+$ is equal to
\begin{align}
&=\left|2R-I(X;X'SZ)-I(X';SZ)\right|^+\nonumber\\
&=\left|2R-I(X;SZ)-I(X;X'|SZ)-I(X';SZ)\right|^+\label{equation:suff_small_2b}
\\&=\left|2R-I(X;SZ)-I(X';XSZ)\right|^+\label{equation:suff_small_2c}
\\&\leq \left|2R-I(X;Z)-I(X';X\oplus S)\right|^+\label{equation:suff_small_2d}
\\&\leq \left|2R-(1-H(d/n))-(1-H(p+\tilde{\delta_{n}}))\right|^+\label{equation:suff_small_2e}\\
&\leq H(d/n)-H(p)-9 \delta_n \label{equation:suff_small_2f} \\
&\leq { H(q)-H(p)- 8 \delta_n},
\end{align}
where \eqref{equation:suff_small_2b} and \eqref{equation:suff_small_2c} follow from the chain rule, \eqref{equation:suff_small_2d} from the data processing inequality, \eqref{equation:suff_small_2e} from the last two conditions of \eqref{equation:rv_constraints}, and \eqref{equation:suff_small_2f} holds for $n$ sufficiently large.
In either case, \eqref{equation:type_bd} is upper bounded by
$2^{n\left(H(q)-H(p)-8\delta_{n}+\epsilon_{1,n}\right)}$.
 Since $\epsilon_{1,n}< 0.5\delta_{n}$ and there are only a polynomial number of types, for sufficiently large $n$, we have shown \eqref{equation:num_lemma}.

\end{proof}

\bibliographystyle{IEEEtran}
\bibliography{refs}
\end{document}